\documentclass[11pt]{article}
\usepackage[utf8]{inputenc}
\usepackage{style}
\usepackage{framed}
\usepackage{nicefrac}
\usepackage[parfill]{parskip}
\usepackage{setspace}
\usepackage{hyperref}
\usepackage{mdwlist}
\usepackage{mdframed}
\usepackage[ruled,linesnumbered]{algorithm2e}
\usepackage[utf8]{inputenc}
\usepackage{style}
\usepackage{mdframed}
\usepackage{mdwlist}
\usepackage{nicefrac}
\usepackage{graphicx}
\usepackage[parfill]{parskip}
\usepackage{setspace}
\usepackage[colorlinks]{hyperref}
\usepackage[ruled,linesnumbered]{algorithm2e}
\usepackage{amsmath}
\usepackage{amssymb}
\usepackage{complexity}
\newcommand{\ug}{\textsc{UniqueGames}}

\newcommand{\Ex}{\E}
\renewcommand{\cP}{{\cal P}}

\allowdisplaybreaks

\newcommand{\defeq}{\overset{\rm def}{=}}
\usepackage{microtype}

\usepackage{thm-restate}

\usepackage{boxedminipage}

\newcommand{\smallsetexpansion}{{\sc SmallSetExpansion}}

\newcommand{\uniquegames}{{\sc Unique Games}}

\renewcommand{\ug}{{\sc Unique Games}}
\newcommand{\mkcsp}{{\sf Max}\mbox{-}k\mbox{-}{\sf CSP}}
\newcommand{\mcspt}{\ensuremath{{\sf Max}\mbox{-}2\mbox{-}{\sf CSP}}}
\newcommand{\cH}{\mathcal{H}}

\renewcommand{\cS}{\mathcal{S}}
\newcommand{\cN}{\mathcal{N}}

\newcommand{\cE}{{\mathcal{E}}}

\newcommand{\wh}[1]{\widehat{#1}}

\newcommand{\non}{\nonumber}
\newcommand{\dksh}{\ensuremath{{\sf D}k{\sf SH}}}
\newcommand{\ldksh}{\ensuremath{{\sf At}\mbox{-}{\sf Most}\mbox{-}{\sf D}k{\sf SH}}}
\newcommand{\dks}{\ensuremath{{\sf D}k{\sf S}}}
\newcommand{\cM}{\mathcal{M}}

\newcommand{\corr}[1]{\underset{#1}{\sim}}

\newcommand{\bgamma}{\bar{\gamma}}
\newcommand{\bomega}{\boldsymbol{\omega}}

\newcommand{\uw}{{\sf uw}}

\renewcommand{\cP}{\mathcal{P}}

\newcommand{\Inf}[2]{{\sf Inf}_{#1}\left[#2\right]}

\newcommand{\mcsp}{{\sc Max\mbox{-}CSP}}
\allowdisplaybreaks

\title{A Characterization of Approximability for Biased CSPs}
\author{ Suprovat Ghoshal \\ University of Michigan \\ suprovat@umich.edu
	\and Euiwoong Lee \\ University of Michigan \\ euiwoong@umich.edu}
\begin{document}

\begin{titlepage}
	\maketitle
	\begin{abstract}
		A $\mu$-biased \mcsp~instance with predicate $\psi:\{0,1\}^r \to \{0,1\}$ is an instance of Constraint Satisfaction Problem (CSP) where the objective is to find a labeling of relative weight at most $\mu$ which satisfies the maximum fraction of constraints. Biased CSPs are versatile and express several well studied problems such as {\sc Densest}-$k$-{\sc Sub(Hyper)graph} and {\sc SmallSetExpansion}.\\ 
		
		In this work, we explore the role played by the bias parameter $\mu$ on the approximability of biased CSPs. We show that the approximability of such CSPs can be characterized (up to loss of factors of arity $r$) using the bias-approximation curve of {\sc Densest}-$k$-{\sc SubHypergraph} (${\sf D}k{\sf SH}$). In particular, this gives a tight characterization of predicates which admit approximation guarantees that are independent of the bias parameter $\mu$. \\
		
		Motivated by the above, we give new approximation and hardness results for $\dksh$. In particular, assuming the {\em Small Set Expansion Hypothesis} (SSEH), we show that $\dksh$ with arity $r$ and $k = \mu n$ is \NP-hard to approximate to a factor of $\Omega(r^3\mu^{r-1}\log(1/\mu))$ for every $r \geq 2$ and $\mu < 2^{-r}$. We also give a $O(\mu^{r-1}\log(1/\mu))$-approximation algorithm for the same setting. Our upper and lower bounds are tight up to constant factors, when the arity $r$ is a constant, and in particular, imply the first tight approximation bounds for the {\sc Densest}-$k$-{\sc Subgraph} problem in the linear bias regime. Furthermore, using the above characterization, our results also imply matching algorithms and hardness for every biased CSP of constant arity.
	\end{abstract}
\end{titlepage}

\tableofcontents

\newpage

\section{Introduction}

Constraint Satisfaction Problems (CSPs) are a class of extensively studied combinatorial optimization problems in theoretical computer science. Typically, an instantiation of an $r$-CSP $\Psi(V,E,[R],\{\Pi_e\}_{e \in E})$ is characterized by an underlying $r$-ary hypergraph $\Psi = (V,E)$ with label set $[R]$, and a constraint $\Pi_e \subset [R]^r$ for every edge $e \in E$. The objective is to find a labeling $\sigma:V \to [R]$ that satisfies the maximum fraction of constraints -- here, the labeling $\sigma$ {\em satisfies} a hyperedge $e= (v_1,\ldots,v_r)$ if $(\sigma(v_1),\sigma(v_2),\ldots,\sigma(v_r)) \in \Pi_e$. The expressive power of CSPs is evident from the long list of fundamental and well-studied combinatorial optimization problems that can be expressed as a CSP: {\sc Max-Cut}~\cite{GW94,OW94}, {\sc Coloring}~\cite{Wigderson82,KMS98}, {\sc Unique Games}~\cite{Khot02a,KKMO07} are all examples of CSPs, each of which has been studied extensively by themselves (see~\cite{MM17CSP} for a comprehensive overview). The tight interplay between CSPs and Probabilistically Checkable Proofs (PCPs) has led to a line of works spanning decades resulting in substantial progress in the theory of approximation of CSPs, eventually leading to landmark results such as tight upper and lower bounds for every CSP assuming the Unique Games Conjecture~\cite{Rag08}.

A well-studied variant of CSPs are CSPs with cardinality constraints i.e., CSPs where there are {\em global} constraints on the relative weight of vertices that can be assigned a particular label. Perhaps one of the simplest instantiations of a CSP with a cardinality constraint is the {\sc Densest}-$k$-{\sc Subgraph} (D$k$S) problem. Here, given an undirected graph $G = (V,E)$ and a parameter $k \in \mathbbm{N}$, the objective is to find a subset of $k$-vertices such that the number edges induced by the subset is maximized. It is easy to see that this is an instantiation of a Boolean CSP of arity $2$ with the underlying graph as $G$, the edge constraints being {\sf AND}, and the global cardinality constraint is that exactly $k$ vertices of the CSP can be assigned the label $1$. Furthermore, this is also a relaxation of the $k$-{\sc Clique} problem, and consequently, it is not surprising that there has been many works which study its approximability~\cite{FS97,BCCFV10,BCGVZ12,Man17}. 

While the unconstrained version of this problem  -- i.e., {\sc Max-AND} with no negations -- is trivially polynomial time solvable, the additional simple cardinality constraint makes the problem significantly harder. In particular, Raghavendra and Steurer~\cite{RS10} showed that assuming the {\em Small Set Expansion Hypothesis} (SSEH), \dks~is \NP-hard to approximate to any constant factor. Furthermore, Manurangsi~\cite{Man17} showed assuming the {\em Exponential Time Hypothesis} (ETH), there are no polynomial time algorithms which can approximate \dks~up to an almost polynomial ratio. A similar phenomena was also observed by Austrin and Stankovic~\cite{AS19} for the setting of cardinality constrained {\sc Max-Cut} as well. Furthermore, the nature of how the approximability of the CSP is affected is also predicate dependent. For instance, while in the case of D$k$S, the cardinality constraint makes it constant factor inapproximable for any constant, in the case of {\sc Max-Cut}, there exists a $0.858$-approximation factor for any $k = \mu n$ with $\mu \in (0,1)$. Hence we are motivated to ask the following question:

\begin{itemize}
	\item[$\triangleright$] {\it Can we characterize predicates which admit approximation factors which are independent of the bias parameter $\mu$?} 
\end{itemize}

In this work, we focus on understanding the above phenomena at a more fine grained level. In particular, we aim to explicitly quantify the role of the bias parameter $\mu:= k/n$ in the approximability of a Boolean CSP with a cardinality constraint. Formally, for any $\mu \in (0,1)$, the $\mu$-biased instance of a Boolean CSP $\Psi(V,E)$ with $r$-ary predicate $\psi:\{0,1\}^r \to \{0,1\}$ (denoted by $\Psi_{(\mu)}{(V,E)}$) is an instantiation where the objective is to find a labeling $\sigma:V \to \{0,1\}$ of relative weight at most $\mu$ which satisfies the maximum fraction of edge constraints in $\Psi$. Furthermore, let $\alpha_{\leq \mu}(\psi)$ -- referred to as the {\em bias-approximation curve} -- denote the optimal approximation factor efficiently\footnote{Here, we say that a $\alpha$-factor approximation is efficiently achievable if the problem of finding an $\alpha$-approximate solution to the biased \mcsp~problem is in \P.} achievable for $\mu$-biased vertex weighted instances on predicate $\psi$. Given this setup, it is natural to ask the following:

\begin{itemize}
	\item[$\triangleright$] {\it Can we give matching upper and lower bounds for $\alpha_{\leq \mu}(\psi)$ for every constant bias $\mu$ and predicate $\psi$?} 
\end{itemize}

The above, despite being a natural question, has only been studied for very specific instantiations of $\psi$ such as {\sc Max}-$k$-{\sc Vertex Cover}~\cite{Man17}, {\sc Max-Cut}~\cite{AS19}. Furthermore, tight lower bounds are known for even fewer settings such as the almost satisfiable regimes of {\sc SmallSetExpansion}, {\sc Max-Bisection}, {\sc BalancedSeparator}~\cite{RST12}, and as such, a finer understanding of $\alpha_{\leq \mu}(\cdot)$ is absent even for natural problems such as \dks.

\subsection{Our Main Results}

In this work, we make substantial progress towards answering the above questions. In order to formally state our results, we need to introduce some additional notation. Given a predicate $\psi:\{0,1\}^r \to \{0,1\}$, let $\psi^{-1}(1)$ be the set of accepting strings for predicate $\psi$.  
Let $\cM_\psi$ denote the set of {\em minimal elements} of $\psi^{-1}(1)$ under the ordering imposed by the containment relationship\footnote{Here the containment relationship refers to the containment relationship induced by interpreting the Boolean strings as indicators of subsets.}. We will think of instances of {\sc Max-CSPs} as vertex weighted, and an instance of $\mu$-biased CSP with predicate $\psi$ is one where the objective is to find a global labeling of the vertices with {\em relative weight}\footnote{Given a labeling $\sigma:V \to \{0,1\}$, its relative weight with respect to vertex weight function $w:V \to \{0,1\}$ is defined as $w(\sigma):= \sum_{i:\sigma(i) = 1} w(i)/w(V)$, where $w(V)$ denotes the total vertex weight.} at most $\mu$ which satisfies the maximum fraction of constraints. 
Furthermore, we say that a predicate $\psi$ is {\em bias dependent} if $\inf_{\mu \in (0,1/2)} \alpha_{\leq \mu}(\psi) = 0$. For any $i \geq 2$, we use $\dksh_i$ to denote the {\sc Densest}-$k$-{\sc SubHypergraph} problem on hypergraphs of arity $i$. Finally, for any $i \in \mathbbm{Z}_{\geq 0}$, we use $\alpha^{\uw}_{(\mu)}(\dksh_i)$ to denote the bias approximation curve of uniformly weighted $\dksh_i$ instance. 

Our first result is the following theorem which completely characterizes predicates which are bias dependent. 

\begin{restatable}{rethm}{predcond}				\label{thm:pred-cond}
	 The following holds assuming SSEH. A predicate $\psi:\{0,1\}^r \to \{0,1\}$ is bias independent if and only if $\cM_{\psi}  \subseteq \cS_{\leq 1}$, where $\cS_{\leq 1}$ is the set of $r$-length strings of Hamming weight at most $1$.
\end{restatable}

As a useful exercise, we instantiate the above theorem for $\psi := {\sf NEQ}$ (i.e, Biased {\sc Max-Cut}) and $\psi := {\sf AND}$ (i.e., \dks). Note that the ${\sf NEQ}^{-1}(1) = \{(0,1),(1,0)\} \subset \cS_{\leq 1}$, where as ${\sf AND}^{-1}(1) = \{(1,1)\} \not\subset \cS_{\leq 1}$, which using Theorem \ref{thm:pred-cond} implies that the former admits a bias independent approximation factor, whereas the latter would be bias dependent. Our next theorem gives an unconditional tight characterization (up to factors of $r$) of the bias-approximation curve of a predicate in terms of $\dksh$.

\begin{restatable}{rethm}{predapprx}				\label{thm:pred-approx}
	For every integer $r \geq 2$ and for any $\mu \in (0,1/2)$, the following holds for any predicate $\psi:\{0,1\}^r \to \{0,1\}$:
	\[
	\alpha_{\leq \mu}(\psi) \asymp_r \min_{\beta \in \cM_\psi} \alpha^\uw_{(\mu)}\left({\sf D}k{\sf SH}_{{\|\beta\|_0}}\right),
	\]
	where $\|\cdot\|_0$ denotes the Hamming weight of a string, and $\asymp_r$ used to denote that the two sides are equal up to multiplicative factors depending on $r$.
\end{restatable}

{\bf Hardness and Approximation for \dksh}. Theorem \ref{thm:pred-approx} directly implies that we can reduce the task of understanding the bias-approximation curve of general boolean predicates to that of \dksh, up to loss of multiplicative factors dependent on $r$. We introduce an additional notation: let $\delta^{(r)}_{\mu|V|}(H)$ denotes the optimal value of $\dksh_r$ with bias $\mu$ on hypergraph $H$ of arity $r$.

Our first result here is the following theorem which gives the first bias dependent hardness for \dksh.

\begin{restatable}{rethm}{dkshhard}				\label{thm:dksh-hard}
	The following holds assuming SSEH for every $r \geq 2$ and $\mu < 2^{-r}$. Given a hypergraph $H = (V,E)$ of arity $r$, it is $\NP$-hard to distinguish between the following two cases:
	\[
	\textnormal{\bf YES Case}: \delta^{(r)}_{\mu|V|}(H) \gtrsim \frac{\mu}{r^3\log(1/\mu)}
	 \ \ \ \ \ \  \textnormal{ and } \ \ \ \ \ \ 
	 \textnormal{\bf NO Case}: \delta^{(r)}_{\mu|V|}(H) \lesssim  \mu^r.
	\]
\end{restatable}

The above theorem implies that $\dksh_r$ with bias parameter $\mu$ is hard to approximate up to factor a of $ O(r^3) \cdot \mu^{r-1}\log(1/\mu)$. We complement the above hardness result with the following theorem which gives bias dependent approximation for $\dksh$.

\begin{restatable}{rethm}{dkshalg}				\label{thm:dksh-alg}
	The following holds for any $\mu \in (0,1)$ and $r \geq 2$. There exists a randomized algorithm which on input a hyperegraph $H = (V,E)$ of arity $r$, runs in time $|V|^{{\rm poly}(1/\mu)}$ and returns a set $S \subset V$ such that $|S| = \mu n$ and $|E_H[S]| \geq C\mu^{r-1}\log(1/\mu) \cdot \delta^{(r)}_{\mu|V|}(H)$.
\end{restatable}  

The upper and lower bounds on the optimal approximation factor from the above theorems are tight up to factor $O(r^3)$, and are therefore tight up to multiplicative constants for constant $r$. In particular, for the setting of $r = 2$ i.e, {\sc Densest}-$k$-{\sc Subgraph}, the above imply the tight approximation bound of $\Theta(\mu\log(1/\mu))$. Finally, Theorems \ref{thm:pred-approx}, \ref{thm:dksh-hard} and \ref{thm:dksh-alg} together imply the following corollary which gives tight bias dependent approximation bounds for every constant $r$.

\begin{corollary}			\label{corr:csp-approx}
	The following holds for any predicate $\psi:\{0,1\}^r \to \{0,1\}$ assuming SSEH.
	\[
	\alpha_{\leq  \mu}(\psi) \asymp_r \min_{\beta \in \cM_{\psi}} \mu^{\|\beta\|_0 - 1} \log(1/\mu).
	\]	
\end{corollary}

\begin{remark}
 	The above results also generalize readily to the setting where the variables are allowed to be negated by applying the above results (Theorems \ref{thm:pred-cond} and \ref{thm:pred-approx}) to each of the $2^r$ predicates obtained by applying the $2^r$ negation patterns to the literals.
 \end{remark}

\begin{remark}\label{rem:smooth}
	We point out that in our setting, we allow algorithms to output solutions with relative weight slightly larger than $\mu$, say $\mu(1 + \eta)$, where $\eta$ is a constant. This additional multiplicative slack is indeed necessary as Theorem \ref{thm:pred-approx} does not hold in the case where algorithms are constrained to output a solution of relative weight at most $\mu$. This is mainly due to the observation that in general, weighted instances of $\dksh$ can be much harder than unweighted instances (even for the same $k$) -- we illustrate this concretely using an example in Appendix \ref{sec:example}. Allowing a constant multiplicative slack in the relative weight enables us to bypass this technical difficulty. Furthermore, we note that our upper bound for the bias approximation curve (i.e, the hardness) also holds for algorithms which are allowed this multiplicative slack -- for details we refer the readers to Lemma \ref{lem:pred-hard}.
\end{remark}

{\bf Application to {\sc Max}-$k$-{\sc CSP}}. Extending our techniques from Theorem \ref{thm:dksh-hard}, we also prove the following new approximation bound for {\sc Max}-$k$-{\sc CSP}s in the large alphabet regime.

\begin{restatable}{rethm}{csphard}				\label{thm:csp-hard}
	The following holds assuming the Unique Games Conjecture, for every $k \geq 2$ and $R \geq 2^k$. Given a {\sc Max}-$k$-{\sc CSP} instance $\Psi(V,E,[R],\{\Pi_e\}_{e \in E})$, it is \NP-hard to distinguish between 
	\[
	\textnormal{\bf YES Case}: {\sf Opt}(\Psi) \geq \frac{C_1}{k^2\log(R)}
	\ \ \ \ \ \  \textnormal{ and } \ \ \ \ \ \ \
	\textnormal{\bf NO Case}: {\sf Opt}(\Psi) \leq \frac{C_2}{R^{(k-1)}}.
	\]
	where $C_1,C_2 > 0$ are absolute positive constants independent of $R$ and $k$.
\end{restatable}

The above implies that {\sc Max}-$k$-{\sc CSP}s on label sets $[R]$ are Unique Games hard to approximate up to a factor of $\Omega(k^2R^{-(k-1)}\log(R))$, this improves on the previous lower bound $\Omega(k^3R^{-(k-1)}\log(R))$ implicit in the work of Khot and Saket~\cite{KS15}\footnote{In particular, Khot and Saket~\cite{KS15} show that any $\alpha$-integrality gap linear programs for $k$-CSPs on label set $[R]$ can be lifted to $\Omega(\alpha k^{3} \log R)$ hardness assuming UGC. Combining this with the known $R^{-(k-1)}$-LP integrality gap derives the $O(k^3 \log R /R^{k-1})$-hardness.}. Furthermore, since \cite{MNT16} gave a $O(R^{-(k-1)}\log(R))$-approximation algorithm for {\sc Max}-$k$-{\sc CSP}s, Theorem \ref{thm:csp-hard} is tight up to factor of $O(k^2)$, and in particular is tight for all constant $k$.

\subsection{Related Works}

{\bf CSPs with Global Constraints}. There have been several works which study specific instances of CSPs with global constraints. Of particular interest is the {\sc Max-Bisection} problem which is {\sc Max-Cut} with a global bisection constraint. The question of whether {\sc Max-Bisection} is strictly harder than {\sc Max-Cut} has been a tantalizingly open question that has been studied by several works \cite{FJ97}\cite{Ye99}\cite{Zwick02}\cite{RT12}, the current best known approximation factor being $0.8776$ by Austrin, Benabbas and Georgiou~\cite{ABG12}. Another well studied problem in the framework is the \smallsetexpansion~problem, due to its connection to the SSEH~\cite{RS10} and its consequences. In particular,  Raghavendra, Steurer and Tetali~\cite{RST10a} gave an algorithm, which when a graph have a set of volume $\delta$ with expansion $\epsilon$, outputs a set of volume at most $O(\delta)$ with expansion at most $O(\sqrt{\epsilon \log (1/\delta)})$, which was later shown to be tight by Raghavendra, Steurer and Tulsiani~\cite{RST12}.  There have been several works which also give frameworks for approximating  general CSPs with global constraints. Guruswami and Sinop~\cite{GS11} gave Lasserre hierarchy based algorithms for the setting when underlying label extended graph has low threshold rank. Raghavendra and Tan~\cite{RT12} also propose a Lasserre hierarchy based framework for general settings. More recently, \cite{BansalSticky20} also study such CSPs using {\sc Sticky Brownian Motion} based rounding algorithms. 

{\bf $\dks$ and $\dksh$.} There is a long line of works which study the complexity of approximating $\dks$. Feige \cite{Feige3Sat} showed constant \dks~is \APX-hard assuming the hardness of refuting random $3$-SAT formulas. Subsequently Khot~\cite{Khot06} also established \APX-hardness assuming no sub-exponential time algorithms exist for {\sf SAT}. Stronger inapproximability results are known under alternative hypotheses. The SSEH of Raghavendra and Steurer~\cite{RS10} immediately implies constant factor inapproximability of $\dks$ where as Manurangsi~\cite{Man17} showed almost polynomial ratio ETH based hardness. There are also results which establish running time lower bounds under alternative hypotheses~\cite{CCKLMNT17},\cite{MRS21}. On the algorithmic front, Feige and Seltser~\cite{FS97} give a $n/k$-approximation algorithm for $\dks$. For $k$-independent bounds, Fiege, Kortsarz and Peleg~\cite{FKP01} gave a $n^{1/3 - \epsilon}$ -approximation algorithm, which was later improved to a $n^{1/4 + \epsilon}$ by~\cite{BCCFV10}. In comparison, there have been relatively fewer works which study the hypergraph variant i.e., \dksh. Assuming the existence of certain one way functions, Applebaum \cite{App13} showed that \dksh~is hard to approximate for hypergraphs on $n$-vertices up to a factor of $n^\epsilon$, for some constant $\epsilon > 0$. The results of \cite{M17} also implies that assuming SSEH, $\dksh$ is inapproximable for any constant factor with large enough arity.

{\bf Max-CSPs on large alphabets.} There is a vast literature which study CSPs on non-boolean alphabets. For arity the setting of arity $2$,  Kindler, Kolla and Trevisan~\cite{KKT15} gave a $\Omega((1/R)\log R)$-approximation algorithm, this matches the Unique Games based hardness from~\cite{KKMO07}. For the setting of larger arities, Makarychev and Makarychev~\cite{MM12} gave a $\Omega(k/R^{k-1})$-approximation algorithm, which is tight for the setting of $k \geq R$ from the work of Chan~\cite{Chan16}. For the setting of small $k \leq R$, the best known lower bound is the Unique Games based $\Omega(R^{-(k-1)}(\log R)^{k/2})$-hardness by Manurangsi, Nakkiran and Trevisan~\cite{MNT16}, for which they also give a $O(R^{-(k-1)}\log R)$-approximation algorithm. The tightest known lower bound is due to the work of \cite{KS15}, whose result along with a known $R^{-k + 1}$-integrality gap for linear programs implies a $\Omega(k^3 R^{-(k-1)} \log R)$-hardness for $2$-CSPs assuming UGC.

\section{Overview: Bias Independence Characterization}

In this section, we briefly describe the challenges towards establishing our results and the techniques used to address them. 

\subsection{Characterization of Bias Independence via Minimal Sets}

Our first step is to understand what makes the optimal approximation factor for a predicate bias dependent. For the purpose of exposition, we shall just focus on the behavior of predicates in the range $\mu \in (0,1/2)$, these ideas presented here will extend naturally to the setting $\mu \in (1/2,1)$ as well. Furthermore, as a warm up, we will first restrict our attention to {\em symmetric} Boolean predicates i.e., predicates whose set of accepting strings is permutation invariant. In particular, one can always express a symmetric predicate $\psi:\{0,1\}^r \to \{0,1\}$ as 
\begin{equation}			\label{eqn:symm-decomp}
\psi := \psi_{i_1} \vee \psi_{i_2} \vee \cdots \vee \psi_{i_t}
\end{equation}
for some $i_1,\ldots,i_t \in \{0,1,\ldots,r\}$ where $\psi_j$ is the predicate which accepts a string if and only if the string has Hamming weight $j$. The above decomposition hints at the following possibility that in order to characterize the bias-approximation curve of $\psi$, it would suffice to characterize the approximation curves of $\psi_{i_1},\ldots,\psi_{i_t}$ (we shall elaborate further on this in the latter part of this section). Hence as a further simplification, we will now restrict our attention on the class of predicates $\{\psi_j\}$.

{\bf Easy Cases $\psi_0,\psi_r$}. To begin with, the cases $i = 0,r$ can be characterized almost immediately. When $i = 0$ i.e., $\psi_i$ is the predicate which accepts if and only if the string is all zeros. Then the all zeros labeling would satisfy all constraints, and hence this trivially yields a $1$-approximation algorithm for $\psi_0$. On the other hand when $i = r$, then $\psi_i$ corresponds to the ${\sf AND}$ predicate, and in particular expresses the {\sc Densest}-$k$-{\sc SubHypergraph} (\dksh) problem. Since for $r = 2$, there exists a $\mu$-approximation algorithm~\cite{FS97} and SSEH impiles that \dks~is constant factor inapproximable~\cite{RS10}, it follows that $\lim_{\mu \to 0} \alpha_{(\mu)}(\psi_r) = 0$. 

{\bf $\psi_i$'s are as hard as \dksh}. As a next step, we can show that for any $i \in \{2,\ldots,r-1\}$, the $i^{th}$ symmetric predicate $\psi_i$ is at least as hard as $\dksh_{i}$, up to loss of some factors in $r$. The key idea underlying this observation is the following: given a $\dksh_i$ instance $\Psi_{(\mu)}(V,E)$ one can construct a $\mu$-biased CSP $\Psi'$ on predicate $\psi_r$, by adding $r - i$ dummy vertices, each with infinite weight, and then including these dummy vertices in every hyperedge $e \in E$. Since any finite weight labeling of $\Psi'$ would be forced to assign $0$'s to the dummy vertices, it follows that the set of edges satisfied by any $\mu$-biased labeling in $\Psi'$ with respect to predicate $\psi_i$ exactly corresponds to the set of hyperedges induced by the set indicated by that labeling in $V$. Given this one-to-one correspondence between labelings\footnote{In the actual reduction, the biases of the labeling can differ up to a multiplicative factor of $r$, but we ignore this issue here to keep the presentation simple.} in $\Psi$ and $\Psi'$ it follows that an $\alpha$-approximation algorithm for approximated biased CSPs with predicate $\psi_r$ yields an $\alpha$-approximation algorithm for $\dksh_i$, and hence $\psi_i$ is at least as bias dependent as ${\sf AND}_i$.

{\bf $\psi_i$'s are as easy as \dksh}. Furthermore, we can also establish a converse to the above by showing that $\psi_i$'s are at least as easy to approximate as $\dksh_i$ (again, with loss of some multiplicative factors in $r$).  Given a $\mu$-biased instance $\Psi'_{(\mu)}(V,E)$ with predicate $\psi_i$, we can naturally construct a hypergraph $H = (V,E')$ on the same vertex set where the set of hyperedges is as follows: for every constraint $e \in E$, and every $i$-sized subset $S \subset e$, we introduce a hyperedge $(e,S)$. It turns out that unlike in the previous case, here we can't establish an exact one-to-one correspondence between labelings in $\Psi$ and $H$. However, we can still establish an approximate version of it. In particular, we can show that (i) the optimal $\mu$-biased value of $H$ (w.r.t. predicate ${\sf AND}_i$) is at least the $\mu$-biased value of $\Psi$ (w.r.t. predicate $\psi_i$) and (ii) given a $\mu$-biased labeling $\sigma$, by sub-sampling, we can construct another $\mu$-biased labeling $\sigma'$ such that the expected fraction of constraints satisfied by $\sigma'$ in $\Psi'$ is at least $2^{-r}$ fraction of hyperedges induced by the sets indicate by $\sigma$ in $H$. This approximate one-to-one correspondence implies that that $\alpha(\mu)$-approximation algorithm for $\dksh_i$ imply a $\Omega_r(\alpha(\mu))$-approximation algorithm for $\mu$-biased CSPs with predicate $\psi_i$. 

The above arguments combined together roughly establishes the following:
\begin{equation}			\label{eqn:bias-symm}
\alpha_{(\mu)}\left({\sf D}k{\sf SH}_i\right) \lesssim_r \alpha_{(\mu)}\left(\psi_i\right) \lesssim_r \alpha_{(\mu)}\left({\sf D}k{\sf SH}_i\right),
\end{equation}
where $\lesssim_r$ hides multiplicative factors in $r$. In particular, \eqref{eqn:bias-symm} completely characterizes the bias dependence of predicate $\{\psi_i\}^r_{i = 1}$. For $i \in \{2,\ldots,r\}$, predicate $\psi_i$ is as hard as (up to factors of $r$) as $\dksh_i$, and therefore are bias dependent. On the other hand $\psi_1$ is at least as easy as biased CSP with the single variable {\sf AND} formulae as predicates, which is exactly solvable in polynomial time, and hence, $\psi_1$ admits a bias independent approximation factor.

{\bf Handling General Symmetric Predicates}. Now recall that a symmetric predicate $\psi:\{0,1\}^r \to \{0,1\}$ can be always expressed as $\psi = \vee_{j \in [t]} \psi_{i_j}$. Clearly, the approximability of $\psi$ is determined by the choice of $i_1,\ldots,i_t$, and in particular, it is natural to suggest that $\psi$ is as easy to approximate as the easiest predicate i.e., $\psi_{i_j}$ with the smallest $i_j$ (which we denote by $i^*$). It turns out that this is indeed the right characterization as we can establish that 
\begin{equation}			\label{eqn:bias-symm1}
\alpha_{(\mu)}\left(\psi\right) \asymp_r \min_{\ell \in \{i_1,\ldots,i_t\}}\alpha_{(\mu)}\left({\sf D}k{\sf SH}_\ell\right),
\end{equation}
where $\asymp_r$ implies that the LHS is within multiplicative factors of $r$ of the RHS. While the hardness of $\psi$ using $\dksh_{i^*}$ again follows by  introducing dummy vertices with infinite weights, establishing the converse -- i.e., $\psi_{i^*}$ is as easy as ${\sf D}k{\sf SH}_{i^*}$ -- requires more work due to the following issue. Given a $\mu$-biased CSP $\Psi_{(\mu)}(V,E)$ on predicate $\psi$, it might be the case that all labelings which assign strings of weight $i^*$ to a significant fraction of edges satisfy negligible fraction of constraints in comparison to the optimal $\mu$-biased labeling. Since our previous argument relied on the existence of labelings which assign strings of weight $i^*$ to a large fraction of constraints, it cannot be used argue good approximation for such instances. This is remedied by ruling out the existence of such instances. In particular, given any labeling $\sigma$, we show that by sub-sampling we can construct another labeling $\sigma'$ which assigns satisfies a significant fraction of edges in comparison to $\sigma$, while assigning them strings of weight exactly $i^*$.

{\bf Handling General Predicates using Minimal Sets}. Now we relax our setting to that of general Boolean predicates $\psi:\{0,1\}^r \to \{0,1\}$. Note that since $\psi$ is not symmetric, it is no longer guaranteed to admit a decomposition of the form \eqref{eqn:symm-decomp}, and therefore it is not clear if one can still characterize the bias approximation curve of $\psi$ using that of $\dksh$. In order to motivate our characterization here, consider the  following notion of partial ordering among predicates. For a pair of Boolean predicate $\psi,\tilde{\psi}:\{0,1\}^r \to \{0,1\}$, we say $\psi \succeq \tilde{\psi}$ if $\psi^{-1}(1) \supseteq \tilde{\psi}^{-1}(1)$, and for every accepting string $\beta \in \psi^{-1}(1)$, there exists $\tilde{\beta} \in \psi^{-1}(1) \cap \tilde{\psi}^{-1}(1)$ such that ${\rm supp}({\beta}) \supseteq {\rm supp}(\tilde{\beta})$. For any such $\psi, \tilde{\psi}$-pair it is not too difficult to show that $\alpha_{(\mu)}({\psi}) \gtrsim_r \alpha_{(\mu)}(\tilde{\psi})$.	

Indeed, given a $\mu$-biased instance $\Psi_{(\mu)}(V,E)$ on predicate $\psi$, let $\sigma:V \to \{0,1\}$ be the $\mu$-biased labeling which achieves the optimal value, say $\gamma$. Now given $\sigma$, consider the following sub-sampling process to construct $\sigma': V \to \{0,1\}$. For every $i \in V$ we do the following independently: if $\sigma(i) = 1$, sample $\sigma'(i)$ uniformly from $\{0,1\}$ otherwise, we set it $0$. Now fix a constraint $e$ satisfied by $\sigma$ and let $\tilde{\beta}_e \in \psi^{-1}(1) \cap \tilde{\psi}^{-1}(1)$ such that ${\rm supp}(\sigma(e)) \supseteq {\rm supp}(\tilde{\beta}_e)$. Since the sub-sampling is independent, with high probability $\sigma'$ is at most $\mu$ biased. Furthermore, 
\[
\Ex_{e \sim E}\Pr_{\sigma'}\Big[\tilde{\psi}\left(\sigma'(e)\right) = 1 \Big] 
\geq \Pr_{e \sim E}\Big[\sigma(e) = 1\Big]\Pr_{\sigma'}\Big[\sigma'(e) = \tilde{\beta}_e\Big] \geq 2^{-r} \gamma,
\]
i.e., $\sigma'$ satisfies at $2^{-r}\gamma$-fraction of constraints $e$ in $\Psi$ by assigning strings from $\tilde{\psi}^{-1}(1)$. Furthermore, since $\psi^{-1}(1) \supseteq \tilde{\psi}^{-1}(1)$, if a labeling satisfies at least $\gamma'$ fraction of edges in $\Psi$ with respect to predicate $\tilde{\psi}$, it also satisfies at least $\gamma'$-fraction of edges with respect to predicate $\psi$. Hence, the  $\alpha_{(\mu)}(\tilde{\psi})$-approximation algorithm for $\tilde{\psi}$ is also a $\Omega(2^{-r}\alpha_{(\mu)}(\tilde{\psi}))$-approximation algorithm for $\psi$. In summary, this establishes that whenever $\psi \succeq \tilde{\psi}$ we have $\alpha_{(\mu)}({\psi}) \gtrsim_r \alpha_{(\mu)}(\tilde{\psi})$.

The above partial ordering and its properties immediately imply that a predicate $\psi$ will be at least as easy as the set of minimal elements dominated by it. A reduction based argument will also show that it is as hard as its minimal elements. Furthermore, a straightforward argument also shows that the minimal elements $\tilde{\psi}$ dominated by $\psi$ are predicates for which the accepting set is a singleton set i.e, they satisfy $|\tilde{\psi}^{-1}(1)| = 1$. It turns out that for such predicates, the arguments used in the setting of symmetric predicates generalize readily. These observations taken together imply the following characterization. Given a predicate $\psi$, we have 
\begin{equation}			\label{eqn:bias-symm3}
\alpha_{(\mu)}\left(\psi\right) \asymp_r \min_{\beta \in \cM_\psi}\alpha_{(\mu)}\left({\sf D}k{\sf SH}_{\|\beta\|_0}\right),
\end{equation}
where $\cM_\psi$ is the set of minimal elements of $\psi^{-1}(1)$. In particular, the above immediately reduces our task to characterizing the bias-approximation curve of $\dksh_i$. The remainder of this section deals with the complexity theoretic aspects of $\dksh$. 

{\bf About Weighted vs. Unweighted settings}. We conclude the first part of the overview by discussing some of added complications that arise while handling vertex weights. While we ignore the difference between weighted and non-weighted settings in the above discussion, the precise statements of our results (Theorem \ref{thm:pred-approx} in particular) actually relate the bias approximation curve of weighted biased CSP problems to that of unweighted D$k$SH. While this does not affect the arguments used to lower bound the bias approximation curve, it presents several subtle challenges in the direction of the upper bound where we use algorithms for {\em unweighted} $\dksh$ as blackboxes for solving {\em weighted} biased CSP instance. In particular, to achieve matching upper and lower bounds, we allow the vertex weights of the CSP to be polynomially large (for e.g., recall that the reduction from $\dksh$ to biased CSPs sets the weights of the dummy vertices to infinite). However, for such weighted instances, techniques for reducing to the unweighted setting don't apply as is and in fact, biased CSPs with arbitrary vertex weights can be strictly harder than unweighted biased CSPs (e.g, see Section \ref{sec:example}).  This issue is addressed by allowing a multiplicative slack in the bias of the labeling -- i.e, where we allow algorithms to output solutions with relative weight at most $\mu(1 + \eta)$ for some constant $\eta$ -- this multiplicative slack is crucially used in relating the approximation curve of weighted $\dksh$ to that of unweighted $\dksh$. Note that this relaxation does not change the lower bound on the approximation curve since our reduction from $\dksh$ to biased CSP holds as is even in this setting. We refer the readers to Section \ref{sec:dksh-red} for more details on this point.

\section{Overview: Hardness of $\dksh$}

While hardness of approximating $\dksh$~for general arities is relatively less explored, there has been substantial work on lower bounds for $\dks$. The strongest known results here are the constant factor inapproximability by Raghavendra and Steurer~\cite{RS10} assuming SSEH, and the almost polynomial ratio hardness by Manurangsi~\cite{Man17} assuming ETH. However the techniques from the above works don't apply to our setting since we seek to explicitly quantify the bias approximation curve. Instead, our approach towards establishing Theorem \ref{thm:dksh-hard} would be to treat $\dksh$~as instances of {\sc Max-AND} subject to a global cardinality constraint. Hence, as is standard, our reduction will use the framework of composing a {\em dictatorship test} with an appropriate {\em outer verifier}. This is a well studied approach that has been used successfully to show often optimal inapproximability bounds for \mcsp s~(see~\cite{Khot-longcode} for an overview of such reductions). 

Informally for a bias $\mu$, the $(c(\mu),s(\mu))$ dictatorship test (for ${\sf AND}$ predicate) in our setting is a distribution $\cD$ over tuples of element i.e., hyperedges $(x_1,\ldots,x_r)$ drawn from some product probability space $\Omega^t$. The distribution naturally defines a weighted hypergraph $\cH = (\Omega^t,\cD,w)$ on the set of vertices $\Omega^t$ and every Boolean function  $f:\Omega^t \to \{0,1\}$ indicates a subset $S_f$ of $\Omega^t$, and therefore can be associated with weight $w(S_f)$. In addition, we seek the following properties from the distribution:
\begin{itemize}
	\item {\bf Completeness}. If $f:\Omega^t \to \{0,1\}$ is a dictator of weight $\mu$, then the $S_f$ induces at least $c(\mu)$-fraction of hyperedges in $\cH$.
	\item {\bf Soundness}. If $f:\Omega^t \to \{0,1\}$ is a function of weight $\mu$ such that $S_f$ induces at least $s(\mu)$-fraction of hyperedges in $\cH$, then $f$ has at least one influential coordinate.
\end{itemize}  
Here the notion of dictators and influential coordinates are the natural analogues of their counterparts in a long code test. The above is typically the key component in dictatorship test reductions, where it is well understood that a $(c,s)$-dictatorship test for a predicate $\psi$ almost immediately leads to $(s/c)$-hardness of approximation (for the unconstrained \mcsp) by composing it with a suitable outer verifier such as {\sc LabelCover} or {\sc UniqueGames}. Therefore, the obvious first challenge here is to design a family of bias dependent dictatorship tests with the right completeness soundness tradeoff. However, unlike the standard unconstrained setting, the composition step requires stronger properties from the outer verifier. In particular, it needs to ensure that a global bias constraint on the set of averaged long code tables $\Ex_{v \sim V} \left(\E_x g_v(x) \right)= \mu$ translates to local bias control i.e., $\Ex_x g_v(x) \approx \mu$ for most $v \in V$ (we shall discuss this issue formally later in the Section \ref{sec:outer}). In the remainder of this section, we discuss the design aspects of the dictatorship test and the choice of outer verifier.

\subsection{Choice of Dictatorship Test}				\label{sec:dict}	

A somewhat loose restatement of our completeness and soundness properties from above would be that we require a distribution over a space $\Omega^t$ with most edges incident that are incident on the set indicated by a dictator function stay inside the set, or equivalently, dictator cuts have small edge expansion. In fact, objects with this property, namely the {\em noisy hypercube} and its variants, have been successfully used to show optimal hardness of several problems such {\sc MaxCut}, {\sc UniqueGames}~\cite{KKMO07}, {\sc SmallSetExpansion}~\cite{RST12}. Due to the additional bias constraint, this motivates us to study the $\mu$-biased $(1 - \rho)$- noisy hypercube $\cH_{\mu,\rho}$. We describe the dictatorship test as a distribution over hyperedges for $\cH_{\mu,\rho}$ in Figure \ref{fig:noise-test} below.

\begin{figure}[ht!]
	\begin{mdframed}
		\vspace{3pt}
		\begin{center}
			\underline{\bf Hyperedge Distribution on $\cH_{\mu,\rho}$}
		\end{center}
		\vspace{3pt}
		{\bf Distribution}. 
		\begin{enumerate}
			\item Sample $x \sim \{0,1\}^t_\mu$.
			\item Sample independent $\rho$-correlated copies $x_1,\ldots,x_r \underset{\rho}{\sim} x$.
			\item Output hyperedge $(x_1,\ldots,x_r)$.
		\end{enumerate}
	\end{mdframed}
\caption{Biased Noisy Hypercube Test}
\label{fig:noise-test}
\end{figure} 

Here $\{0,1\}^t_\mu$ denotes the probability space where each bit $i \in [t]$ is independently sampled from the Bernoulli distribution with bias $\mu$. Furthermore, fixing a $x \in \{0,1\}^t$, a $\rho$-correlated copy $x'$ of $x$ is sampled by setting each bit $x'(i)$ to $x(i)$ with probability $\rho$ and re-sampling $x'(i) \sim \{0,1\}_\mu$ with probability $1 - \rho$. The completeness and soundness of the above test can be analyzed using standard Fourier analytic techniques. A useful first observation is that the fraction of hyperedges induced by a set indicated by a function $f:\{0,1\}^t \to \{0,1\}$ can be expressed as 
\[
w(E[S_f]) = \Ex_{x_1,\ldots,x_r}\left[\prod_{i \in [r]} f(x_i)\right].
\]
For completeness, we see that when $f$ is the dictator function $f(x) = x(1)$, then
\[
\Ex_{x_1,\ldots,x_r}\left[\prod_{i \in [r]} f(x_i)\right] 
\geq \Pr_x\big[x(1) = 1\big] \Pr_{\left(x_i\right)^r_{i = 1}| x} \Big[\forall i \in [r], x_i(1) = x(1)\Big] \geq \mu\rho^r
\]
and $w(S_f) = \Ex_x\left[x(1)\right] = \mu$ i.e., $f$ satisfies the weight constraint. On other hand, for analyzing the soundness guarantee, fix a function $f$ having no influential coordinates and $w(S_f) = \mu$, then by a combination of {\em Invariance Principle} and {\em Gaussian Stability bounds} (for e.g, see Theorem 2.10~\cite{KS15}), it can be shown that  
\begin{equation}			\label{eqn:sound-val}
\Ex_{x_1,\ldots,x_r}\left[\prod_{i \in [r]} f(x_i)\right] \leq  2 \mu^r
\end{equation}
when $\rho := 1/\sqrt{r^2\log(1/\mu)}$. Combining the above, we get a completeness-soundness ratio of $\mu^{r-1}(\log (1/\mu))^{r/2}$, which is off by a factor of $\log(1/\mu))^{r/2 - 1}$ from the desired ratio. The issue here is that the completeness value of the test has a $\rho^r$ multiplier due to the independent $\rho$-correlated re-sampling. In particular, conditioned on $x(1) = 1$, the completeness pays an extra multiplier of $\rho^r$ since for every $i \in [t]$, $x_i(1)$ can be chosen to resampled independently with probability $1 - \rho$.

To fix the above, we allow the noise pattern of variables $x_1,\ldots,x_r$ to be correlated instead of being fully independent. Formally, observe that we can reinterpret the original $\rho$-correlated sampling along a coordinate $j \in [t]$ in the following way.
\begin{itemize}
	\item[(i)] Sample $\theta_1(j),\theta_2(j),\ldots,\theta_r(j) \sim \{0,1\}^t_\rho$.
	\item[(ii)] For every $i \in [r]$, do the following: if $\theta_{i}(j) = 1$, set $x_i(j) = x(j)$ otherwise sample $x_i(j) \sim \{0,1\}_\mu$ independently.
\end{itemize}
The above results in a distribution where each pair of $x_i,x_j$ variables are $\rho^2$-correlated. Now, the crucial observation here is that, as is the case with noise stability type arguments, the soundness analysis for the above test distribution just relies on the second moment structure, and in particular, just uses the fact that the $(x_i,x_j)$ variables are pairwise $\rho^2$-correlated. This is due to the folklore observation that for any distribution on $x_1,\ldots,x_r$ that is pairwise $\rho^2$-correlated, using techniques from \cite{Mos10} one can show 
\[
\Ex_{x_1,\ldots,x_r} \left[\prod_{i \in [r]} f(x_i)\right] \approx \prod_{i \in [r]} \Ex_{x_i}\big[f(x_i)\big].
\]
Furthermore, it is well known (for e.g, \cite{KS15}) that this weaker condition on the test distribution can be realized with more correlated noise patterns. In particular, we can consider the following alternative distribution: 
\[
\textnormal{\it W.p. $\rho^2$, set $\theta_1(j),\ldots,\theta_r(j)$ to $1$, othewise set $\theta_1(j),\ldots,\theta_r(j)$ to $0$}
\]
In other words, for any $j \in [t]$, with probability $\rho^2$, we set all $x_i(j)$ variables to $x(j)$, otherwise we resample all variables independently. It is easy to see that the above again results in a distribution where each $(x_i,x_j)$ variables pair is $\rho^2$-correlated. However, note that under the new distribution, the probability of realizing a all-ones assignment along any coordinate $j \in [t]$ is at least
\[
\Pr\left[x(j) = 1\right] \cdot \Pr\left[\theta(j) = 1 \right] = \rho^2 \mu,
\]
which improves on $\mu\rho^r$ from the previous test distribution -- this is the key component towards deriving the intended completeness-soundness ratio. We conclude our discussion by giving a brief sketch of the completeness and soundness analysis of the test with respect to the new distribution. Call the new distribution over $r$-tuples $\cD^*$. For the completeness, for a dictator function $f(x) = x(1)$, we proceed as before and get
\begin{align*}
\Ex_{(x_1,\ldots,x_r) \sim \cD^*}\left[\prod_{i \in [r]} f(x_i)\right] 
&\geq \Pr_{x \sim \{0,1\}^t_\mu}\Big[x(1) = 1\Big] \Pr_{\left(x_i\right)^r_{i = 1}| x} \Big[\forall i \in [r], \ x_i(1) = x(1)\Big] \\
&\geq \mu \Pr_{\theta_1 \sim \cD_{k,\rho}} \Big[\forall i \in [r], \ \theta_i(1) = 1\Big] \\
& = \mu\rho^2.
\end{align*}
Our soundness analysis employs the noise stability analysis from Khot and Saket~\cite{KS15}. Consider a function $f:\{0,1\}^t \to \{0,1\}$ having no influential coordinates satisfying $\Ex_{x}\left[f(x)\right] = \mu$. The first step is to observe that since under the test distribution, the variables $x_1,\ldots,x_r$ are pairwise $\rho^2$-correlated, using a multidimensional version of {\em Borell's Isoperimetric Inequality} (Theorem \ref{thm:ks-stab1}), one can show that
\[
\Ex_{x_1,\ldots,x_r}\left[\prod_{i \in [r]} f(x_i)\right] \leq \Gamma^{(r)}_{\rho^2}(\mu)
\]
where $\Gamma^{(r)}_{\rho^2}(\mu)$ is the iterated $r$-ary Gaussian stability of a halfspace with volume $\mu$ (see \eqref{eqn:gamma-def} for a formal definition). Furthermore, when $\rho^2 \leq O(1/(r^2\log(1/\mu)))$, \cite{KS15} shows that 
\[
\Gamma^{(r)}_{\rho^2}(\mu) \lesssim \mu^r
\]
which concludes the soundness analysis.

\subsection{Choice of Outer Verifier}				\label{sec:outer}

Given the above dictatorship test, we now proceed to discuss the composition step. Typically, for CSPs without global constraints, a dictatorship test for a predicate can be plugged in almost immediately into \ug~(or often even {\sc LabelCover}), and result in hardness matching the completeness soundness ratio. However, that technique fails to work for CSPs with global cardinality constraints since this does not provide {\em local bias control}. Formally, the composition step with the above dictatorship test will introduce a long code table $f_v : \Omega^t \to \{0,1\}$ for every vertex $v$ of the outer verifier CSP (say \uniquegames), denoted by $\Psi$. Then, it embeds the dictatorship test with the outer verifier in such a way that the overall reduction can be thought of as the following two step process.
\begin{itemize}
	\item Sample a vertex $v \in \Psi$. Let $g_v$ denote the weighted average of the long codes of the neighbors of $v$ in $\Psi$.
	\item Test $g_v$ on the distribution $(x_1,\ldots,x_t)$ i.e, accept if and only if 
	\[ 
	g_v(x_1) = \cdots = g_v(x_t) = 1
	\]
\end{itemize} 
The key idea used in the above setup is that if $\Psi$ admits a labeling $\sigma$ which satisfies most edges, then the set dictator assignment $f_v = x_{\sigma(v)}$ induces a large fraction of edges. This is because since most edges are consistent with the labeling $\sigma$, this translates to the effect that even the averaged function $g_v$ still behaves like $f_v$ which is a dictator, and hence, the test accepts with probability close to the completeness of the distribution. On the other hand, if the optimal value of $\Psi$ is small, then for any fixed labeling, most edges will be inconsistent, and therefore most averaged functions don't have influential coordinates. Now suppose in addition, we could guarantee that for most averaged function $g_v$, we have $\Ex[g_v] \approx \mu$, we can use the soundness guarantee of the distribution to argue that the test accepts with probability at most the soundness value of the test.

However, note that the composition step as is can only guarantee that the expected bias of a long code $\{f_v\}$ for a randomly chosen vertex $v$ is $\mu$, and as such this does not imply the above concentration guarantee that is needed to argue soundness. In particular, in the context of the reduction, this can allow cheating assignments where the adversary can set a subset of long codes to be the constant all ones functions and the remaining to be all zeros, which can cause the above analysis to fail. Therefore, in order to ensure that even under such assignments, the biases of the averaged long codes concentrate around the global bias $\mu$, the averaging operator should have good mixing properties, or equivalently, have large spectral gap. This requires the use of non-standard outer verifiers such as the Quasirandom PCP~\cite{Khot06} or \smallsetexpansion~\cite{RS10,RST12}. In particular, we shall use the SSE based framework introduced in \cite{RST12}. 

The key component in the \cite{RST12}~framework is a family of noise operators $\{M_z\}$, referred to as {\em Noise Operators with Leakage}. Formally, given a string $z \in \{\bot,\top\}^t$ the corresponding noise operator $M_z$ on the space $\Omega^t$ is defined as follows. For any $\omega \in \Omega^t$, one can sample $\omega' \sim M_z(\omega)$ using the following process. For every $i \in [t]$, do the following independently: if $z(i) = \top$, set $\omega'(i) = \omega(i)$ otherwise re-sample $\omega'(i) \sim \Omega$. In particular, for a random draw of $\{\bot,\top\}^t_\beta$, the noise operator $M_z$ behaves like the standard noise operator $T_\beta$ on the space $L_2(\Omega)$. By incorporating the above noise operator in the averaging step, one can guarantee that the spectral gap of the averaging operator is at least $1 - \beta$, thus guaranteeing the aforementioned concentration on the biases. We describe our overall reduction in Figure \ref{fig:noise-test1}.

\begin{figure}[ht!]
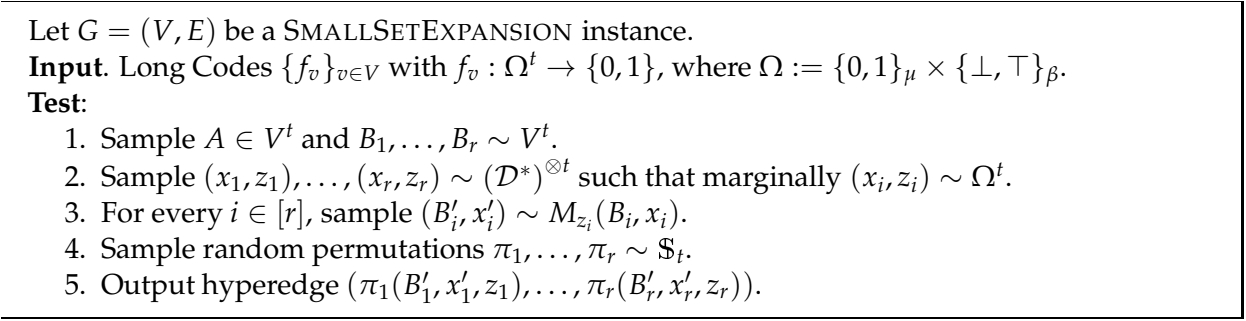

	\begin{mdframed}
		Let $G = (V,E)$ be a \smallsetexpansion~instance.\\
		{\bf Input}. Long Codes $\{f_v\}_{v \in V}$ with $f_v:\Omega^t \to \{0,1\}$, where $\Omega := \{0,1\}_\mu \times \{\bot,\top\}_\beta$. \\
		{\bf Test}: 
		\begin{enumerate}
			\item Sample $A \in V^t$ and $B_1,\ldots,B_r \sim V^t$.
			\item Sample $(x_1,z_1),\ldots,(x_r,z_r) \sim \left(\cD^*\right)^{\otimes t}$ such that marginally $(x_i,z_i) \sim \Omega^t$.
			\item For every $i \in [r]$, sample $(B'_i,x'_i) \sim M_{z_i}(B_i,x_i)$.
			\item Sample random permutations $\pi_1,\ldots,\pi_r \sim \mathbbm{S}_t$.
			\item Output hyperedge $\left(\pi_1(B'_1,x'_1,z_1),\ldots,\pi_r(B'_r,x'_r,z_r)\right)$.
		\end{enumerate}
	\end{mdframed}
	\caption{Reduction from SSE (Informal)}
	\label{fig:noise-test1}
\end{figure} 

The analysis of the above reduction combines the arguments from Section \ref{sec:dict} with techniques from \cite{RST12}. A key difference is that since our test distribution works in the almost uncorrelated regime, we need a more careful completeness analysis i.e,. the reduction from \cite{RST12} works in the setting $\rho \to 1$ whereas our tests are based in the setting $\rho \to 0$. Another additional component in our proof is that is that we strengthen $k$-ary noise stability estimate for functions over the $[R]$-ary hypercube (i.e, \eqref{eqn:sound-val}) with small low degree influences to the setting of arbitrary probability spaces.

\subsection{Approximation algorithm for $\dksh$}

Our approximation algorithm for $\dksh$ is similar in spirit to the $O(R^{k-1}\log R)$-approximation algorithm for $\mkcsp$'s from \cite{MNT16}. Following their approach, the overall reduction consists of two components:
\begin{itemize}
	\item A $O(\mu \log (1/\mu))$-approximation algorithm for $\dks$.
	\item A reduction from $\dksh_r$ to $\dks$ which shows that $\alpha_{(\mu)}(\dksh_r) \geq \mu^{r-2} \alpha_{(\mu)}(\dks)$.
\end{itemize} 
For the first point, we use a $O(\mu \log(1/\mu))$-approximation algorithm that can be obtained by using a reduction (\cite{CHK11}) from $\dks$ to $\mcspt$ on label set $\mu$ and then combining it with the $O(\log R/R)$-approximation algorithm for $\mcspt$ from \cite{KKT15}. The second point is based on the reduction from $\mkcsp_R$ to $\mcspt$ in \cite{MNT16} -- in our setting, this translates to reducing a $\dksh$ instance $H$ into the $\dks$ instance $G$ by adding the clique-expansion of every edge in $H$ to $G$. Combining these ideas immediately yields the $O(\mu^{r-1 \log(1/\mu)})$-approximation algorithm for $\dksh$.

\part{Bias Dependence Characterization}
\section{Preliminaries}

We introduce some notation and definitions that will be used in this part. Throughout we will be dealing with Boolean predicates $\psi:\{0,1\}^r \to \{0,1\}$. For any $i \in [r]$ we shall use $\psi_i$ to denote the $r$-ary predicate supported on strings of Hamming weight exactly $i$, and we shall use $\cS_i$ to denote the set of $r$-length Boolean strings with Hamming weight exactly $i$. Furthermore, for any $i \in \mathbbm{Z}_{\geq 1}$, we use ${\bf 1}_i$ and ${\bf 0}_i$ to denote the all ones and all zeros vectors of length $i$. We will drop the indexing with $i$ and denote ${\bf 1} = {\bf 1}_{i}$ and ${\bf 0} = {\bf 0}_i$ when the length of the string is clear from context. 

{\bf Labelings}. For any set $V$, a Boolean labeling is a mapping $\sigma:V \to \{0,1\}$. Given a weight function $w:V \to \mathbbm{R}_{\geq 0}$, the weight of the labeling $\sigma$ is denoted by $w(\sigma) = \sum_{i \in \sigma^{-1}(1)} w(i)$, and the relative weight is defined as $w(\sigma)/w(V)$, where $w(V) = \sum_{i \in V} w(i)$ is the total weight. Note that the relative weights of the vertices define a probability distribution on $V$, and we use  $i \sim w$ to denote a draw of a vertex $i$ from $V$ using this distribution. We use ${\rm supp}(\sigma)$ to denote the set of non-zero indices in $\sigma$ and for any ordered set $e = (i_1,\ldots,i_t) \subseteq V$, we use $\sigma(e)$ to denote the vector obtained by applying $\sigma$ coordinate-wise i.e., 
\[
\sigma(e) = \Big(\sigma(i_1),\sigma(i_2),\ldots,\sigma(i_t)\Big).
\]
Finally, for a pair of labelings $\sigma_1: V_1 \to \{0,1\}$ and $\sigma_2:V_{2} \to \{0,1\}$ defined on disjoint sets $V_1$ and $V_2$, we define the concatenated labeling $\sigma: = \sigma_1 \circ \sigma_2$ to be the following labeling on $V_1 \sqcup V_2$:
\[
\sigma(i) = 
\begin{cases}
	\sigma_1(i) & \mbox{ if } i \in V_1, \\
	\sigma_2(i) & \mbox{ if } i \in V_2.
\end{cases}
\] 
   
\subsection{Biased CSPs}

An instance of a Biased CSP $\Psi = \Psi(V,E,w,\mu,\psi)$ is a {\sc Max-CSP} instance with predicate $\psi$ with the constraint hypergraph $(V,E)$ and nonnegative vertex weight function $w:V \to \mathbbm{R}_{\geq 0}$. We will always identify the vertex set $V$ as $[n]$, and treat the hyperedges in $E$ as ordered with respect to the natural total ordering on $[n]$. The objective of $\Psi$ is to compute the optimal labeling $\sigma:V \to \{0,1\}$ corresponding to 
\begin{eqnarray*}
	\textnormal{Maximize} 		& 	\Pr_{e \sim E}\left[\psi\left(\sigma(e)\right) = 1 \right] \\
	\textnormal{Subject to}		& 	\Ex_{i \sim w}\left[\sigma(i) \right]  \leq \mu,		   
\end{eqnarray*} 
We use ${\sf Val}_{\leq \mu}(\Psi)$ to denote the optimal value achievable. In addition, for a specific labeling $\sigma:V \to \{0,1\}$ we use ${\sf Val}_\sigma(\Psi)$ to denote the value achieved by the labeling $\sigma$ i.e,
\[
{\sf Val}_{\leq \mu}(\Psi) := \Pr_{e \sim E} \left[\psi\left(\sigma(e)\right) = 1\right],
\]
where $e \sim E$ is used to denote a uniformly random draw of an edge from $E$.

\subsection{Densest SubHypergraph Problems}

We formally define the densest subhypergraph problems used in our setup.

\begin{definition}[{\sc Densest}-$k$-{\sc SubHypergraph} ($\dksh$)] 
A $\dksh$ instance is identified by a unweighted $r$-uniform hypergraph $H = (V,E)$. Then the objective is to find a subset $S \subset V$ of size $k$ that maximizes the fraction of edges induced on $S$ i.e., the goal is to compute
\begin{eqnarray*}
	\textnormal{Maximize} 		& 	\Pr_{e \sim E}\left[\sigma(e) = {\bf 1}\right] \\
	\textnormal{Subject to}		& 	\Ex_{i \sim V}\left[\sigma(i) \right]  = \mu,		   
\end{eqnarray*} 
where $\mu := k/|V|$. We will use $\alpha^\uw_{(\mu)}\left(\dksh_r\right)$ to denote the best approximation factor achievable in polynomial time for $r$-uniform unweighted $\dksh_r$ instances with $k = \mu |V|$. Furthermore, given a $\dksh_r$ instance $H$, we use the notation ${\sf Val}_{(\mu)}(H)$ to denote the optimal value achievable for $H$ with $k = \mu |V|$. For a labeling $\sigma:V \to \{0,1\}$, we use ${\sf Val}_{\sigma}(H)$ to define the value of the labeling on $H$ as defined above i.e.,
\[
{\sf Val}_\sigma(H) := \Pr_{e \sim E} \left[\sigma(e) = {\bf 1}_r\right].
\]
\end{definition}

\begin{remark}
	Note that we use the notation $\alpha^{\uw}_{(\mu)}$  instead of $\alpha^\uw_{\leq \mu}$ notation for $\dksh$ to emphasize that $\dksh$ uses the constraint $k = \mu|V|$ instead of $k \leq \mu|V|$, although the two constraints are equivalent when the instance is uniformly weighted.
\end{remark}

We shall also need the ``at most'' + ``weighted'' variant of the above, which we denote using $\ldksh$.

\begin{definition}[$\ldksh$]
An $\ldksh$ instance is identified by a {\bf vertex weighted} hypergraph $H = (V,E,w,\mu)$ and a bias parameter $\mu \in (0,1)$. Then the objective is to find a subset $S \subset V$ of relative weight at most $\mu$ that maximizes the fraction of edges induced on $S$ i.e., the goal is to compute:
\begin{eqnarray*}
	\textnormal{Maximize} 		& 	\Pr_{e \sim E}\left[\sigma(e) = {\bf 1}\right] \\
	\textnormal{Subject to}		& 	\Ex_{i \sim w}\left[ \sigma(i)\right] \leq \mu.		   
\end{eqnarray*} 
We use the notation ${\sf Val}_{\leq \mu}(H)$ to denote the optimal value achievable in the objective for $H$ using sets of relative weight at most $\mu$.
\end{definition}

We shall need the following folklore lemma which relates the bias approximation curve for $\dksh$ for different biases.

\begin{restatable}{lem}{comp}			\label{lem:comp}
	For any $\ell \in \mathbbm{N}$, $r \in \mathbbm{N}_{\geq 2}$ and $\mu \in (0,1/\ell)$ we have
	\[
	\alpha^\uw_{(\ell \mu)}\left(\dksh_r\right) \asymp_{\ell,r} \alpha^\uw_{(\mu)}\left(\dksh_r\right),
	\]
	where $\asymp_{\ell,r}$ notation implies that both sides are equivalent up to multiplicative factors depending only on $\ell$ and $r$
\end{restatable}
\begin{proof}
	Clearly, it suffices to show that the LHS and RHS quantities satisfy $\gtrsim_{\ell,r}$ and $\lesssim_{\ell,r}$ relationships with respect to each other.
	
	{\bf The $\gtrsim_{\ell,r}$ direction}. We begin by showing that $\alpha^\uw_{(\ell \mu)}(\dksh_r) \geq \ell^{-r} \cdot \alpha^\uw_{(\mu)} (\dksh_r)$. Let $H = (V,E)$ be a $\dksh_r$ instance with bias constraint $\ell \mu$. Let $S^* \subset V$ be the optimal $\ell \mu|V|$-sized subset achieving $\delta_{\ell \mu|V|}(H)$. Now consider a random $\mu |V|$ sized subset $S'$ of $S^*$. It follows that 
	\[
	\Ex_{S'}\Big[|E_H[S']|\Big] \geq \ell^{-r} |E_H[S^*]| = \ell^{-r} \delta_{\ell \mu |V|} (H)
	\]
	which in turn implies that 
	\[
	\delta_{\mu|V|}(H) \geq \ell^{-r} \delta_{\ell \mu|V|} (H).
	\]
	The above immediately suggests the following algorithm: given $H$, run the $\alpha^\uw_{(\mu)}(\dksh_r)$ algorithm on $H$ to output a $\mu|V|$ sized subset $\hat{S}$. Then arbitrarily add $\mu(\ell - 1)|V|$ additional vertices to make it a set $S$ of size $\ell \mu |V|$. The above chain arguments imply that 
	\begin{align*}
		|E_H[S]| \geq |E_H[S']| \geq \alpha^\uw_{(\mu)}\left(\dksh_r\right) \delta_{\mu |V|} (H) \geq \ell^{-r} \alpha^\uw_{(\mu)}\left(\dksh_r\right) \cdot \delta_{\ell \mu|V|}(H).
	\end{align*} 
	Since the above holds for any $\dksh_r$ instance, this establishes the $\gtrsim_{\ell,r}$ direction.
	
	{\bf The $\lesssim_{\ell,r}$ direction}. This direction again uses similar arguments. Given a $\dksh_r$ instance with bias constraint $\mu$, we do the following:
	\begin{itemize}
		\item Run the $\alpha^\uw_{(\ell \mu)}(\dksh_r)$ approximation algorithm on $H$ to output a set $S'$ of size $\ell \mu |V|$ which satisfies 
		\[
		|E_H[S']| \geq \alpha^\uw_{ (\ell \mu)} \left(\dksh_r\right) \cdot \delta_{\ell \mu |V|}(H) \geq \alpha^\uw_{(\ell \mu)}(\dksh_r) \cdot \delta_{\mu |V|} (H).
		\]
		\item Sub-sample a random $\mu |V|$ sized subset $S$ of $S'$ and output $S$. It is easy to argue that 
		\[
		\Ex_S\left[|E_H[S]|\right] \geq \ell^{-r} |E_H[S']| \geq \ell^{-r} \alpha^\uw_{(\ell \mu)}(\dksh_r) \cdot \delta_{\mu |V|} (H).
		\]
	\end{itemize}
	Since the above again holds for any $\dksh_r$ instance $H$, this concludes the proof of the $\lesssim_{\ell,r}$-direction.
\end{proof}

\section{Bias Independence}

In this part, we give a clean characterization of conditions under which the approximability of the predicate is bias dependent. Recall that given a predicate $\psi:\{0,1\}^r \to \{0,1\}$, we will use $\alpha_{\leq \mu}(\psi)$ to denote the optimal $\mu$-biased approximation factor achievable using polynomial time algorithms. We formally define bias independence as follows.

\begin{definition}[Bias Independence]
	A predicate $\psi:\{0,1\}^r \to \{0,1\}$ is bias independent if and only if there exists a constant $c_{\psi} \in (0,1)$ such that 
	\[
	\inf_{\mu \in (0,1/2) } \alpha_{\leq \mu}(\psi) = c_\psi.
	\]
\end{definition}
In other words, a predicate is bias independent, if there exists a polynomial time constant factor approximation algorithm for solving \mcsp~instances on the predicate subject to the constraint $\Ex_{v \sim V} [\sigma(v) ]\leq \mu$ for any constant $\mu \in (0,1/2)$. We proceed towards formalizing the condition of interest. Given a predicate $\psi:\{0,1\}^r \to \{0,1\}$, let $\psi^{-1}(1)$ denote the set of accepting strings of the predicate. Alternatively, thinking of $\psi^{-1}(1)$ as a family of subsets equipped with a partial ordering under the {\em containment} relationship, we can define a set of minimal elements (w.r.t containment relationship), denoted by $\cM_\psi$, for that predicate. Our first main result is the following theorem which gives a tight characterization of bias dependence for predicates.

\predcond*

In addition to the above, we show the following theorem which shows that the approximability of bias dependent predicates is tightly connected to that of \dksh, up to factors of $r$.

\predapprx*

In the rest of the section we shall prove Theorems \ref{thm:pred-cond} and Theorem \ref{thm:pred-approx} using various intermediate results that relate the bias approximation curve of a predicate $\psi$ to that of $\dksh$. 

\subsection{Hardness Using $\dksh$}

We show the following hardness result.
\begin{lemma}					\label{lem:pred-hard}
	The following holds for every $\mu \in (0,1/2)$. Let $\psi:\{0,1\}^r \to \{0,1\}$ be a $r$-ary predicate such that there exists a minimal element $\beta \in \cM_{\psi}$ which does not belong to $\cS_0 \cup \cS_1$. Then
	\[
	\alpha_{\leq \mu'}(\psi) \leq \alpha^{\uw}_{(\mu)}(\dksh_{i^*})
	\]
	where $i^* = \|\beta\|_0$ and $\mu' = \mu/(r - i^* + 1)$.
\end{lemma}
\begin{proof}
	Without loss of generality, let the minimal element in $\psi^{-1}(1)$ be of the form $\beta = \sum_{j = 1}^{i^*} e_j $ i.e, it is the vector consisting of $i^*$-ones followed by zeros. We will establish the lemma by reducing $\dksh$ instances with arity $i^*$ to biased CSP instances on predicate $\psi$. Our reduction is described in Figure \ref{fig:ref-1}.

	\begin{figure}[ht!]
		\begin{mdframed}
			Given  $\dksh$ instance $H = (V,E)$ of arity $i^*$, with $k = \mu|V|$. From $H$, we construct a new vertex weighted $\mu'$-biased instance $\Psi'(V',E',w,\mu',\psi)$ with predicate $\psi$ as follows.\\
				
			{\bf Vertex Set}. The vertex set is $V' = V \cup \{a_j\}^{r}_{j = i^* + 1}$, where $a_{i^* + 1},\ldots,a_r$ are dummy vertices. \\
			
			{\bf Constraint Set}. The constraint set $E'$ is as follows. For every hyperedge $e \in E$, introduce the ordered constraint on $e' :=  e \circ (a_j)^{r}_{j = i^* + 1}$. \\
			
			{\bf Vertex Weights}. Furthermore, we define the vertex weights $w: V' \to \mathbbm{R}_{\geq 0}$ as follows. For every $v \in V$, we set $w(v) = 1$, and for every $a_j \in V' \setminus V$, we set $w(a_j) = n$ where $n = |V|$. Finally the objective is to maximize the fraction of constraints satisfied in $\Psi'$ over assignments with relative weight at most $\mu' = \mu/(r- i^* + 1)$. This completes the description of  $\Psi'$. \\
			
			{\bf Algorithm}. Let $\cA$ denote the $\alpha^* := \alpha_{\leq \mu'}(\psi)$-approximation algorithm for $\mu'$ biased CSPs with predicate $\psi$. Now consider the following procedure for solving $H$.

			\begin{enumerate}
				\item Construct $\Psi' = (V',E',w,\mu',\psi)$ from $H$ as above.
				\item Compute a labeling $\sigma'$ of relative weight at most $\mu'$ by running algorithm $\cA$ on predicate $\Psi'$.		\label{step:appx}
				\item Let $\sigma:= \sigma'|_V$. where $\sigma'|_{V}$ is the restriction of $\sigma'$ to $V$.
				\item If $w(\sigma) < \mu |V|$, label $\mu|V| - w(\sigma)$-arbitrary $0$-labeled vertices in $\sigma$ as $1$ such that the resulting labeling $\sigma_f$ satisfies $w(\sigma_f) = \mu|V|$.
				\item Return $\sigma_f$.   
			\end{enumerate}
		\end{mdframed}
	\caption{$\dksh \to \psi$ Reduction}	
	\label{fig:ref-1}
	\end{figure}
	 
	The following is the main claim which states the guarantee of the above reduction.
	
	\begin{claim}			\label{cl:red}
		The algorithm in Figure \ref{fig:ref-1} returns a labeling $\sigma_f:V \to \{0,1\}$ of weight $\mu n$ in $H$ such that 
		\[
		{\sf Val}_{\sigma_f}(H) \geq \alpha_{\leq \mu'}(\psi) \cdot {\sf Val}_{(\mu)}(H).
		\]
	\end{claim}
	
	Before we prove the above, note that the above claim concludes the proof of the lemma since it yields an $\alpha_{\leq \mu'}(\psi)$-approximation algorithm for $\dksh$ instances with $k = \mu|V|$, thus implying that $\alpha^\uw_{(\mu)}(\dksh_{i^*}) \geq \alpha_{\leq \mu'}(\psi)$.

	\begin{proof}[Proof of Claim \ref{cl:red}]
	The proof of the claim involves establishing the following:
	\begin{enumerate}
		\item Firstly we will show that the optimal value of $\Psi'$ is large i.e, ${\sf Val}_{\leq \mu'}(\Psi') \geq {\sf Val}_{(\mu)}(H)$ (Eq. \eqref{eqn:opt-1}). This would imply that the labeling $\sigma'$ in step \ref{step:appx} has large value in $\Psi'$ (Eq. \eqref{eqn:sig}).
		\item Next we shall show that the feasibility of $\sigma'$ implies that its support is contained in $V$, and hence it can be trivially decoded to a labeling $\sigma$ (and consequently $\sigma_f$) for $H$ which has large value (Eq. \eqref{eqn:sig2}).
	\end{enumerate}	
	
	We now proceed to establish the above steps.
		
	{\bf $\Psi'$ has large value}. We introduce an additional notation. For any edge $e \in E'$, we write $e = e^{\leq i^*} \circ e^{> i^*}$ where $e^{\leq i^*} := e \cap V$.
	Now, let $\sigma^*:V \to \{0,1\}$ be a labeling of relative weight $\mu$ for $H$ which achieves the optimal value ${\sf Val}_{(\mu)}(H)$ in $H$. Then, consider the labeling $\sigma_1:V' \to \{0,1\}$ which we define as follows:
	\[
	\sigma_1(i) := 
	\begin{cases}
		\sigma^*(i) & \mbox{ if } i \in V, \\
		0 		  & \mbox{ if } i \in V' \setminus V,
	\end{cases}
	\] 
	The definition of $\sigma_1$ implies that the relative weight of $\sigma_1$ in $\Psi'$ is exactly $w(\sigma^*)/((r - i^* + 1) n) = \mu n/((r - i^* + 1)n) = \mu'$ i.e., $\sigma_1$ is a feasible labeling for $\Psi'$. Hence,
	\begin{align}
		{\sf Val}_{\leq \mu'}(\Psi') \geq {\sf Val}_{\sigma_1}(\Psi') 
		&\geq \Pr_{e \sim E'}\left[\sigma_1(e) = \beta\right]		\non \\
		&= \Pr_{e \sim E'}\left[\{\sigma_1(e^{\leq i^*}) = {\bf 1}\} \wedge \{\sigma_1(e^{> i^*}) = {\bf 0}\}\right] 		\non\\
		&= \Pr_{e \sim E'}\left[\sigma_1(e^{\leq i^*}) = {\bf 1}\right] 	\tag{Since $\sigma_1(a_j) = 0,~\forall j$}	\non\\
		&= \Pr_{e \sim E}\left[\sigma^*(e) = {\bf 1}\right] 		\non\\
		&= {\sf Val}_{(\mu)}(H).				\label{eqn:opt-1}
	\end{align}

	Therefore, Line \ref{step:appx} in Figure \ref{fig:ref-1} returns a labeling $\sigma'$ of relative weight at most $\mu'$ in $\Psi'$ satisfying 
	\begin{equation}			\label{eqn:sig}
	{\sf Val}_{\sigma'}(\Psi') \geq \alpha^* \cdot {\sf Val}_{\leq \mu}(\Psi') \geq \alpha^* \cdot {\sf Val}_{(\mu)}(H),
	\end{equation}
	where recall that $\alpha^* = \alpha_{\leq \mu}(\psi)$ and the last step is due to \eqref{eqn:opt-1}. 
	
	{\bf Decoding $\sigma_f$ from $\sigma'$}. Next we claim that since $\sigma'$ has relative weight at most $\mu'$ in $\Psi'$, we must have ${\rm supp}(\sigma') \subseteq V$; since otherwise we have  
	\begin{equation}			\label{eqn:empty}
	\frac{w(\sigma')}{w(V')} \geq \frac{n\left|\left\{ j : \sigma'(a_j) = 1 \right\}\right|}{(r - i^* + 1)n} 
	\geq \frac{1}{r - i^* + 1} 
	> \frac{\mu}{r - i^* + 1} = \mu',
	\end{equation}
	which is a contradiction. This immediately implies the following observations for the labeling $\sigma'$:
	\begin{itemize}
		\item For any edge $e \in E'$, we have $\sigma'(e^{>i^*}) = {\bf 0}$.
		\item If $\sigma'$ satisfies edge $e \in E'$, then $\sigma'(e) = \beta$. This follows from the observation that since $\beta \in \cM_{\psi}$, it is the unique accepting string for $\psi$ whose last $(r - i^*)$-bits are set to $0$.
	\end{itemize}
	
	Using the above, for any edge $e = e^{\leq i^*} \circ e^{>i^*} \in E'$ we have the following chain of implications:
	\begin{align*}
		\psi(\sigma'(e)) = 1 
		& \Rightarrow \sigma'(e) = \beta 			\tag{Minimality of $\beta$} \\
		& \Rightarrow \left\{\sigma'(e^{\leq i^*}) = {\bf 1}\right\} \wedge \left\{\sigma'(e^{>i^*}) = {\bf 0}\right\} \\
		& \Rightarrow \left\{\sigma(e^{\leq i^*}) = {\bf 1}\right\},  
	\end{align*}
	i.e, whenever an edge $e \in E'$ is satisfied by the labeling $\sigma'$ the corresponding truncated edge $e^{\leq i^*} \in E$ is satisfied by $\sigma$. Hence, using the above observations and the one-to-one correspondence between the constraints in $H$ and $\Psi'$ we have 
	\begin{equation}				\label{eqn:sig2}
	{\sf Val}_{\sigma}(H) =\Pr_{e \sim E}\left[\sigma(e) = {\bf 1}\right] \geq \Pr_{e \sim E'}\left[\psi(\sigma'(e)) = 1\right] = \alpha^* \cdot {\sf Val}_{(\mu)}(H).
	\end{equation}
	Finally, observe that the weight of the truncated labeling $\sigma$ is identical to that of $\sigma'$ which is at most $\mu|V|$, which is then extended to a labeling $\sigma_f$ of weight exactly $\mu|V|$. Additionally, since ${\rm supp}(\sigma_f) \supseteq {\rm supp}(\sigma)$, we have ${\sf Val}_{\sigma_f}(H) \geq {\sf Val}_{\sigma}(H)$. These observations put together complete the proof of the claim.
\end{proof}

\end{proof}

\section{Algorithm using $\dksh$}

In this section, we establish an (approximate) converse to the Lemma \ref{lem:pred-hard} in the following lemma which lower bounds the bias of approximation curve of general CSPs using that of $\dksh$.

\begin{lemma}				\label{lem:at-most-0}
	The following holds for any predicate $\psi:\{0,1\}^r \to \{0,1\}$, and any choices of $\mu \in (0,1/2)$ and $\eta \in (\mu^2,1)$. Given a biased CSP $\Psi(V,E,w,\mu,\psi)$, there exists an efficient randomized algorithm which on input $\Psi$ returns a labeling $\sigma':V \to \{0,1\}$ of relative weight at most $\mu(1 + \eta)$ such that 
	\[
	{\sf Val}_{\sigma'}(\Psi) \gtrsim_r \eta^r \cdot {\sf Val}_{\leq \mu}(\Psi) \cdot \min_{\beta \in \cM_\psi} \alpha^\uw_{(\mu)}\left(\dksh_{\|\beta\|_0}\right) 
	\] 
\end{lemma}

The proof of the above lemma involves a sequence of algorithmic reductions which eventually reduces to the setting of unweighted $\dksh$ instances. We give a brief outline of the various steps involved:

\begin{itemize}
	\item Firstly, we consider the setting of ``single string predicates'' i.e. predicates whose accepting set is a singleton and show that they can be reduced to weighted $\ldksh$ instances (Section \ref{sec:single-pred}).
	\item Then, we use use the algorithm for single string predicates as a black-box to derive approximation guarantees for general predicates (Section \ref{sec:gen-pred}).
	\item Finally, we use folklore techniques to reduce weighted $\ldksh$ instances to uniformly weighted $\dksh$ instances at the cost of a multiplicative slack in the bias (Section \ref{sec:dksh-red}).
\end{itemize}

The rest of this and the next section establishes the various steps sketched above.

\subsection{Approximating Single String Predicates}				\label{sec:single-pred}

In this section, we show that we can use approximation algorithms for \dksh~as black-boxes for approximating \mcsp's on single string predicates -- i.e, predicates whose set of accepting strings is a singleton. In particular, for any $\beta \in \{0,1\}^r$, we define the $\psi_\beta$-predicate instance on a constraint graph $\Psi(V,E)$ as follows. Define an ordering $v_1 \prec v_2 \prec \cdots \prec v_n$ on the set of vertices. Then a labeling $\sigma:V \to \{0,1\}$ satisfies the {\em ordered} constraint $e = (v_{i_1},v_{i_2},\ldots,v_{i_r}) \in E$ (where the vertices are ordered according to $\prec$) if and only if $(\sigma(v_{i_1}),\ldots,\sigma(v_{i_r})) = \beta$.

It is easy to see that $\dksh$ can be recovered as a special case of $\psi_\beta$ (i.e., $\beta = (1,\ldots,1)$). Our main observation here is that up to factors depending on $r$, every $\psi_\beta$-predicate is at least as easy to approximate as $\dksh_{\|\beta\|_0}$.

\begin{lemma}			\label{lem:at-most}
	The following holds for any $\beta \in \{0,1\}^r$, $\mu \in (0,1/2)$ and $\eta \in (0,1)$. Given an instance $\Psi(V,E,w,\mu,\psi_\beta)$ with single string predicate constraints $\psi_\beta$, there exists an efficient algorithm which returns a labeling $\sigma:V \to \{0,1\}$ of relative weight at most $(1 + \eta) \mu$ satisfying:
	\[
	{\sf Val}_{\sigma}(\Psi) \gtrsim_r  \eta^r \alpha^\uw_{(\mu)} \left(\dksh_{\|\beta\|_0}\right) \cdot {\sf Val}_{\leq \mu}(\Psi).
	\]
\end{lemma}
\begin{proof}
	By reordering, assume $\beta = \sum^{i^*}_{j = 1} e_j$. We shall prove the lemma by demonstrating an efficient $\Omega_r\left(\alpha^\uw_{(\mu)}\left(\dksh_{\|\beta\|_0}\right)\right)$-approximation algorithm. The algorithm for the lemma is described in Figure \ref{fig:alg-1}.
	
	\begin{figure}[ht!]
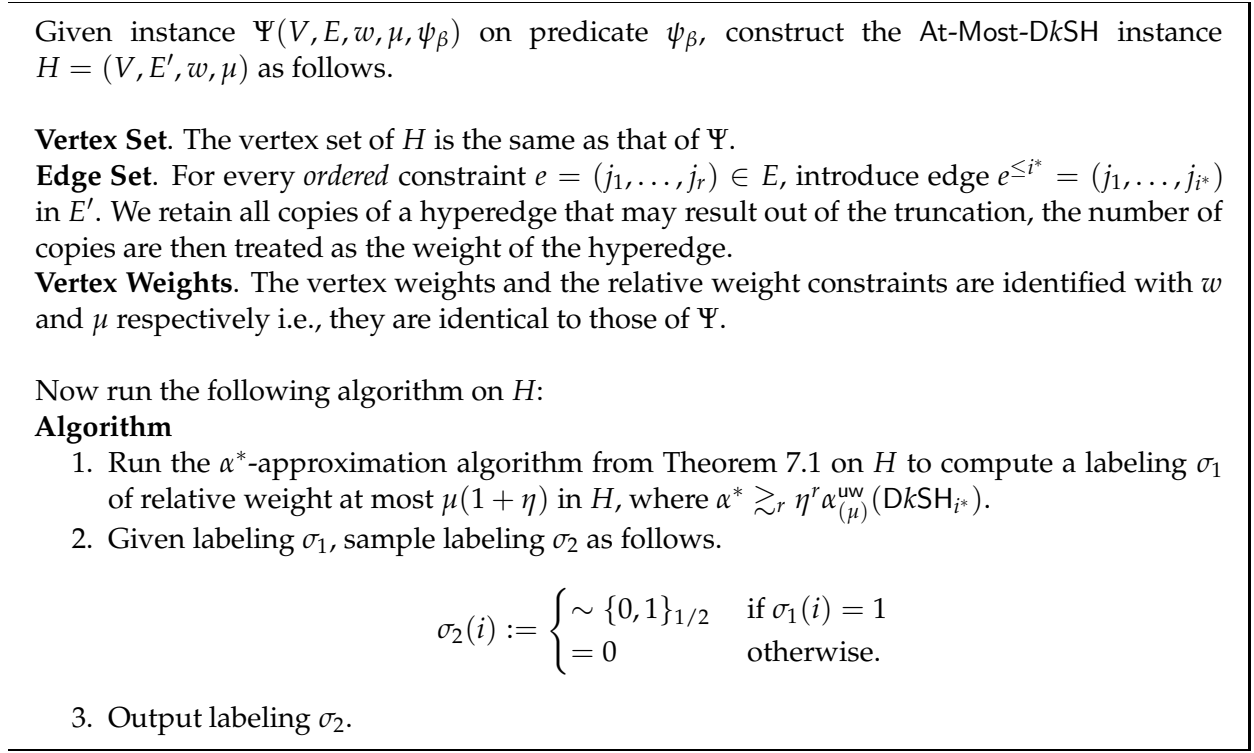

		\begin{mdframed}
			Given instance $\Psi(V,E,w,\mu,\psi_\beta)$ on predicate $\psi_\beta$, construct the $\ldksh$ instance $H =(V,E',w,\mu)$ as follows. \\
			
			{\bf Vertex Set}. The vertex set of $H$ is the same as that of $\Psi$. \\
			{\bf Edge Set}. For every {\em ordered} constraint $e = (j_1,\ldots,j_r) \in E$, introduce edge $e^{\leq i^*} = (j_1,\ldots,j_{i^*})$ in $E'$. We retain all copies of a hyperedge that may result out of the truncation, the number of copies are then treated as the weight of the hyperedge. \\
			{\bf Vertex Weights}. The vertex weights and the relative weight constraints are identified with $w$ and $\mu$ respectively i.e., they are identical to those of $\Psi$. \\
			
			Now run the following algorithm on $H$:
			
			{\bf Algorithm}
			\begin{enumerate}
				\item Run the $\alpha^*$-approximation algorithm from Theorem \ref{thm:dksh-weight} on $H$ to compute a labeling $\sigma_1$ of relative weight at most $\mu(1 + \eta)$ in $H$, where $\alpha^*  \gtrsim_r \eta^r{\alpha}^\uw_{(\mu)}(\dksh_{i^*})$.
				\item Given labeling $\sigma_1$, sample labeling $\sigma_2$ as follows.
				\[
				\sigma_2(i) := 
				\begin{cases}
					\sim \{0,1\}_{1/2} & \mbox{ if } \sigma_1(i) = 1 \\
					= 0 & \mbox{ otherwise. }
				\end{cases}
				\]
				\item Output labeling $\sigma_2$.	
			\end{enumerate}
		\end{mdframed}
		\caption{Approximating $\psi_\beta$}
		\label{fig:alg-1}
	\end{figure}
	
	We shall now analyze the above algorithm. To begin with, we have the following claim.
	\begin{claim}
		The instance $H$ constructed in Figure \ref{fig:alg-1} satisfies
		\[
		{\sf Val}_{\leq \mu}(H) \geq {\sf Val}_{\leq\mu}(\Psi).
		\]
	\end{claim}
	
	\begin{proof}
		Let $\sigma^*:V \to \{0,1\}$ be the optimal labeling of weight at most $\mu$ for $\Psi$ on predicate $\psi_\beta$. Then for any edge $e = (j_1,\ldots,j_r)\in E$ we have 
		\[
		\psi_{\beta}(\sigma^*(e) ) = 1\Rightarrow \bigwedge_{\ell \in [i^*]} \sigma^*(j_\ell) = 1,
		\]
		i.e., whenever $\sigma^*$ satisfies an edge $e$ in $\Psi$, it also induces the corresponding truncated hyperedge $e^{\leq i^*}$ in $H$. Then using the one-to-one correspondence between the constraints in $\Psi$ and $H$ we have ${\sf Val}_{\sigma^*} (H) \geq {\sf Val}_{\sigma^*}(\Psi)$. Furthermore, note that the vertex weights for $H$ and $\Psi$ are identical and hence $\sigma^*$ is also a labeling of relative weight at most $\mu$ in $H$. Hence,
		\[
		{\sf Val}_{\leq \mu} (H) \geq {\sf Val}_{\sigma^*}(H) \geq {\sf Val}_{\sigma^*}(\Psi) = {\sf Val}_{\leq\mu}(\Psi).
		\]
	\end{proof}
	
	From the above claim, it follows that in step $1$, the algorithm from Theorem \ref{thm:dksh-weight} returns a labeling $\sigma_1$ of relative weight at most $\mu(1 + \eta)$ in $H$ such that 
	\begin{equation}			\label{eqn:beta-val}
		{\sf Val}_{\sigma_1}\left(H\right) \gtrsim_r 
		\alpha^* \cdot {\sf Val}_{\leq \mu}(H) \geq \alpha^* \cdot {\sf Val}_{\leq \mu}\left(\Psi\right).
	\end{equation}
	where $\alpha^* \gtrsim_r \eta^r \cdot \alpha^\uw_{(\mu)}(\dksh_{\|\beta\|_0})$. 
	
	Now we derive the guarantees for the final labeling $\sigma_2$ using $\sigma_1$. Towards that, we have the following observations.
	
	\begin{claim}				\label{cl:alg-1a}
		With probability $1$, the relative weight of $\sigma_2$ in $\Psi$ is at most $\mu(1 + \eta)$.
	\end{claim} 
	\begin{proof}
		From the sub-sampling we know that ${\rm supp}(\sigma_2) \subseteq {\rm supp}(\sigma_1)$, and hence $w(\sigma_2) \leq w(\sigma_1) \leq\mu(1 + \eta)$ with probability $1$.
	\end{proof}
	Next we want to bound the fraction of edges satisfied by $\sigma_2$. Towards that, we have the following claim:
	\begin{claim}				\label{cl:alg-1b}
		Let $e \in E$ be an edge for which the truncated edge $e^{\leq i^*}$ is satisfied by labeling $\sigma_1$. Then,
		\[
		\Pr_{\sigma_2}\left[ \sigma_2 \mbox{ satisfies } e^{\leq i^*}  \right] \geq 2^{-r}.
		\]
	\end{claim}
	\begin{proof}
		Let $e = (j_1,\ldots,j_r)$. Then,
		\begin{align}
			\Pr_{\sigma_2}\Big[ \sigma_2 \mbox{ satisfies } e  \Big]
			& = \Pr_{\sigma_2}\Big[ \sigma_2(e) = \beta  \Big] \\
			& = \Pr_{\sigma_2} \left[\bigwedge_{\ell \leq i^*} \sigma_2(j_\ell) = 1, \bigwedge_{\ell > i^*} \sigma_2(j_\ell) = 0\right] \non\\
			& \overset{1}{=} \Pr_{\sigma_2} \left[\bigwedge_{\ell \leq i^*} \sigma_2(j_\ell) = 1, \bigwedge_{\ell > i^*: \sigma_1(j_\ell) = 1} \sigma_2(i) = 0\right] 		\label{eqn:expr}\\
			& \overset{2}{=} 2^{-i^* + |\{\ell > i^*: \sigma_1(j_\ell) = 1\}|} 	\non\\
			& \geq 2^{-r}		\non
		\end{align}
		where in step $1$, we use the fact that whenever $\sigma_1(j_\ell) = 0$, we have $\sigma_2(j_\ell) = 0$ with probability $1$. For step $2$,  we observe that  since $\sigma_1$ satisfies $e^{\leq i^*}$, $\sigma_1(j_\ell) = 1$ for every $\ell \leq i^*$. Furthermore, conditioned on $\sigma_1(j_\ell) = 1$, $\sigma_2(j_\ell)$ is an independent uniform $\{0,1\}$ random variable. Therefore, for every index $j_\ell$ present in the expression \eqref{eqn:expr}, $\sigma_2(j_\ell)$ is an independent $\{0,1\}_{1/2}$ random variable, which implies step $2$.  
	\end{proof}
	The above claim along with the bound on \eqref{eqn:beta-val} yields the following corollary.
	\begin{corollary}				\label{corr:alg-1c}
		\[
		\Ex_{\sigma_2}\left[{\sf Val}_{\sigma_2}(\Psi)\right] \geq 2^{-r} \cdot \alpha^* \cdot {\sf Val}_{\leq \mu} \left(\Psi\right).
		\]
	\end{corollary}
	\begin{proof}
		Let $E_{\rm sat} \subset E$ be the subset of edges $e$ for which the truncated edge $e^{\leq i^*}$ is satisfied by  $\sigma_1$ i.e, $\sigma_1(e^{\leq i^*}) = {\bf 1}_{i^*}$. Then,
		\begin{align*}
			\Ex_{\sigma_2}\Big[{\sf Val}_{\sigma_2}(H)\Big]
			&= \Ex_{e \sim E} \Pr_{\sigma_2}\Big[\sigma_2 \mbox{ satisfies } e\Big] \\
			&= \Pr_{e \sim E}\Big[e \in E_{\rm sat}\Big] \Ex_{e \sim E| e\in E_{\rm sat}} \Pr_{\sigma_2}\Big[\sigma_2 \mbox{ satisfies } e\Big] \\
			&\geq 2^{-r} \Pr_{e \sim E}\Big[e \in E_{\rm sat}\Big] 			\tag{Claim \ref{cl:alg-1b}}\\
			& = 2^{-r} \Pr_{e \sim E}\Big[\sigma_1(e^{\leq i^*}) = {\bf 1}_{i^*}\Big] \\
			& = 2^{-r} \Pr_{e' \sim E'}\Big[\sigma_1 \mbox{ satisfies } e'\Big] \\
			& = 2^{-r} \cdot {\sf Val}_{\sigma_1}(H) 			\\
			&  \geq 2^{-r} \cdot \alpha^* \cdot {\sf Val}_{\leq \mu} \left(\Psi\right).		\tag{Using \eqref{eqn:beta-val}}
		\end{align*}
	\end{proof}

	{\bf Cleaning Up}. Using Corollary \ref{corr:alg-1c}, it follows that the above algorithm returns a labeling of relative weight at most $\mu(1 + \eta)$ in $\Psi$ which satisfies 
	\[
	\Ex_{\sigma_2}\left[{\sf Val}_{\sigma_2}\left(\Psi\right)\right] \geq 2^{-r} \alpha^* \cdot {\sf Val}_{\leq \mu}\left(\Psi\right).
	\]
	Furthermore, with probability 1, $\sigma_2$ is of relative weight at most $\mu(1 + \eta)$ (Claim \ref{cl:alg-1a}). Therefore, assuming ${\sf Val}_{\leq \mu}\left(\Psi\right) \geq \mu^{O(r)}$, with probability at least $2^{-r} \alpha^* \mu^{O(r)}$, the labeling $\sigma_2$ is of weight at most $\mu(1 + \eta)$ and satisfies:
	\[
	{\sf Val}_{\sigma_2}\left(\Psi\right) \geq 2^{-r-1} \alpha^* \cdot {\sf Val}_{\leq \mu}\left(\Psi\right)
	\]
	Hence repeating the algorithm for $(1/\mu)^{O(r)}$-independent rounds and choosing the best labeling would yield a labeling satisfying the guarantees of the lemma with high probability.
\end{proof}

\subsection{Proof of Lemma \ref{lem:at-most-0}}				\label{sec:gen-pred}

Using Lemma \ref{lem:at-most}, we now prove Lemma \ref{lem:at-most-0}.

\begin{proof}
	The algorithm for Theorem \ref{lem:at-most-0} is described as Algorithm \ref{alg:ref-0}. 
	
	\begin{algorithm}[ht!]
		\SetAlgoLined
		{\bf Input}. Instance $\Psi(V,E,w,\mu,\psi)$ of arity $r$ with predicate $\psi$. \;
		\For{$ \beta \in \cM_\psi$}
		{
			Consider the instance $\Psi_{\beta}(V,E,w,\mu,\psi_\beta)$ which is the instance $\Psi$ with the edge constraints $\psi$ replaced by $\psi_{\beta}$\; 
			Use the algorithm from Lemma \ref{lem:at-most} on $\Psi_\beta$ to compute a labeling $\sigma_\beta:V \to \{0,1\}$ of relative weight at most $\mu(1 + \eta)$ in $\Psi_\beta$\;
		}
		Let 
		\begin{equation}				\label{eqn:sigp-def}
		\sigma' = \argmax_{\beta \in \cM_{\psi}} {\sf Val}_{\sigma_\beta}(\Psi).
		\end{equation}			\label{step:sigp-def}
		Return labeling $\sigma'$\;
		\caption{``At-most'' Approximation Algorithm}
		\label{alg:ref-0}
	\end{algorithm} 
	
	Towards analyzing the above algorithm, we shall first need the following claim which shows that one can always find a labeling of relative weight at most $\mu$ which satisfies a significant fraction of edges (in comparison to the optimal) using patterns from the minimal set $\cM_\psi$. 
	
	\begin{claim}				\label{cl:min-val}
		There exists a labeling $\sigma:V \to \{0,1\}$ of relative weight at most $\mu$ such that 
		\begin{equation}			\label{eqn:sat}
			\Pr_{e \sim E} \left[\sigma(e) \in \cM_\psi \right] \geq 2^{-r} \cdot{\sf Val}_{\leq \mu}(\Psi).
		\end{equation}	
	\end{claim}
	\begin{proof}
		We establish the claim using the probabilistic method. Let $\sigma^*: V \to \{0,1\}$ be the optimal labeling of relative weight at most $\mu$ achieving ${\sf Val}_{\leq \mu}(\Psi)$ in $\Psi$. Then given $\sigma^*$, we construct $\sigma:V \to \{0,1\}$ as follows:
		\begin{equation}				\label{eqn:sig-def0}
		\sigma(i) := 
		\begin{cases}
			\sim \{0,1\}_{1/2} & \mbox{ if } \sigma^*(i) = 1, \\
			0 & \mbox{ if } \sigma^*(i) = 0.
		\end{cases}
		\end{equation}
		By definition, we have ${\rm supp}(\sigma) \subseteq {\rm supp}(\sigma^*)$ and hence, $\sigma$ is also a labeling of relative weight at most $\mu$ with probability $1$. Now towards establishing \eqref{eqn:sat}, let $E' \subseteq E$ be the subset of edges satisfied by the labeling $\sigma^*$ in $\Psi$. For every edge $e \in E'$, let $\beta_e \in \cM_\psi$ be such that $\beta_e \preceq \sigma^*(e)$ (note that such a $\beta_e$ must exist since $\sigma^*(e) \in \psi^{-1}(1)$ whenever $e \in E'$). Then,
		\begin{align}
			&\Ex_{\sigma}\Ex_{e \sim E} \left[ \mathbbm{1}_{\{\sigma(e) \in \cM_{\psi}\}}\right] \\		
			& \geq \Pr_{e \sim E}\left[e \in E'\right]\Ex_{e \sim E|e \in E'} \Ex_{\sigma}\left[ \mathbbm{1}_{\{\sigma(e) \in \cM_{\psi}\}}\right] 
			\non\\
			& \geq \Pr_{e \sim E}\left[e \in E'\right]\Ex_{e \sim E | e \in E'} \Ex_{\sigma}\left[ \mathbbm{1}_{\{\sigma(e) = \beta_e\}}\right] 
			\tag{Since $\beta_e \in \cM_\psi$}		\non\\
			& {=} \Pr_{e \sim E}\left[e \in E'\right]\Ex_{e = (i_1,\ldots,i_r) \sim E| e \in E'} \Pr_{\sigma}\left[\bigwedge_{j \in [r]} \sigma(i_j) = \beta_e(j) \right] 		\non\\
			& \overset{1}{=} \Pr_{e \sim E}\left[e \in E'\right]\Ex_{e = (i_1,\ldots,i_r) \sim E| e \in E'} \Pr_{\sigma}\left[ \bigwedge_{j \in [r], \sigma^*(i_j) = 1} \sigma(i_j) = \beta_e(j)\right] 		\non\\
			& \overset{2}{=} \Pr_{e \sim E}\left[e \in E'\right]\Ex_{e = (i_1,\ldots,i_r) \sim E | e \in E'} \left[2^{-\Big|\{j \in [r]: \sigma^*(i_j) = 1\}\Big|}\right] 			\non\\
			& \geq 2^{-r} \Pr_{e \sim E}\left[e \in E'\right] 		\non\\
			& = 2^{-r} \Pr_{e \sim E}\left[\psi(\sigma^*(e)) = 1\right] 		\non\\
			& = 2^{-r} {\sf Val}_{\leq \mu}(\Psi),			\label{eqn:sig-val}
		\end{align}
		where step $1$ can be argued as follows. Note that for any $e = (i_1,\ldots,i_r) \in E'$, and for any $j \in [r]$,  $\sigma^*(i_j) = 0$ implies the events $\{\beta_e(j) = 0\}$ (since  $\beta_e \preceq \sigma^*(e))$ and $\{\sigma(i_j) = 0\}$ (using the definition of $\sigma$ from \eqref{eqn:sig-def0}). Hence, 
		\[
		\bigwedge_{j \in [r]: \sigma^*(i_j) = 0} (\sigma(i_j) = \beta(j))
		\]
		holds with probability $1$. Therefore, in order to have $\sigma(e) = \beta_e$, it suffices to ensure that $\sigma(i_j) = \beta_e(j)$ for every $j \in [r]$ such that $\sigma^*(i_j) = 1$. In step $2$, we use the fact that conditioned on $\sigma^*(i_j) = 1$, each $\sigma(i_j)$ is an independent uniform $\{0,1\}$-random variable. 
		
		Now to conclude the proof, observe that every realizable labeling in the support of distribution of $\sigma$ has relative weight at most $\mu$, and in expectation, $\sigma$ satisfies \eqref{eqn:sig-val}. Hence there must exists a labeling of relative weight at most $\mu$ satisfying \eqref{eqn:sig-val}.
		
	\end{proof}
	
	Let $\sigma:V \to \{0,1\}$ be the labeling guaranteed by Claim \ref{cl:min-val}. Then,
	\begin{align*}
	2^{-r} {\sf Val}_{\leq \mu}\Big(\Psi\Big) &\leq \Pr_{e \sim E} \left[\bigvee_{\beta \in \cM_\psi} \sigma(e) = \beta\right] \\
											  &\leq \sum_{\beta \in\cM_{\psi}} \Pr_{e \sim E} \Big[\sigma(e) = \beta\Big] \\
											  &\leq |\cM_\psi|\max_{\beta \in\cM_{\psi}}~\Pr_{e \sim E} \Big[\sigma(e) = \beta\Big] 	\\
											  &\leq 2^r\max_{\beta \in\cM_{\psi}}~\Pr_{e \sim E} \Big[\sigma(e) = \beta\Big], 	
	\end{align*}
	which implies that there exists $\beta^* \in \cM_{\psi}$ satisfying
	\[
	\Pr_{e \sim E} \left[\sigma(e) = \beta^*\right] \geq 2^{-2r} {\sf Val}_{\leq \mu}(\Psi).
	\]
	Furthermore, note that in Algorithm \ref{alg:ref-0}, the instance $\Psi_{\beta^*}$ shares the same vertex set and vertex weights as $\Psi$ and hence $\sigma$ is a labeling of relative weight at most $\mu$ in $\Psi_{\beta^*}$ as well. Hence,
	\begin{equation}			\label{eqn:beta}
		{\sf Val}_{\leq \mu} \left(\Psi_{\beta^*}\right) \geq {\sf Val}_{\sigma}(\Psi_{\beta^*}) =
		\Pr_{e \sim E} \left[\sigma(e) = \beta^*\right] \geq 2^{-2r} {\sf Val}_{\leq \mu}\left(\Psi\right).
	\end{equation}
	
	Therefore, in the for loop, instantiating $\beta = \beta^*$, running the algorithm from Lemma \ref{lem:at-most} on $\Psi_{\beta^*}$ returns a labeling $\sigma_{\beta^*}:V \to \{0,1\}$ of relative weight at most $\mu(1 + \eta)$ satisfying:
	\begin{align}
		{\sf Val}_{\sigma_{\beta^*}}\left(\Psi_{\beta^*}\right) 
		&\gtrsim_r \eta^r\alpha^\uw_{(\mu)}\left(\dksh_{\|\beta^*\|_0}\right) \cdot {\sf Val}_{\leq \mu}\left(\Psi_{\beta^*}\right) 	\non\\
		&\geq 2^{-2r}{\sf Val}_{\leq \mu}\left(\Psi\right) \cdot \eta^r\alpha^\uw_{(\mu)}\left(\dksh_{\|\beta^*\|_0}\right) \tag{Using \eqref{eqn:beta}} 
		\non\\
		&\geq 2^{-2r}{\sf Val}_{\leq \mu}\left(\Psi\right) \cdot \eta^r\min_{\beta \in \cM_{\psi}}\alpha^\uw_{(\mu)}\left(\dksh_{\|\beta\|_0}\right) 
		\label{eqn:star}.
	\end{align} 
	
	{\bf Putting Things Together}. We conclude the proof by arguing that the labeling $\sigma'$ (from \eqref{eqn:sigp-def}) satisfies the guarantees stated by the lemma. We first observe that whenever the labeling $\sigma_{\beta^*}$ satisfies a constraint $e \in E$  in $\Psi_{\beta^*}$ with respect to $\psi_{\beta^*}$, it also satisfies the edge $e$ with respect to constraint $\psi$ in $\Psi$ (since $\beta^* \in \psi^{-1}(1)$) and hence ${\sf Val}_{\sigma_{\beta^*}}(\Psi) \geq {\sf Val}_{\sigma_{\beta^*}}(\Psi_{\beta^*})$. Therefore the final labeling $\sigma'$ from Line \ref{step:sigp-def} satisfies:
	\begin{align*}
	{\sf Val}_{\sigma'}(\Psi) 
	& = \max_{\beta \in \cM_\psi} {\sf Val}_{\sigma_\beta}(\Psi) \\
	&\geq {\sf Val}_{\sigma_{\beta^*}}(\Psi) \\
	&\geq {\sf Val}_{\sigma_{\beta^*}}(\Psi_{\beta^*}) \\
	&\gtrsim_r \eta^r \cdot {\sf Val}_{\leq \mu}\left(\Psi\right) \cdot \min_{\beta \in \cM_{\psi}}\alpha^\uw_{(\mu)}\left(\dksh_{\|\beta\|_0}\right). \tag{Using \eqref{eqn:star}}
	\end{align*}
	Now note that for every $\beta \in \{0,1\}^r$, the labeling $\sigma_\beta$ is also of relative weight at most $\mu(1 + \eta)$ in $\Psi$ using the guarantee of Lemma \ref{lem:at-most}. Again, since $\sigma' = \sigma_{\beta'}$ for some $\beta' \in \{0,1\}^r$,  $\sigma'$ has relative weight at most $\mu(1 + \eta)$ in $\Psi$. 
	
\end{proof}

\section{Reducing weighted $\ldksh$ to $\dksh$.}				\label{sec:dksh-red}

In this section, we give a reduction from weighted $\ldksh$ instance to uniformly weighted $\dksh$ instances, as stated in the following theorem.

\begin{theorem}				\label{thm:dksh-weight}
	For any constant $\eta \in (\mu,1)$, there exists an efficient algorithm that on input $\ldksh$ instances $H = (V,E,w,\mu)$ with arbitrary vertex weights, returns a labeling $\tilde{\pi}:V \to \{0,1\}$ of $V$ with relative weight at most $\mu(1 + \eta)$ satisfying
	\[
	{\sf Val}_{\tilde{\pi}}(H) \gtrsim_r \left(\eta^r \alpha^{\sf uw}_{(\mu)} \left(\dksh_r\right)\right) \cdot {\sf Val}_{\leq \mu}(H),
	\]
	where $\alpha^{\sf uw}_{(\mu)}\left(\dksh_r\right)$ is the optimal approximation guarantee for uniformly weighted $\mu$-biased $\dksh_r$ instances.
\end{theorem}

The algorithm for the above theorem is described in Algorithm \ref{alg:ref-1}.

\begin{algorithm}[ht!]
	\SetAlgoLined
	{\bf Input}. An $\ldksh$-instance $H = (V,E,w,\mu)$ of arity $r$ satisfying\footnotemark $\|w\|_1 = 1$ \;
	Define
	\[
	T := \left\{i \in V \Big| w(i) > \mu^{10} \right\}.
	\]
	Note that by definition $|T| \leq (1/\mu)^{10}$\;
	\For{ every partial labeling $\sigma_T:\{0,1\}^T \to \{0,1\}$ such that $w(\sigma_T) \leq \mu$}		
	{
		\If{$w(V\setminus T) < \mu\eta$}		
		{\label{step:for-1}
			Construct labeling $\pi'_{\sigma_T}:V \to \{0,1\}$ as 
			\[
			\pi'_{\sigma_T}(i) = 
			\begin{cases}
				1  				& \mbox{ if } i \in V \setminus T \\
				\sigma_T(i) 	& \mbox{ if } i \in T
			\end{cases}
			\]
			Skip to next $\sigma_T$		\label{step:for-a}	\;
		}		
		Let $H_{\sigma_T} = (V \setminus T,E_{\sigma_T},w_{\sigma_T},\delta_{\sigma_T})$ be the induced $\ldksh$-instance of arity at most $r$ constructed as follows.\;
		\begin{mdframed}	
			{\bf Vertex Set}. The vertex set is $V \setminus T$. \\
			{\bf Edge Set}. The edge set is constructed as follows. For every edge $e \in E$ such that $\sigma_T(e|_T) = {\bf 1}$, introduce edge $e|_{T^c}$ on the vertex set $V \setminus T$, where $e|_{T} := e \cap T$ and similarly $e|_{T^c}$. Note that the truncation may introduce multiple copies of the same edge. We shall retain them all and treat the number of copies as the weight of the edge (also see Remark \ref{rem:weights}). \\
			{\bf Vertex Weights}. For every vertex $i \in V \setminus T$, set $w_{\sigma_T}(i) = w(i)/w(V \setminus T)$. Define
			\begin{equation}				\label{eqn:del-def}
				\delta_{\sigma_T} :=  \frac{\mu(1 + \eta) - w(\sigma_T)}{w(V \setminus T)}.
			\end{equation}
		\end{mdframed}
		
		Let $\pi_{\sigma_T}:V \setminus T \to \{0,1\}$ be the labeling of weight at most $\delta_{\sigma_{T}}$ obtained by running the $\alpha^{\sf uw}_{(\delta_{\sigma_T})}(\dksh_r)$-approximation algorithm from Lemma \ref{lem:dksh-red-2} on $H_{\sigma_T}$\label{step:dksh}\;
		Finally, let $\pi'_{\sigma_T}:= \sigma_T \circ \pi_{\sigma_T}$ be the concatenation of the two labelings\label{step:for-2}	 \;
	}
	Let $\tilde{\pi} : = \argmax_{\pi'_{\sigma_T}: w(\sigma_T) \leq \mu} {\sf Val}_{\pi'_{\sigma_T}} (H)$		\label{step:final}\;
	Return labeling $\tilde{\pi}$\;
	\caption{Weighted $\dksh$}
	\label{alg:ref-1}
\end{algorithm} 
\footnotetext{This is without loss of generality by rescaling the vertex weights accordingly.}
The key principle behind the above algorithm is the observation that for any $c \in (0,1)$, the set of vertices with relative weight larger than $\mu^c$ (i.e the set $T$) is at most ${\rm poly}(1/\mu)$, and hence for a constant $\mu$, one can enumerate over all feasible labelings for the set of large weight vertices in polynomial time and solve the instance induced for each possible guess (Lines \ref{step:for-1} - \ref{step:for-2}). Furthermore, without loss of generality, one may assume that the induced instance $H_{\sigma_T}$ on the remaining vertices has weight $\Omega(\mu)$ (otherwise the algorithm can trivially label all the remaining vertices as $1$ and be done), which along with the bounds on the weight of vertices in $V \setminus T$ implies that $H_{\sigma_T}$ has bounded relative weights. Since weighted instances with bounded weights are as easy as unweighted instances (we show this in Lemma \ref{lem:dksh-red-2}), we can solve the induced instances with approximation guarantee comparable to unweighted $\dksh$.

In the remainder of the section, we prove Theorem \ref{thm:dksh-weight} by analyzing Algorithm \ref{alg:ref-1}.

\subsection{Proof of Theorem \ref{thm:dksh-weight}}

We begin with a remark that sets up notation and conventions used in the analysis of the algorithm.

\begin{remark}[Weights in $H_{\sigma_T}$]			\label{rem:weights}
	Before we proceed with the analysis of Algorithm \ref{alg:ref-1}, we point out that the construction of $H_{\sigma_T}$ from the partial assignment $\sigma_T$ induces the following (possibly non-uniform) distribution over edges: draw a random edge $e \sim E$ conditioned on $\sigma_{T}(e|_{T}) = {\bf 1}$ and output $e|_{T^c}$. In particular, note that the resulting distribution may assign non-zero mass to the empty set edge in $V \setminus T$ -- we will treat such edges as being trivially satisfied. 
	
	Therefore, we use the notation $e \sim H_{\sigma_T}$ to denote the random draw from the above distribution, so as to distinguish from uniformly random sampling. Finally, it is useful to note that the value of a partial labeling $\pi:V \setminus T \to \{0,1\}$ on $H_{\sigma_T}$ is defined with respect to a random draw $e \sim H_{\sigma_T}$ instead of a uniformly random draw i.e., 
	\begin{equation}			\label{eqn:val-def}
		{\sf Val}_{\pi}(H_{\sigma_T}) = \Pr_{e \sim H_{\sigma_T}} \left[\pi(e) = {\bf 1}\right] = \Pr_{e \sim E| \sigma_T(e|_T) = {\bf 1}} \left[\pi(e|_{T^c}) = {\bf 1}\right].
	\end{equation}
\end{remark}

Let $\sigma^*:V \to \{0,1\}$ be the optimal labeling of weight at most $\mu$ achieving ${\sf Val}_{\leq \mu}(H)$ for $H$. We write $\sigma^* = \sigma^*_T \circ \pi^*_T$ where $\sigma^*_T:T \to \{0,1\}$ and $\pi^*_T:V \setminus T \to \{0,1\}$ are the restrictions of the labeling $\sigma^*$ to vertex sets $T$ and  $V \setminus T$ respectively. Since $w(\sigma^*_T) \leq w(\sigma^*) \leq \mu$, the {\em for loop} block (Lines \ref{step:for-1}-\ref{step:for-2}) will consider the iteration $\sigma_T = \sigma^*_T$. The rest of the proof will establish that the  $\sigma^*_T$-iteration will produce a labeling $\pi'_{\sigma^*_T}$ that satisfies the approximation guarantees claimed by the theorem statement -- this will suffice since the final labeling $\tilde{\pi}$ output by the algorithm is guaranteed to be at least as good as $\pi'_{\sigma^*_T}$. Overall, we will break the analysis into a set of cases (Lemmas \ref{lem:red-1}, \ref{lem:red-2}, and \ref{lem:red-3}), depending on whether $w(V \setminus T)$ is small or large, and then depending on whether the optimal labeling $\sigma^*$ covers a large fraction of edges induced on $T$ -- for each of these, we will show that $\pi'_{\sigma^*_T}$ is a labeling that satisfies the required guarantees of the theorem.

\begin{lemma}				\label{lem:red-1}
	Suppose $w(V \setminus T) < \mu \eta$. Then the labeling $\pi'_{\sigma^*_T}$ has relative weight at most $\mu(1 + \eta)$ and satisfies:
	\[
	{\sf Val}_{\pi'_{\sigma^*_T}}(H) \geq {\sf Val}_{\leq \mu}(H).
	\]
\end{lemma}
\begin{proof}
	Since $w(V \setminus T) < \mu \eta$, this case is addressed by lines \ref{step:for-1} - \ref{step:for-a} in the algorithm. Here Algorithm \ref{alg:ref-1} sets all the variables in $V \setminus T$ to $1$ in the partial labeling $\pi_{\sigma^*_T}$. Therefore, using the definition $\pi'_{\sigma^*_T} = \sigma^*_T \circ \pi_{\sigma^*_T}$, we have ${\rm supp}(\pi'_{\sigma^*_T}) \supseteq {\rm supp}(\sigma^*)$ and hence 
	\begin{equation}		\label{eqn:apprx1}
	{\sf Val}_{\pi'_{\sigma^*_T}}(H) \geq {\sf Val}_{\sigma^*}(H) = {\sf Val}_{\leq \mu}(H)
	\end{equation}
	On the other hand, note that we also have 
	\[
	w(\pi'_{\sigma^*_T}) = w(\sigma^*_T) + w(\pi_{\sigma^*_T}) \leq w(\sigma^*_T) +  w(V \setminus T) \leq  \mu(1 + \eta)
	\] 
	which along with \eqref{eqn:apprx1} concludes the proof.
\end{proof}

The above claim addresses the case $w(V\setminus T) < \mu \eta$, in the remainder, we shall analyze the case $w(V\setminus T) \geq \mu \eta$. Towards that, fixing $\sigma^*$ we first observe: 
\begin{align}
	{\sf Val}_{\leq \mu}(H) 
	&= \Pr_{e \sim E} \left[\sigma^*(e) = {\bf 1}\right] \non\\
	&= \Pr_{e \sim E} \left[\Big\{\sigma^*(e) = {\bf 1}\Big\} \wedge \Big\{e \subseteq T\Big\}\right] + \Pr_{e \sim E} \left[\Big\{\sigma^*(e) = {\bf 1}\Big\} \wedge \Big\{e \not\subseteq T\Big\}\right] \non	\\
	&= \Pr_{e \sim E} \left[\Big\{\sigma^*_T(e|_T) = {\bf 1}\Big\} \wedge \Big\{e \subseteq T\Big\}\right] + \Pr_{e \sim E} \left[\Big\{\sigma^*(e) = {\bf 1}\Big\} \wedge \Big\{e \not\subseteq T\Big\}\right], 	
	\label{eqn:red-rhs}
\end{align}
where in the RHS \eqref{eqn:red-rhs}, the first term corresponds to the set of edges induced by $T$ that are covered by $\sigma^*_T$, and the second term corresponds to the remaining edges that are covered by $\sigma^*$. Now, by averaging, at least one of the two probability terms in \eqref{eqn:red-rhs} is at least ${\sf Val}_{\leq \mu}(H)/2$. We address these cases in the next couple of lemmas. 

\begin{lemma}			\label{lem:red-2}
	Suppose $w(V \setminus T) \geq \mu \eta$ and 
	\[
	\Pr_{e \sim E} \left[\Big\{\sigma^*_T(e|_T) = {\bf 1}\Big\} \wedge \Big\{e \subseteq T\Big\}\right] \geq \frac{{\sf Val}_{\leq \mu}(H)}{2}.
	\]
	Then,  $\pi'_{\sigma^*_T}$ is a labeling with $w(\pi'_{\sigma^*_T}) \leq \mu(1 + \eta)$ satisfying
	\[
	{\sf Val}_{\pi'_{\sigma^*_T}}(H') \geq \frac{{\sf Val}_{\leq \mu}(H)}{2}.
	\]
\end{lemma}
\begin{proof}
Note that since $w(V \setminus T) \geq \mu \eta$, this case is addressed using Lines $7$-$8$ of the algorithm. Furthermore $\pi'_{\sigma^*_T} = \sigma^*_T \circ \pi_{\sigma^*_T}$, where $\pi_{\sigma^*_T}$ is guaranteed to satisfy $w_{\sigma^*_T}(\pi_{\sigma^*_T}) \leq \delta_{\sigma^*_T}$. Hence, 
\begin{equation}				\label{eqn:w-bound}
	w(\pi'_{\sigma^*_T}) = w(\sigma^*_T) + w(\pi_{\sigma^*_T}) \leq w(\sigma^*_T) + w(V \setminus T) \cdot \delta_{\sigma^*_T} \leq \mu(1 + \eta), 
\end{equation}
where the first inequality is using the fact $w(\pi_{\sigma^*_T}) = w_{\sigma^*_T}(\pi_{\sigma^*_T})w(V \setminus T) \leq \delta_{\sigma^*_T}w(V\setminus T)$, and the second inequality is using the definition of $\delta_{\sigma^*_T}$ in \eqref{eqn:del-def}. On the other hand,
\begin{align}
	{\sf Val}_{\pi'_{\sigma^*_T}}(H) 
	&= \Pr_{e \sim E} \left[\pi'_{\sigma^*_T}(e) = {\bf 1}\right] 			\non\\
	&= \Pr_{e \sim E} \left[\left\{\sigma^*_T(e|_T) = {\bf 1} \right\}\wedge \left\{\pi_{\sigma^*_T}(e|_{T^c}) = {\bf 1}\right\}\right] 		\non\\
	&\geq \Pr_{e \sim E} \left[\left\{\sigma^*_T(e|_T) = {\bf 1} \right\}\wedge \left\{\pi_{\sigma^*_T}(e|_{T^c}) = {\bf 1}\right\} \wedge \left\{e \subseteq T\right\}\right] 			\non\\
	&= \Pr_{e \sim E} \left[\left\{\sigma^*_T(e|_T) = {\bf 1} \right\}\wedge \left\{e \subseteq T\right\}\right] 		\non\\
	&\geq \frac12 {\sf Val}_{\leq \mu}(H),					\label{eqn:case-1}
\end{align}
where the last step is due to the setting of the lemma. Now we argue the weight bound for $\pi'_{\sigma^*_T}$. Since \eqref{eqn:w-bound} and \eqref{eqn:case-1} establish the guarantees of the lemma, we are done.
\end{proof}

Now we have the final lemma which deals with the case where $w(V \setminus T) \geq \mu \eta$ and the second term of \eqref{eqn:red-rhs} is large.

\begin{lemma}			\label{lem:red-3}
	Suppose $w(V \setminus T) \geq \mu \eta$ and 
	\begin{equation}				\label{eqn:set-1}
	\Pr_{e \sim E} \left[\Big\{\sigma^*(e) = {\bf 1}\Big\} \wedge \Big\{e \not\subseteq T\Big\}\right] 	\geq \frac{{\sf Val}_{\leq \mu}(H)}{2}.
	\end{equation}
	Then,  $\pi'_{\sigma^*_T}$ is a labeling with $w(\pi'_{\sigma^*_T}) \leq \mu(1 + \eta)$ satisfying
	\[
	{\sf Val}_{\pi'_{\sigma^*_T}}(H') \gtrsim \eta^r\alpha^* \frac{{\sf Val}_{\leq \mu}(H)}{2}.	
	\]
\end{lemma}
\begin{proof}
	Note that this case is again handled using Lines $7$-$8$ fo the algorithm. Our first step here is to show that the optimal value of the induced instance $H_{\sigma^*_T}$ is large. 
	\begin{claim}				\label{cl:h-bound}
		Suppose \eqref{eqn:set-1} holds. Then,
		\[
		{\sf Val}_{\leq \delta_{\sigma^*_T}}(H_{\sigma^*_T}) \geq \frac{{\sf Val}_{\leq \mu}(H)}{2\Pr_{e \sim E}\left[\sigma^*_T(e|_T) = {\bf 1}\right]},
		\]
		where $\delta_{\sigma^*_T}$ is defined as in \eqref{eqn:del-def}.
	\end{claim}
	\begin{proof}
		Let $E'$ be the subset of edges $e \in E$ satisfying $\sigma^*_T(e|_T) = {\bf 1}$ and let $\cE$ denote the event that $e \in E'$ for a random draw of $e \sim E$. Note that $\sigma^*$ has weight at most $\mu$ which implies that $w(\pi^*_T) + w(\sigma^*_T) \leq \mu$ and hence
		\[
		w_{\sigma^*_T}(\pi^*_T) = \frac{w(\pi^*_T)}{w(V \setminus T)} \leq \frac{\mu - w(\sigma^*_T)}{w(V \setminus T)} \leq \delta_{\sigma^*_T}, \qquad\qquad\textnormal{(From \eqref{eqn:del-def})}
		\]
		which along with the observation $\|w_{\sigma^*_T}\|_1 = 1$ implies that $\pi^*_T$ is a labeling of relative weight at most $\delta_{\sigma^*_T}$ in $H_{\sigma^*_T}$. Hence,
		\begin{align*}
			{\sf Val}_{\leq \delta_{\sigma^*_T}}(H_{\sigma^*_T}) \geq 
			{\sf Val}_{\pi^*_T}(H_{\sigma^*_T})
			& = \Pr_{e \sim H_{\sigma^*_T}} \left[\pi^*_T(e) = {\bf 1}\right] \\
			& = \Pr_{e \sim E|\cE} \left[\pi^*_T(e|_{T^c}) = {\bf 1}\right] \tag{Using \eqref{eqn:val-def}} \\
			& = \frac{\Pr_{e \sim E} \left[\big\{\sigma^*_T(e|_T) = {\bf 1} \big\} \wedge \big\{\pi^*_T(e|_{T^c}) = {\bf 1}\big\}\right]}{\Pr_{e \sim E}\left[\cE\right]} \\
			& \geq \frac{\Pr_{e \sim E} \left[\big\{\sigma^*_T(e|_T) = {\bf 1} \big\} \wedge \big\{\pi^*_T(e|_{T^c}) = {\bf 1}\big\} \wedge  \big\{e \not\subseteq T\big\}\right]}{\Pr_{e \sim E}\left[\cE\right]} \\
			& = \frac{\Pr_{e \sim E} \left[\sigma^*(e) = {\bf 1}, e \not\subseteq T\right]}{\Pr_{e \sim E}\left[\cE\right]} \\
			& \geq \frac{{\sf Val}_{\leq \mu}(H)}{2\Pr_{e \sim E}\left[\sigma^*_T(e|_T) = {\bf 1}\right]}	
		\end{align*}
		where the last inequality follows from \eqref{eqn:set-1} and the definition of the event $\cE$.
	\end{proof}
	
	Now recall that in the setting of the lemma we have $w(V \setminus T) \geq \mu \eta$. This along with the definition of $w_{\sigma^*_T}$ and the bound $\eta \geq \mu$ implies that: 
	\[
	\|w_{\sigma^*_T}\|_\infty = \frac{\max_{i \in V \setminus T} w(i)}{w(V \setminus T)} \leq \frac{\mu^{10}}{\mu\eta} \leq \mu^8,
	\]
	i.e, the vertex weights of $H_{\sigma^*_T}$ are $\mu^8$-bounded. On the other hand, we also have: 
	\begin{equation}				\label{eqn:del-bound}
		\delta_{\sigma^*_T} \geq \frac{\mu(1 + \eta) - w(\sigma^*_T)}{w(V \setminus T)} \geq \mu(1 + \eta) - \mu \geq \mu\eta.
	\end{equation}
	Therefore, in Step \ref{step:dksh}, the $\ldksh$ algorithm from Lemma \ref{lem:dksh-red-2} on $H_{\sigma^*_T}$ returns a labeling $\pi_{\sigma^*_T}:V \setminus T \to \{0,1\}$ of relative weight at most $\delta_{\sigma^*_T}$ in $H_{\sigma^*_T}$ satisfying:
	\begin{align}				
		{\sf Val}_{\pi_{\sigma^*_T}}(H_{\sigma^*_T}) 	
		&\geq \alpha^{\uw}_{(\delta_{\sigma^*_T})}(\dksh_r) \cdot {\sf Val}_{\leq \delta_{\sigma^*_T}}(H_{\sigma^*_T}) \tag{Lemma \ref{lem:dksh-red-2}}		\non \\
		&\geq \alpha^{\uw}_{(\delta_{\sigma^*_T})}(\dksh_r) \cdot \frac{{\sf Val}_{\leq \mu}(H)}{2\Pr_{e \sim E}\left[\sigma^*_T(e|_T) = {\bf 1}\right]}  \tag{Claim \ref{cl:h-bound}}		\non \\
		&\gtrsim_r \eta^r \alpha^{\uw}_{(\mu)}(\dksh_r) \cdot \frac{{\sf Val}_{\leq \mu}(H)}{2\Pr_{e \sim E}\left[\sigma^*_T(e|_T) = {\bf 1}\right]}    \label{eqn:pi-val}
	\end{align}
	where the last inequality uses the fact that $\delta_{\sigma^*_T} \geq \mu \eta$ from \eqref{eqn:del-bound} and Lemma \ref{lem:comp}. Then the final concatenated labeling $\pi'_{\sigma^*_T} = \sigma^*_T \circ \pi_{\sigma^*_T}$ satisfies:
	\begin{align}
		{\sf Val}_{\pi'_{\sigma^*_T}}(H)
		& = \Pr_{e \sim E} \left[\pi'_{\sigma^*_T}(e) = {\bf 1}\right]  \non \\
		& = \Pr_{e \sim E} \left[\sigma^*_T(e|_T) = {\bf 1} \wedge \pi_{\sigma^*_T}(e|_{T^c}) = {\bf 1}\right] \non\\
		& = \Pr_{e \sim E} \left[\sigma^*_T(e|_T) = {\bf 1} \right] \Pr_{e \sim E| \sigma^*_T(e|_T) = 1} \left[ \pi_{\sigma^*_T}(e|_{T^c}) = {\bf 1}\right] 		\non\\
		& = \Pr_{e \sim E} \left[\sigma^*_T(e|_T) = {\bf 1} \right] \cdot {\sf Val}_{\pi_{\sigma^*_T}}(H_{\sigma^*_T}) 		\tag{Using \eqref{eqn:val-def}} 		\non\\
		& \gtrsim_r \Pr_{e \sim E} \left[\sigma^*_T(e|_T) = {\bf 1} \right]  \cdot \eta^r \alpha^* \frac{{\sf Val}_{\leq \mu}(H)}{2\Pr_{e \sim E}\left[\sigma^*_T(e|_T) = {\bf 1}\right]}			\tag{Using \eqref{eqn:pi-val}}		\non \\
		& =  \eta^r\alpha^* \frac{{\sf Val}_{\leq \mu}(H)}{2},				\label{eqn:val-f}
	\end{align}
	where $\alpha^* := \alpha^{\uw}_{(\mu)}(\dksh_{r})$. 
	
	Finally, using arguments identical to \eqref{eqn:w-bound} we have that $w(\pi'_{\sigma^*_T}) \leq \mu(1 + \eta)$ which along with \eqref{eqn:val-f} finishes the proof of the lemma.
	
\end{proof}

{\bf Putting Things Together}. Note that the settings of Lemmas \ref{lem:red-1}, \ref{lem:red-2} and \ref{lem:red-3} exhaustively cover all cases possible. In each case, the lemmas establish that the labeling $\pi'_{\sigma^*_T}$ has relative weight at most $\mu(1 + \eta)$ and it satisfies:
\begin{equation}			\label{eqn:val-p}
{\sf Val}_{\pi'_{\sigma^*_T}}(H) \gtrsim_r \eta^r\alpha^* \frac{{\sf Val}_{\leq \mu}(H)}{2}.
\end{equation}
Now we observe that the labeling $\tilde{\pi}$ (from Line \ref{step:final}) returned by the algorithm satisfies 
\begin{align*}
{\sf Val}_{\tilde{\pi}}\left(H\right) 
= \max_{\sigma_T: w(\sigma_T) \leq \mu} {\sf Val}_{\pi'_{\sigma_T}}(H) 
\geq {\sf Val}_{\pi'_{\sigma^*_T}}(H) 
\gtrsim_r \eta^r  \alpha^\uw_{(\mu)}(\dksh_r) \cdot \frac{{\sf Val}_{\leq \mu}(H)}{2},			
\end{align*}
Furthermore, since $\tilde{\pi} = \pi'_{\sigma_T}$ for some $\sigma_T$ satisfying $w(\sigma_T) \leq \mu$, we must have $w(\tilde{\pi}) = w(\sigma_T)\leq \mu(1+\eta)$ where the bound for $w(\sigma_T)$ can be established using arguments identical to Lemmas \ref{lem:red-1}, \ref{lem:red-2} and \ref{lem:red-3}. The above arguments put together conclude the proof of the theorem.

\subsection{$\dksh$ With Bounded Weights}

The following lemma is folklore that reduces $\dksh$ instances with bounded relative weights to $\dksh$ instance with uniform vertex weights.

\begin{lemma}				\label{lem:dksh-red-2}
	For any fixed $\eta \in (\mu^2,1)$, there exists an algorithm which on input $\ldksh$-instances $H = (V,E,w,\mu)$ of arity $r$ with $\|w\|_1 = 1$ and $\|w\|_\infty \leq \mu^{8}$, outputs a labeling $\sigma:V \to \{0,1\}$ of relative weight $\mu'$ satisfying
	\[
	{\sf Val}_{\sigma}(H) \geq \alpha^{\sf uw}_{(\mu)}(\dksh_{r}) \cdot {\sf Val}_{\leq \mu}(H),
	\]
	and $|\mu' - \mu| \leq \mu \eta$.
\end{lemma}
\begin{proof}
	
	For ease of notation, we denote $\alpha^* := \alpha^{\sf uw}_{(\mu)}(\dksh_{r})$. Now consider the algorithm in Algorithm \ref{alg:ref-2}.
	
	\begin{algorithm}[ht!]
		\SetAlgoLined
		
		{\bf Input}. An $\ldksh$-instance $H = (V,E,w,\mu)$ of arity $r$ \;
		Choose $N \in \mathbbm{N}$ to be large enough\footnotemark such that $w(i) \cdot N \in \mathbbm{N}$ for every $i \in V$\;
		For every $i \in V$, denote $\ell(i) := w(i) \cdot N$\;
		Construct a new unweighted $\dksh$ instance $H' = (V',E')$ with $k = \mu|V'|$ as follows		\label{step:cons-1}\; 
		{\bf Vertex Set}. For every vertex $i \in V$, introduce a cloud of vertices  $\cC_i := \{(i,j)\}_{j \in [\ell(i)]}$ corresponding to vertex $i$. The final vertex set is $V' := \cup_{i \in V} \cC_i$\;
		{\bf Edge Set}. For every edge $e = (i_1,\ldots,i_r) \in E$ and $(j_1,\ldots,j_r) \in \times^{r}_{t = 1} [\ell(i_t)]$, introduce edge $\{(i_1,j_1),\ldots,(i_r,j_r)\}$ in $E'$ with weight $1/\prod_{t \in [r]}\ell(i_t)$			\label{step:cons-2}\;
		Run the $\alpha^{\uw}_{(\mu)}(\dksh_r)$-approximation algorithm on $H'$ and let $\sigma':V' \to \{0,1\}$ be the labeling of relative weight $\mu$ returned by the algorithm\;
		Sample $\sigma:V \to \{0,1\}$ as follows: for every $i \in V$ sample $x(i) \sim [\ell(i)]$ independently, and set $\sigma(i) = \sigma'(i,x(i))$ 		\label{step:sig-f}\;
		Return labeling $\sigma$\;
		\caption{Bounded Weights $\dksh$}
		\label{alg:ref-2}
	\end{algorithm} 
	\footnotetext{We assume that the reciprocals of the weights are polynomial in the instance size.}
	In the above construction, note that the new hypergraph is hyperedge weighted as well, and henceforth, we shall use $e \sim H'$ to denote a random draw of a hyperedge according to the hyperedge weights. The following related observation will be useful in the rest of the analysis.
	
	\begin{observation}				\label{obs:h-samp}
		A random draw of an edge $e \sim H'$ can be simulated using the following process:
		\begin{itemize}
			\item Sample edge $(i_1,\ldots,i_r) \sim E$.
			\item Sample $(j_t)_{t \in [r]} \sim \times^r_{t = 1}\left[\ell(i_t)\right]$.
			\item Output edge $\{(i_1,j_1),\ldots,(i_r,j_r)\}$.
		\end{itemize}
		Consequently, the value of an assignment $\sigma':V' \to \{0,1\}$ in $H'$ can be expressed as:
		\[
		{\sf Val}_{\sigma'}(H') = \Ex_{(i_1,\ldots,i_r) \sim E} \Pr_{(j_t)_{t \in [r]} \sim \times^r_{t = 1} [\ell(i_t)]} \left[\bigwedge_{t \in [r]} \sigma'(i_t,j_t) = 1\right].
		\]
	\end{observation}
	\begin{proof}
		Let $\kappa:E' \to \mathbbm{R}^+$ denote the hyperedge relative weight function for $H'$. Then by definition, for any edge $e = \{(i_1,j_1),\ldots,(i_r,j_r)\} \in E'$ the relative weight of $e$ can be re-expressed as: 
		\[
		\kappa(e) = \frac{\frac{1}{\prod_{t \in [r]} \ell(i_t)}}{\sum_{(i'_t)^r_{t = 1} \in E'} \sum_{(j'_t)^r_{t = 1} \in \times_{t \in [r]} [\ell(i'_t)]} \frac{1}{\prod_{t \in [r]} \ell(i'_t)}} = \frac{1}{|E|} \prod_{t \in [r]}\frac{1}{\ell(i_t)},
		\]
		which is exactly the probability with which edge $e$ is sampled according to the process in the statement i.e., sampling according to the relative edge weights is equivalent to sampling according to the above process. The second point of the observation now follows directly from the first.
	\end{proof}

	Towards analyzing Algorithm \ref{alg:ref-2}, we begin with following claim which shows that optimal value of instance $H'$ is at least the optimal value of $H$.
	
	\begin{claim}			\label{cl:Hp-val}
		The instance $H'$ constructed in lines \ref{step:cons-1} - \ref{step:cons-2} satisfies
		\[
		{\sf Val}_{(\mu)}(H') \geq {\sf Val}_{\leq \mu}(H),
		\]
		where ${\sf Val}_{(\mu)}(H')$ is the optimal value of $\dksh$ instance $H'$ achieved using labelings of relative weight exactly $\mu$.
	\end{claim}
	\begin{proof}
		Let $\sigma^*:V \to \{0,1\}$ be an optimal labeling for $H$ of weight at most $\mu$ that attains ${\sf Val}_{\leq \mu}(H)$. Then we construct $\sigma_1:V' \to \{0,1\}$ as 
		\[
		\sigma_1(i,j) :=  \sigma^*(i), \qquad\qquad\forall~i \in V, j \in [\ell(i)].
		\]
		We claim that $\sigma_1$ has relative weight at most $\mu$ in $H'$. To see this, observe that $H'$ has uniform weights and hence we can bound the relative weight of $\sigma_1$ in $H'$ as:
		\begin{align*}
		\mu_1 := \frac{\sum_{a \in V'} \sigma_1(a)}{|V'|} = \frac{\sum_{i \in V} \sum_{j \in [\ell(i)]} \sigma_1(i,j)}{|V'|}
		&=\frac{\sum_{i \in V}\ell(i) \cdot \sigma_1(i)}{\sum_{i \in V}\ell(i)} \\
		&= \frac{\sum_{i \in V} w(i) \sigma^*(i)}{\sum_{i \in V}w(i)}	\tag{since $\ell(i) = w(i) \cdot N$} \\
		&\leq \mu,
		\end{align*}
		which implies that the labeling $\sigma_1$ is of relative weight at most $\mu$ in the instance $H'$. Next, since $H$ is uniformly weighted, we can further extend $\sigma_1$ to a labeling $\sigma_2$ whose relative weight is exactly $\mu$ in $H'$ by setting $(\mu - \mu_1)$-fraction of zeros to ones in $\sigma_1$. Observe that since ${\rm supp}(\sigma_2) \supseteq {\rm supp}(\sigma_1)$, the resulting labeling must satisfy ${\sf Val}_{\sigma_2}(H') \geq {\sf Val}_{\sigma_1}(H')$. Hence,		
		\begin{align*}
			{\sf Val}_{(\mu)}(H')
			\geq {\sf Val}_{\sigma_2}(H') 
			&\geq {\sf Val}_{\sigma_1}(H') \\
			&= \Pr_{e \sim H'}\left[\sigma_1(e) = {\bf 1}\right] 			\\
			& = \Ex_{(i_1,\ldots,i_r) \sim E} ~~ \Pr_{(j_t)_{t \in [r]} \sim \times^r_{t = 1} [\ell(i_t)]} \left[\bigwedge_{t \in [r]}\sigma_1(i_t,j_t) = 1\right] 		\tag{Observation \ref{obs:h-samp}}\\
			& = \Ex_{(i_1,\ldots,i_r) \sim E} ~~ \Pr_{(j_t)_{t \in [r]} \sim \times^r_{t = 1} [\ell(i_t)]} \left[\bigwedge_{t \in [r]}\sigma^*(i_t) = 1\right] 		\tag{Definition of $\sigma_1$}\\
			& = \Pr_{(i_1,\ldots,i_r) \sim E}  \left[\bigwedge_{t \in [r]}\sigma^*(i_t) = 1\right] \\
			& = {\sf Val}_{\leq \mu}(H). 
		\end{align*}
	\end{proof}
	Therefore, the above claim implies that in Line $7$ the algorithm computes the intermediate labeling $\sigma':V' \to \{0,1\}$ of weight $\mu|V'|$ satisfying
	\begin{equation}				\label{eqn:val-2}
		{\sf Val}_{\sigma'}(H') \geq \alpha^* \cdot {\sf Val}_{(\mu)}(H') \geq \alpha^* \cdot {\sf Val}_{\leq \mu}(H),
	\end{equation}
	where the last step is due to Claim \ref{cl:Hp-val} and $\alpha^* := \alpha^{\uw}_{(\mu)}(\dksh_r)$. Now we argue guarantees for the final labeling $\sigma$ (from Line \ref{step:sig-f}) using a couple of claims.
	
	\begin{claim}				\label{cl:red-1}
		The distribution over labeling $\sigma$ satisfies
		\[
		\Ex_{\sigma} \Big[{\sf Val}_\sigma(H)\Big] \geq \alpha^* \cdot {\sf Val}_{\leq \mu}(H)
		\]
	\end{claim}
	\begin{proof}
	Observe that 
	\begin{align}
		&\Ex_{\sigma} \Big[{\sf Val}_\sigma(H)\Big] \\
		&=\Ex_{\sigma} \Ex_{e = (i_1,\ldots,i_r) \sim E}\left[\mathbbm{1}\left(\bigwedge_{t \in [r]} \sigma(i_t) = 1 \right)\right] 	\non\\
		& = \Ex_{e = (i_1,\ldots,i_r) \sim E}\Ex_{\sigma} \left[\mathbbm{1}\left(\bigwedge_{t \in [r]} \sigma(i_t) = 1 \right)\right] 	\non\\
		& = \Ex_{e = (i_1,\ldots,i_r) \sim E} \frac{1}{\prod_{t \in [r]}\ell(i_t)} \sum_{(j_t)_{t \in [r]} \in \times^r_{t = 1} [\ell(i_t)]} \left[\mathbbm{1}\left(\bigwedge_{t \in [r]} \sigma'(i_t,j_t) = 1 \right)\right] 		\tag{Definition of $\sigma'$}		\non\\
		& = \Ex_{e = (i_1,\ldots,i_r) \sim E} \Ex_{(j_t)_{t \in [r]} \sim \times^r_{t = 1}[\ell(i_t)]} \left[\mathbbm{1}\left(\bigwedge_{t \in [r]} \sigma'(i_t,j_t) = 1 \right)\right] 		\non\\
		& = \Ex_{e' \sim H'} \left[\mathbbm{1}\left(\bigwedge_{i \in e'} \sigma'(i) = 1 \right)\right] 	\tag{Observation \ref{obs:h-samp}}	\non\\
		& = {\sf Val}_{\sigma'}(H') 		\non\\
		& \geq \alpha^* \cdot {\sf Val}_{\leq \mu}(H),		\label{eqn:apprx} 
	\end{align}
	where the last step is due to \eqref{eqn:val-2}. 
	\end{proof}
	Next we show that $\sigma$ has relative weight close to $\mu$ with high probability.
	\begin{claim}			\label{cl:red-2}
	With probability at least $1 - e^{-0.5/\mu^2}$ we have $|\mu - w(\sigma) | \leq \mu \eta$.
	\end{claim} 
	\begin{proof}
	We use the definition of the distribution over $\sigma$ (given $\sigma'$) and observe:
	\begin{align*}
		\Ex_{\sigma}\Big[w(\sigma)\Big] = \sum_{i \in V}\Ex_{\sigma}\left[w(i)\sigma(i)\right] = 
		\sum_{i \in V}  w(i)\left(\frac{\sum_{j \in [\ell(i)]} \sigma'(i,j) }{\ell(i)}\right) 
		\overset{1}{=}  \sum_{i \in V} \frac{\sum_{j \in [\ell(i)]}  \sigma'(i,j) }{N} 
		\overset{2}{=} \mu,	
	\end{align*}
	where step $1$ uses $\ell(i) = w(i) \cdot N$ by definition and step $2$ follows from the guarantee on $\sigma'$. Furthermore, $\sigma(1),\ldots,\sigma(|V|)$ are independent $\{0,1\}$ random variables. Hence using Hoeffding's inequality we get that 
	\begin{align}
		\Pr_\sigma\left[ \left|\sum_{i \in V}w(i)\sigma(i) - \mu \right| > \mu\eta\right]
		&\leq \exp\left(-\frac{\mu^2\eta^2}{2\sum_{i \in V} w(i)^2}\right) 		\non\\
		&\leq \exp\left(-\frac{\mu^2\eta^2}{2\mu^{8}\sum_{i \in V} w(i)}\right) 		\tag{Since $\|w\|_\infty \leq \mu^{8}$}	\non\\
		&= \exp\left(-\frac{\eta^2}{2\mu^{6}}\right)									\tag{since $\|w\|_1  = 1$} \non \\
		&\leq \exp\left(-\frac{1}{2\mu^{2}}\right) 				\label{eqn:weight}
	\end{align}
	where the last step is due to $\eta \geq \mu^2$ in the setting of the lemma. 
	\end{proof}
	
	{\bf Combining the Guarantees}. Note that Claim \ref{cl:red-1} implies that
	\[
	\Ex_{\sigma}\left[{\sf Val}_{\sigma}(H)\right] \geq \alpha^* {\sf Val}_{\leq \mu}(H),
	\]
	and ${\sf Val}_{(\mu)}(H) \geq \mu^r$. Then by averaging, for any constant $\epsilon \in (0,1)$, with probability at least $\epsilon\alpha^* \mu^r$ we have 
	\[
	{\sf Val}_{\sigma}(H) \geq (1 - \epsilon)\alpha^*\cdot{\sf Val}_{\leq \mu}(H).
	\]
	On the other hand, Claim \ref{cl:red-2} implies that with probability at least $1 - e^{-0.5(1/\mu)^2}$. the labeling $\sigma'$ has weight in $[\mu(1 - \eta),\mu(1 + \eta)]$. Hence, with probability at least $\epsilon\alpha^* \mu^r/2$, the algorithm returns a labeling $\sigma$ which has weight $\mu(1 \pm \eta)$ with value at least $(1 - \epsilon)\alpha^* {\sf Val}_{(\mu)}(H)$. Therefore repeating the algorithm for at least $(1/\epsilon\mu)^r$-independent rounds returns a labeling with the desired guarantees w.h.p.
	
\end{proof}

\subsection{Proofs of Theorem \ref{thm:pred-cond} and \ref{thm:pred-approx}}

We now prove Theorems \ref{thm:pred-cond} and \ref{thm:pred-approx} using the lemmas proved from the previous sections. Before we proceed, we shall need the following additional elementary observations.

\begin{observation}			\label{obs:triv}
	There exists a $1$-approximation algorithm for $\dksh_0$.
\end{observation}
\begin{proof}
	$\dksh_0$ is the trivial CSP whose constraint set consists of empty set edges which are trivially satisfied by any labeling, and hence it admits a $1$-approximation algorithm.
\end{proof}

\begin{observation}				\label{obs:cover}
	There exists a $(1 - 1/e)$-approximation algorithm for uniformly weighted $\dksh_1$ instances for any $\mu$. Consequently,
	\[
	\alpha^{\uw}_{(\mu)}(\dksh_1) \geq 1 - 1/e.
	\] 
\end{observation}
\begin{proof}
	Let $H(V,E)$ be an instance of $\dksh_1$. Then $E$ is a multi-set of singleton subsets of $V$ i.e., each element of $e \in E$ is actually a vertex of $V$. Let $w:V \to \mathbbm{Z}_+$ denote the function that maps $v \in V$ to the number of occurrences of $v$ in $E$. Then note that the extension of the function $w:2^{V} \to \mathbbm{Z}_{+}$ is submodular. Furthermore, the optimization problem corresponding to $\dksh_1$ is equivalent to: 
	\begin{eqnarray*}
		\textnormal{Maximize} 		& 	w(S) \\
		\textnormal{Subject to}		& 	|S| = k			   
	\end{eqnarray*} 
	Since $w$ is a submodular, the above is an instance of submodular function maximization subject to cardinality constraints, for which the greedy algorithm is known to yield a $(1 - 1/e)$-approximation guarantee (Exercise 2.10~\cite{shmoys-williamson}).
\end{proof}

Using the above observation, now we establish Theorems \ref{thm:pred-cond} and \ref{thm:pred-approx}.

\begin{proof}[Proof of Theorem \ref{thm:pred-cond}] 
Recall that $\alpha_{\leq \mu}(\psi)$ is the optimal approximation guarantee for biased CSP instances with predicate $\psi$ and bias constraint $\mu$. Now, suppose $\cM_{\psi} \subseteq \cS_0 \cup \cS_1$. Then using Lemma \ref{lem:at-most-0} we have 
\begin{align*}
\alpha_{\leq \mu}(\Psi) 
&  \gtrsim_r \min_{\beta \in \cM_{\psi}}\alpha^{\uw}_{(\mu)}\left(\dksh_{\|\beta\|_0}\right) \\
&  \geq \min \left\{\alpha^\uw_{(\mu)}\left(\dksh_{0}\right),\alpha_{(\mu)}\left(\dksh_{1}\right)\right\} \\
& \geq \min\{1,1 - 1/e\} = 1 - 1/e,
\end{align*}
where the penultimate inequality uses Observations \ref{obs:triv} and \ref{obs:cover}. On the other hand, suppose $\cM_\psi \not\subset \cS_0 \cup \cS_1$. Then, fix a $\beta \in \cM_\psi \setminus (\cS_0 \cup \cS_1)$. Furthermore, let $\alpha^{\rm bound}_{\leq\mu}(\dksh_i)$ denote the bias approximation curve for $\dksh$ instances whose relative weights are bounded by $\mu^{10}$ (as in the setting of Lemma \ref{lem:dksh-red-2}). Then, 
\begin{align*}
	\lim_{\mu \to 0} \alpha_{\leq \mu}(\psi) 
	&\lesssim_r \lim_{\mu \to 0} \alpha^\uw_{\leq r'\mu}\left(\dksh_{\|\beta\|_0}\right) 		\tag{Lemma \ref{lem:pred-hard} where $r' \in [2,r]$} \\
	&= \lim_{\mu \to 0} \alpha^{\sf bound}_{\leq r'\mu}\left(\dksh_{\|\beta\|_0}\right) 		\tag{Lemma \ref{lem:dksh-red-2}}\\
	&= 0,					
\end{align*}
where the last step uses Theorem \ref{thm:dksh-hard}.
\end{proof}

\begin{proof}[Proof of Theorem \ref{thm:pred-cond}]
	Firstly, using Lemma \ref{lem:at-most-0} we have 
	\begin{equation}			\label{eqn:pred-0}
	\alpha_{\leq \mu}(\psi) \gtrsim_r \min_{\beta \in \cM_\psi}\alpha^{\uw}_{(\mu)}\left(\dksh_{\|\beta\|_0}\right).
	\end{equation}
	On the other hand, using Lemma \ref{lem:pred-hard}, we know that for every $\beta \in \cM_\psi$, there exists $r_{\beta} \in [2,r]$ for which we have 
	\begin{equation}			\label{eqn:pred-1}
	\alpha_{\leq \mu}(\psi) \lesssim_r \alpha^\uw_{(r_{\beta}\mu)}\left(\dksh_{\|\beta\|_0}\right)
	\lesssim_r \alpha^\uw_{(\mu)}\left(\dksh_{\|\beta\|_0}\right)
	\end{equation}
	where the last inequality follows using the comparison lemma (Lemma \ref{lem:comp}).  Therefore, using \eqref{eqn:pred-1} for every $\beta \in \cM_\psi$ we get that 
	\begin{equation}				\label{eqn:pred-3}
	\alpha_{\leq \mu}(\psi) \lesssim_r \min_{\beta \in \cM_\psi}\alpha^\uw_{(\mu)}\left(\dksh_{\|\beta\|_0}\right).
	\end{equation}
	Combining the bounds from \eqref{eqn:pred-0} and \eqref{eqn:pred-3} gives us both directions of the desired inequality and hence concludes the proof of the theorem.
\end{proof}

\part{Hardness of Approximation}

\section{Additional Technical Preliminaries}

In this section we introduce the technical preliminaries used in our hardness reductions and their analysis. For the most part, we will follow the notation from \cite{Mos10} for the various concepts discussed here.

\subsection{Correlated Probability Spaces}

A finite probability space $(\Omega,\gamma)$ is identified by a finite set $\Omega$ with a measure $\gamma$ over the elements in $\Omega$. A joint correlated probability space over the product set $\prod^r_{i = 1} \Omega_i$ will be denoted as $(\prod^r_{i = 1} \Omega_i,\gamma)$, where $\gamma$ is now a measure over the elements from the Cartesian product set $\prod^r_{i = 1} \Omega_i$. If the sets are all identical, i.e., $\Omega_1 = \cdots = \Omega_r = \Omega$, then for simplicity, we shall denote $\prod^r_{i = 1} \Omega_i = \Omega^r$. 

{\bf Operations on Probability Spaces}. Given a pair of probability spaces $(\Omega_1,\gamma_1),(\Omega_2,\gamma_2)$, we use $(\Omega_1,\gamma_1) \otimes (\Omega_2,\gamma_2) = (\prod_{i = 1,2} \Omega_i, \gamma_1 \otimes \gamma_2)$ to denote the corresponding product probability space, where for any element $(\omega_1,\omega_2) \in \Omega_1 \times \Omega_2$, we associate the measure $(\gamma_1 \otimes \gamma_2)(\omega_1,\omega_2) = \gamma_1(\omega_1)\gamma_2(\omega_2)$. Furthermore, when $(\Omega_1,\gamma_1) = (\Omega_2,\gamma_2) = (\Omega,\gamma)$, for simplicity we denote the corresponding product measure as $(\Omega^2,\gamma^2)$. These conventions are naturally extended to $R$-wise product spaces. 

{\bf Correlation}. Given a pair of correlated spaces $(\prod_{i = 1,2} \Omega_i,\bgamma)$, the correlation between $\Omega_1$ and $\Omega_2$ induced by measure $\bgamma$ is defined as 
\[
\rho(\Omega_1,\Omega_2;\bgamma) = \max_{\substack{f_i \in L_2(\Omega_i) \\ \|f_i\|_2 = 1}} {\rm Cov}_{\bgamma}(f,g) 
\]
We can extend the above definition into $r$-ary correlated spaces $(\prod_{i \in [r]}\Omega_i,\bgamma)$ as 
\[
\rho(\Omega_1,\ldots,\Omega_r;\bgamma) = \max_{i \in [r]} \rho\left( \prod_{j \in [r] \setminus \{i\}} \Omega_j,\Omega_i;\bgamma\right)  
\]
We shall also need the following natural family of pairwise $\rho$-correlated distributions:
\begin{definition}[Distribution $\cA_{r,\rho}(\Omega,\gamma)$]			\label{defn:dist}
	Given a probability space $(\Omega,\gamma)$, the distribution $\cA_{r,\rho}(\Omega,\gamma)$ is the distribution over random variables  $(\omega_1,\ldots,\omega_r)$ supported over $\Omega^r$, generated using the following process.
	\begin{itemize}
		\item With probability $\rho$, sample $\omega \sim \gamma$ and set $\omega_i = \omega$ for every $i \in [r]$.
		\item With probability $1 - \rho$, sample $\omega_i \sim \gamma$ independently for every $i \in [r]$.
	\end{itemize}
\end{definition}

\subsection{Fourier Analysis}				\label{sec:fourier}

Let $(\Omega,\gamma)$ be a finite probability space. Then it is well known that the set of real valued function $f:\Omega \to \mathbbm{R}$ forms a vector space with respect to the usual addition of functions and scalar multiplication. Furthermore, one can equip the vector space with the inner product with respect to the measure $\gamma$ as 
\[
\langle f,g \rangle_\gamma \defeq \Ex_{x \sim \gamma}\left[f(x)g(x)\right]. 
\]
Therefore, given the above setup, one can construct an orthonormal basis for the vector space of such functions. In particular, we will be interested in the so called {\em Fourier Basis} which we define formally below.

\begin{definition}[Fourier Basis~\cite{OD14}]
	For any probability space $(\Omega,\gamma)$, there exists an orthonormal basis $\phi_0,\phi_1,\ldots,\phi_{|\Omega| - 1}: \Omega \to \mathbbm{R}$  with $\phi_0 \equiv 1$ under which any function $f \in L^2(\Omega,\gamma)$ can be uniquely expressed as 
	\[
	f = \sum_{\sigma = 0}^{|\Omega| - 1} \wh{f}(\sigma)\phi_\sigma.
	\]
	Here $\{\wh{f}(\sigma)\}$ are referred to as the Fourier coefficients of the function $f$. 
\end{definition} 

Furthermore, one can naturally extend the above to the setting of product probability spaces. For brevity denote $\ell = |\Omega| - 1$ and $\mathbbm{Z}_{\leq \ell} := \{0,1,\ldots,\ell\}$. Then given a product probability space $(\Omega^t, \gamma^{t})$, the corresponding Fourier basis is identified with $\{\phi_\sigma\}_{\sigma \in \mathbbm{Z}^t_{\leq \ell}}$ where for a given $\sigma \in \mathbbm{Z}^t_{\leq \ell}$, we have  
\[
\phi_{\sigma}(x) \defeq \prod_{i \in [t]} \phi_{i}(x_i) \ \ \ \ \forall x \in \Omega^t.
\]
Analogously, any function $f \in L^2(\Omega^t,\gamma^t)$ admits a unique representation in the Fourier basis.
\[
f = \sum_{\sigma \in \mathbbm{Z}^{t}_{\leq \ell}} \wh{f}(\sigma) \phi_\sigma.
\]

{\bf Influence}. Given a function $f \in L^2(\Omega^t,\gamma^t)$, the influence of the $i^{th}$ coordinate -- denoted by $\Inf{i}{f}$ -- if formally defined as
\begin{equation}				\label{eqn:inf-defn}
\Inf{i}{f} \defeq \Ex_{x \sim \gamma^t}\left[{\rm Var}_i[f]\right],
\end{equation}
i.e., it measures the dependence of the function $f$ on the coordinate $i$. It is well known that influence admits the following closed form expression
\[
\Inf{i}{f} = \sum_{\sigma : \sigma(i) \neq 0} \wh{f}(\sigma)^2
\]

{\bf Correlated Sampling and Noise Operators}. Given an element $x \in \Omega^t$ and $\rho \in (0,1)$, a vector $y$ is a $\rho$-correlated copy of $x$ in the space $(\Omega,\gamma)$ if it is generated using the following process. For every $i \in [t]$, do the following independently.
\begin{enumerate}
\item W.p. $\rho$, set $y_i = x_i$.
\item W.p. $1 - \rho$, sample $y_i \sim \gamma$.
\end{enumerate}
We will denote the above sampling process as $y \underset{\rho}{\sim } x$. With this, we are ready to define the natural family of noise operators on $L^2(\Omega^t,\gamma^t)$.
\begin{definition}[Noise Operator $T_\rho$]
For any $\rho \in (0,1)$, the noise operator $T_\rho$ is a stochastic functional on $L^2(\Omega^t,\gamma^t)$ which is defined as follows.
\[
T_\rho f(x) = \Ex_{y \underset{\rho}{\sim} x} \left[f(x)\right],    \qquad\qquad\forall x \in \Omega^t.
\]
\end{definition}
The following property of the noise operator is well known (e.g., \cite{KKMO07}).

\begin{lemma}[Bounded Influential Coordinates]			\label{lem:inf-size}
	For any function $f: \Omega^t \to [0,1]$ and $\eta,\tau \in (0,1)$ we have 
	\[
	\left|\left\{i \in [t] : \Inf{i}{T_{1 - \eta}f} \geq \tau \right\} \right| \leq \frac{1}{\eta \tau}.
	\]
\end{lemma}

{\bf Noise Stability Bounds}. We shall use the following noise stability bound for our soundness analysis.

\begin{restatable}{rethm}{stab}			\label{thm:ks-stab}
	Let $(\Omega,\gamma)$ be a finite probability space and let $(\Omega^r,\bgamma)$ be the $r$-ary correlated probability space corresponding to distribution $\cA_{r,\rho}(\Omega,\gamma)$ (as in Definition \ref{defn:dist}). Let $\alpha := \min_{\omega \in \Omega^r}\bgamma(\omega)$. Then for every $\nu \in (0,1)$ there exists $\tau = \tau(\nu,r,\alpha)$ such that the following holds. Let $f:\Omega^R \to [0,1]$ be a function in $L_2(\Omega^R,\gamma^R)$ satisfying
	\[
	\max_{i \in [R]} \Inf{i}{f} \leq \tau.
	\]
	Furthermore, let $\mu:= \Ex_{\bomega\sim\gamma^R}\left[f(\bomega)\right]$ satisfy $\mu \leq 2^{-r}$. Then for any $\rho \leq 1/(C'r^2\log(1/\mu))$ we have 
	\[
	\Ex_{({\bomega}_1,\ldots,{\bomega}_r) \sim \bgamma^R} \left[\prod_{i \in [r]} f(\bomega_i)\right] \leq 3\mu^r + \nu,
	\]
	where $C'>0$ is an absolute constant.
\end{restatable}

The above bound was established for the $R$-ary hypercube in \cite{KS15}, however we need the above general version for our application to Theorem \ref{thm:dksh-hard}, although the bound follows as is using the techniques of \cite{KS15}. We include a proof of it in Section \ref{sec:stab} for the sake of completeness.

\subsection{Real Extensions for functions $[R]^t \to [R]$}			\label{sec:r-ary}

Let $F:[R]^t \to [R]$ be a function define on $R$-ary cube. Then, as is standard, one may equivalently express $F$ as $f:[R]^t \to \Delta_R$ where $\Delta_R$ is the $R$-simplex, such that for any $x \in [R]^t$ we have $f(x) = e_{F(x)}$; here $e_i$ is the $i^{th}$ standard basis vector for any $i \in [r]$. Furthermore, one can interpret $f = (f^{(1)},\ldots,f^{(R)})$ as a vector function, where $f^{(j)}$ is the $j^{th}$ coordinate function. 

{\bf Folding over $[R]^t$.} Given long code $f:[R]^t \to [R]$, we define the folded long code $\tilde{f}:[R]^t \to [R]$ as 
\begin{equation}		\label{eqn:folded}
\tilde{f}(x) \defeq {f}(x \oplus_R (-x_1 \cdot {\bf e}_1)) + x_1 
\end{equation}
The following properties of folded long codes are well known.

\begin{proposition}(for e.g., see~\cite{KKMO07})				\label{prop:fold}
	The following properties hold for long codes defined over $[R]^t$. 
	\begin{itemize}
		\item If $f:[R]^t \to [R]$ is a dictator function, then $f$ is folded i.e., $f = \tilde{f}$.
		\item If $f$ is folded, then $\Ex_{x \sim [R]^t}\left[f(x)\right] = 1/R$.
	\end{itemize}
\end{proposition}

\section{Hardness for D$k$SH}

Our hardness result is established using a factor preserving reduction from \smallsetexpansion~which we define formally below. Given a regular graph $G = (V,E)$, and a subset $S \subset V$, the edge expansion of $S$ in $G$, denoted by $\phi_G(S)$, is defined as 
\[
\phi_G(S) := \frac{\Pr_{(i,j) \sim G}\big[i \in S, j \notin S\big]}{\min\{{\sf Vol}(S),{\sf Vol}(S^c)\}}
\]
where ${\rm Vol}(S)$ is the weight of the set $S$ with respect to the stationary measure of a random walk in $G$. The \smallsetexpansion~problem in our setup is defined as follows.

\begin{definition}[$(\epsilon,\delta,M)$-SSE]
	An $(\epsilon,\delta,M)$-SSE instance is characterized by a {\em regular} graph $G = (V,E)$ where the objective is to distinguish between the following two cases.
	\begin{itemize}
		\item {\bf YES Case}. There exists $S \subset V$ such that ${\sf Vol}(S)= \delta$ and $\phi_G(S) \leq \epsilon$.
		\item {\bf NO Case}. For every subset $S \subset V$ such that ${\sf Vol}(S) \in \left[\frac{\delta}{M},M \delta\right]$ we have $\phi_G(S) \geq 1 - \epsilon$.
	\end{itemize}
\end{definition}

We shall use the instances given by the following version of SSEH as our starting point.

\begin{conjecture}[Small Set Expansion Hypothesis~\cite{RST12}]					\label{conj:sseh}
	For every $\epsilon \in (0,1), M \in [1,1/\sqrt{\epsilon}]$, there exists $\delta = \delta(\epsilon,M)$ such that $(\epsilon,\delta,M)$-SSE is \NP-hard.
\end{conjecture}

Our main result here is the following factor preserving reduction from \smallsetexpansion~to D$k$SH.

\begin{theorem}						\label{thm:dksh-redn}
	Let $r \geq 2$ be an integer and $\mu \in (0,1)$ be such that $\mu < 2^{-r}$ . Then there exists $M = M(\mu,r)$ and $\epsilon = \epsilon(\mu,r)$, depending only on $\mu$ and $r$ for which the following holds. Let $\delta = \delta(\epsilon,M)$ be as in Conjecture \ref{conj:sseh}. Then there exists a polynomial time reduction from a $(\epsilon,\delta,M)$-SSE instance $G = (V,E)$ to instance $H = (V_H,E_H,w)$ of D$k$SH such that the following holds:
	\begin{itemize}
		\item {\bf Completeness}. If $G$ is a YES instance, then there exists a set $S \subset V_{H}$ such that ${\sf Vol}(S) = \mu$ and $w(E_H(S)) \geq C_1 r^{-3} \mu\log(1/\mu)$.
		\item {\bf Soundness}. If $G$ is a NO instance, then for every set $S \subset V$ such that ${\sf Vol}(S) = \mu$ and $w(E_G[S]) \leq 4\mu^{r}$.
	\end{itemize}
	where $C_1 > 0$ is an absolute constant independent of $r$ and $\mu$.   
\end{theorem}

\subsection{The PCP Verifier for D$k$SH Hardness}

Let $G = (V,E)$ be a $(\epsilon,\delta,M)$-SSE instance (as in Conjecture \ref{conj:sseh}) where parameters $\epsilon,M$ are set later (see below Figure \ref{fig:dksh-test}). Before we describe our reduction, we need to introduce some additional noise operators used in the reduction.

\begin{itemize}
	\item {\bf Graph Walk Operator}. For a given graph $G = (V,E)$ and a vertex $A \in V$, we sample $B \sim G_{\eta}(A)$ as follows. W.p $1 - \eta$, we sample $B$ by performing a $1$-step random walk on $G$ from $A$, and with probability $\eta$ we sample $B$ from the stationary distribution of the random walk on $G$.
	\item {\bf Noisy Leakage Operator}. For $z \in \{\bot,\top\}$ and $(A,x) \in V \times \{0,1\}_\mu$, we sample $(A',x') \sim M_z(A,x)$ as follows. If $z = \top$ we set $(A',x') = (A,x)$, otherwise we sample $(A',x') \sim V \times \{0,1\}_\mu$.
\end{itemize}
Furthermore, for any integer $R$ and any $z \in \{\bot,\top\}^R$ we define $G^{\otimes R}_\eta$ and $M_z$ to be the corresponding $R$-wise tensored operators. The above operators are both standard in SSE based reductions. In particular, the noisy graph walk operator is useful for the influence decoding step of SSE (see Lemma \ref{lem:sse-dec}), the noisy leakage operators are responsible for large spectral gap of the averaging operator, and allows for good local bias control (see Lemma \ref{lem:bias-conc}). Now we are ready to describe the reduction -- here, the distribution over hyperedges is given by the dictatorship test in Figure \ref{fig:dksh-test}. 

\begin{figure}[ht!]
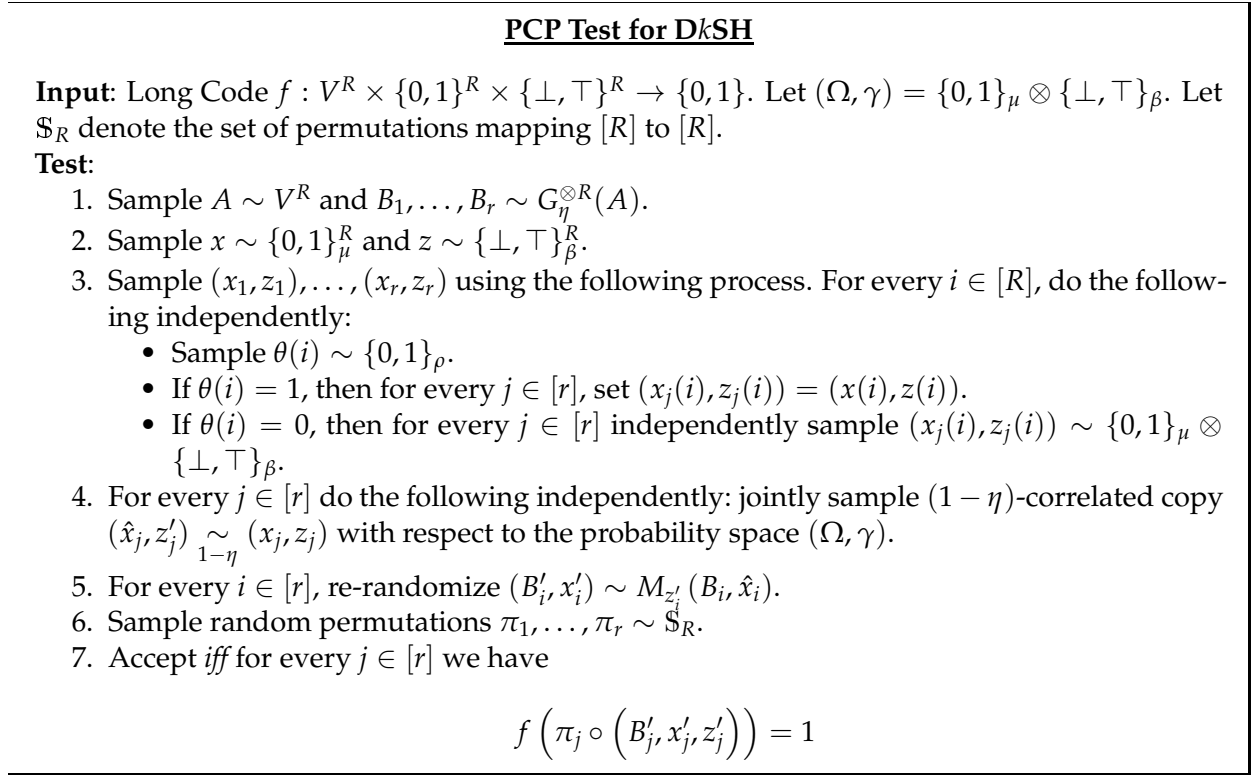

	\begin{mdframed}
		\begin{center}
			\underline{\bf PCP Test for D$k$SH} \\[10pt]
		\end{center}	
		{\bf Input}: Long Code $f:V^R \times \{0,1\}^R \times \{\bot,\top\}^R \to \{0,1\}$. Let $(\Omega,\gamma) = \{0,1\}_\mu \otimes \{\bot,\top\}_\beta$. Let $\mathbbm{S}_R$ denote the set of permutations mapping $[R]$ to $[R]$.\\
		{\bf Test}:
		\begin{enumerate}
			\item Sample $A \sim V^R$ and $B_1,\ldots,B_r \sim G^{\otimes R}_\eta(A)$.
			\item Sample $x \sim \{0,1\}^R_\mu$ and $z \sim \{\bot,\top\}^R_\beta$.
			\item Sample $(x_1,z_1),\ldots,(x_r,z_r)$ using the following process. For every $i \in [R]$, do the following independently:
			\begin{itemize}
				\item Sample $\theta(i) \sim \{0,1\}_\rho$.
				\item If $\theta(i) = 1$, then for every $j \in [r]$, set $(x_j(i),z_j(i)) = (x(i),z(i))$.
				\item If $\theta(i) = 0$, then for every $j \in [r]$ independently sample $(x_j(i),z_j(i)) \sim \{0,1\}_\mu \otimes \{\bot,\top\}_\beta$. 
			\end{itemize}
			 \label{step:xz-set}
			\item For every $j \in [r]$  do the following independently: jointly sample $(1 - \eta)$-correlated copy $(\hat{x}_j,z'_j) \underset{1 - \eta}{\sim} (x_j,z_j)$ with respect to the probability space $(\Omega,\gamma)$.
			\item For every $i \in [r]$, re-randomize $(B'_i,x'_i) \sim M_{{z}'_i}\left(B_i,\hat{x}_i\right)$.
			\item Sample random permutations $\pi_1,\ldots,\pi_r \sim \mathds{S}_R$.
			\item Accept {\em iff} for every $j \in [r]$ we have  
			\[
			f\left(\pi_j\circ \left(B'_j,x'_j,z'_j\right)\right) = 1
			\]
		\end{enumerate}
	\end{mdframed}
	\caption{PCP Verifier for {\sc Densest}-$k$-{\sc SubHypergraph}}
	\label{fig:dksh-test}
\end{figure}

{\bf Parameters of the Reduction}. We set the parameters for the reduction and its analysis as follows.

\begin{itemize}
	\item $C$ is a fixed constant from Lemma \ref{lem:match-ix}.
	\item $\rho = 1/(2C'r^2\log(1/\mu))$ where $C'$ is the constant from Theorem \ref{thm:ks-stab}.
	\item $\beta = \mu^{Cr}$.
	\item $\nu = \min\{\rho^2 \mu e^{-Cr},C\mu^r\}/10$.
	\item $\alpha = \rho \mu^r$.
	\item $\tau = \tau(\nu,r,\alpha)$ as in Theorem \ref{thm:ks-stab}.
	\item $\eta = \beta^2/r$.
	\item $\epsilon = \frac{\beta^2\nu^4\eta^4\tau^6}{2^{24} r^2}$.
	\item $M = 1/\sqrt{\epsilon} = \frac{2^{12}r}{\beta\nu^2\eta^2\tau^3}$.
	\item $R = 1/(r\beta \delta$).
\end{itemize}	

Now we analyze the above reduction.

\subsection{Completeness Analysis for \dksh~Reduction}		\label{sec:dksh-comp}

Let $S \subseteq V$ be the set of volume $\delta$ such that $\phi_G(S) \leq \epsilon$ guaranteed by the YES Case guarantee of $G$. Towards arguing completeness, we begin by setting up some notation. For every $(A,z) \in V^R \times \{\bot,\top\}^R$ define the set $\Pi(A,z) \subseteq [R]$ as 
\[
\Pi(A,z) := \left\{i \in [R] | (A(i),z(i)) \in S \times \{\top\}\right\}.
\]
Furthermore, we define a mapping $i^* : V^R \times \{\bot,\top\}^R \to [R]$ as follows:
\[
i^*(A,z) = 
\begin{cases}
	i & \mbox{ if } \Pi(A,z) = \{i\}, i \in [R] \\
	1 & \mbox{ otherwise }.
\end{cases}
\]
We will find it convenient to define the set $V_{\rm good} := \{(A,z) \in V^R \times \{\bot,\top\}^R | |\Pi(A,z)| = 1\}$ i.e., it is the set of $R$-tuples that intersect uniquely with $S \times \{\top\}$ -- note that for every $(A,z) \in V_{\rm good}$, the corresponding $i^*(A,z)$ value is given at the unique index at which $(A,z)$ intersects with $S \times \{\top\}$. For the remainder of this section, for ease of notation we shall denote  $\cA = (A,z)$. Analogously, for every $j \in [r]$ we shall define $\cB_j := (B_j,z'_j)$ and $\cB'_j = (B'_j,z'_j)$.  The key lemma towards establishing completeness is the following which shows that with constant probability over the choices of $(B'_j,x'_j,z'_j)_{j \in [r]}$, the corresponding triples have matching $\Pi(\cdot)$ sets, in addition to satisfying some other properties.

\begin{lemma}				\label{lem:match-ix}
	There exists a large constant $C>0$ such that 	
	\[
	\Pr\left[ \Big\{\cA \in V_{\rm good}\Big\}  \wedge \Big\{{\theta}(i^*(\cA)) = 1\Big\} \wedge \Big\{\forall j \in [r], \ \Pi(\cA) = \Pi(\cB'_j)\Big\} 
	\right] \geq 0.5 e^{-2C} r^{-1}\rho. 
	\]
\end{lemma}

\begin{proof}
	Define the set $\cS_\top := S \times \{ \top \}$. For ease of notation we shall denote $i^*_{A} = i^*(A,z)$ and $i^*_{B_j} = i^*(B_j,z'_j)$ for every $j \in [r]$. We shall also find it convenient to define the following events: 
	\begin{gather*}
	\cE_0 := \left\{\forall j \in [r], \Pi(\cA) = \Pi(\cB_j)\right\},\qquad\qquad \cE_1 := \left\{\theta(i^*_A) = 1\right\}, \\
	\cE_2 := \cE_1 \wedge \left\{\forall j \in [R]  \ : \ (B_j(i^*_A),z'_j(i^*_A)) \in \cS_\top \right\},  \\
	\cE_{\rm good} := \Big\{\cA \in V_{\rm good}\Big\}.
	\end{gather*}
	
	Our first observation here is that conditioned on the event $\Pi(\cA) = \Pi(\cB_j)$ for every $j \in [r]$, we have $\Pi(\cA) = \Pi(\cB'_j)$ for every $j \in [r]$ with probability $1$ (using Claim \ref{cl:match-ix}). Therefore, it suffices to track the probability of the event $\Pi(\cA) = \Pi(\cB_j)$ for every $j \in [r]$ along with the events $\cE_1,\cE_{\rm good}$.  Hence we proceed as follows: 
	\begin{align}
	&\Pr_{(\cA,\{\cB_j\})} \left[\cE_0 \wedge \cE_1 \wedge \cE_{\rm good}\right] 		\non \\
	&= \Pr\left[\cA \in V_{\rm good}\right] \Pr\left[ \cE_0 \wedge \cE_1 \Big| \cA \in V_{\rm good}\right] 		\non\\
	&= \left(\sum_{i \in [R]} \Pr_{(A,z) \sim V^R \times \{\bot,\top\}^R_\beta}\left[\Big\{(A(i),z(i)) \in \cS_\top \Big\} \wedge 
	\Big\{\forall j \neq i, (A(j),z(j)) \notin \cS_\top\Big\}\right]\right) \Pr\left[ \cE_0 \wedge \cE_1 \Big| \cA \in V_{\rm good}\right] 	\non\\
	&= R{\beta \delta}\left(1- {\beta \delta}\right)^{R-1} \Pr\left[  \cE_0 \wedge \cE_1 \Big| \cA \in V_{\rm good}\right] 	\non\\
	&\geq \frac{e^{-1}}{r} \Pr\left[\cE_0 \wedge \cE_1 \Big| \cA \in V_{\rm good}\right],		\label{eqn:prob-rhs} 
	\end{align}
	where the last step follows from our choice of $R = 1/(r\beta\delta)$. The rest of the proof bounds the probability term in \eqref{eqn:prob-rhs}. For a fixed choice of $(A,z)$, define sets $Q_0,Q_1$ and $Q_2$ as follows.
	\begin{gather*}
		Q_0 \defeq \left\{i \in [R] \Big| (A(i),z(i)) \in S \times \{\bot\}\right\} \qquad\quad\quad 
		Q_1 \defeq \left\{i \in [R] \Big| (A(i),z(i)) \in S^c \times \{\top\} \right\} \\
		Q_2 \defeq \left\{i \in [R] \Big| (A(i),z(i)) \in S^c \times \{\bot\} \right\} \\
	\end{gather*} 
	Together, we shall refer to the above as the configuration, and denoted it by $\cQ := (Q_0,Q_1,Q_2,i^*_A)$ where recall that $i^*_A = i^*(A,z)$. Proceeding to expand the probability term from \eqref{eqn:prob-rhs}: 
	\begin{align}
		&\Pr\left[\cE_0 \wedge \cE_1 \Big| \cA \in V_{\rm good}\right] \\
		&= \Pr\left[  \left\{\forall \ j \in [r], \ \Pi\left(\cA\right) = \Pi\left(\cB_j\right)\right\} 
		\wedge \left\{{\theta}{(i^*_A)} = 1 \right\} \Big| \cA \in V_{\rm good}\right] 		\non\\
		& = \Ex_{\cQ | \cE_{\rm good}}\Ex_{\cA|\cQ}\Pr_{\cB| \cA} \left[\left\{\forall \ j \in [r], \  \Pi\left(\cA\right) = \Pi\left(\cB_j\right) \right\} \wedge \left\{{\theta}{(i^*_A)} = 1\right\}\right]  \non\\
		& =\Ex_{\cQ | \cE_{\rm good}}\Ex_{\cA|\cQ} \Pr_{\cB| \cA} \left[  \cE_2 
		\mbox{ and } \bigwedge^2_{\ell = 0} \Big\{ \forall \ j \in [r], \forall \ i \in Q_\ell, \ (B_j(i),z'_j(i)) \notin \cS_\top \Big\} \right]  \non\\
		&=\Ex_{\cQ | \cE_{\rm good}}\left[\Ex_{\cA|\cQ}\left(\Pr_{\cB| \cA} \left[\cE_2 \right]\right) 
		\cdot\prod^{2}_{\ell = 0}\left(\Ex_{\cA|\cQ}\Pr_{\cB| \cA} \left[ \forall \ j \in [r], \forall \ i \in Q_\ell, \ (B_j(i),z'_j(i)) \notin \cS_\top \right]\right) \right] 		\label{eqn:prob-rhs1}
	\end{align}
	where the last equality follows from the following observations:
	\begin{itemize}
		\item Fixing $\cQ$, the random variables $(A(i),z(i))$ are all independent for distinct $i \in [R]$.
		\item For a fixing of $(A,z)$, for every $i \in [R]$,  the random variable $(B_j(i),z'_j(i))_{j \in [r]}$ is dependent only of $(A(i),z(i))$ and is independent of the random variables $\{(B_j(i'),z'_j(i'))\}_{i' \neq i}$.
		\item The events $\cE_2$ and the events corresponding to the sets $Q_0, Q_1$ and $Q_2$ are all supported on disjoint coordinates in $[R]$, and hence they are all independent conditioned on a fixing of $\cQ$. 
	\end{itemize}
	Now for a fixed instantiation of $\cQ$, we bound each of the inner expectation terms from \eqref{eqn:prob-rhs1} individually. A useful observation here is that fixing $\cQ$ completely determines $z$, and but still leaves open the choice of each $A(i)$ within subsets $S$ or $V \setminus S$.
	
	{\bf Bounding probability of $\cE_2$}. We begin by bounding the probability of the event $\cE_2$:
	\begin{align}
		&\Ex_{\cA|\cQ} \Pr_{\cB| \cA} \left[ \Big\{\forall\ j \in [r], \  (B_{j}(i^*_A),z'_{j}(i^*_A)) \in \cS_\top\Big\} 
		\wedge \Big\{\theta(i^*_A) = 1\Big\} \right] 		\non\\
		& = \Ex_{\cA|\cQ} \Pr_{\cB| \cA}\left[\Big\{\forall \ j \in [r] : (B_{j}(i^*_A),z'_{j}(i^*_A)) \in S \times \{\top\} \Big\} \wedge \Big\{ \theta(i^*_A) = 1 \Big\} \right] 			\non\\
		& \overset{1}{=} \Ex_{A|\cQ} \left[ \Pr_{(B_j(i^*_A))^{r}_{j= 1} \sim G^{R}_\eta(A(i^*_A))} \left[\forall \ j \in [r]:  B_j(i^*_A) \in S\Big| A(i^*_A) \in S \right] \right]   		\non\\
		& \hspace{4cm} \times \Ex_{z|\cQ} \left[\Pr_{\theta(i^*_A)} \left[\theta(i^*_A) = 1 \right] \Pr_{\{z'_j\}^r_{j = 1} | z} \left[\forall \ j \in [r] : z'_{j}(i^*_A) = \top \Big|\theta(i^*_A) = 1 \right] \right]  		\label{eqn:st-1}\\
		& \overset{2}{\geq} \left(1 - r\left(\eta + \epsilon\right)\right) (1 - \eta r) \Pr_{\theta(i^*_A)}\left[ \theta(i^*_A) = 1 \right] 		\label{eqn:st-2}\\
		&\overset{3}{\geq} 0.9 \rho. 	\label{eqn:dksh-0}
	\end{align}
	We argue the above steps in the following way. For step $1$, we observe that fixing $\cQ$, the conditional distributions of $\{B_j(i^*_A)\}^r_{j = 1}$ and  ${\theta}({i^*_A}),\{z'(i^*_A)_j\}^r_{j = 1}$ are independent of each other. For step $2$, the lower bound on the first and the third probability terms can be derived as follows. For the first term, note that conditioning on $\cQ$, $A(i^*_A)$ is a randomly chosen vertex in $S$. Furthermore, recall that sampling $B_j(i^*_A) \sim G_{\eta}({A(i^*_A)})$ results in the following distribution on $B_j(i^*_A)$ : with probability $\eta$,  it is a randomly chosen vertex in $G$ (since $G$ is regular) and with probability $1 - \eta$, it is a $1$-step random walk on $G$ from $A(i^*_A)$. Hence for any fixed $j \in [r]$,
	\[
	\Pr_{B_j(i^*_A) \sim G_\eta(A(i^*_A))} \left[B_j(i^*_A) \notin S\Big| A(i^*_A) \in S \right] 
	\leq \eta + \Pr_{B_j(i^*_A) \sim G(A(i^*_A))}\left[B_j(i^*_A) \notin S \Big| A(i^*_A) \in S\right] 
	\leq \eta + \epsilon,
	\]
	where the second term is bounded using $\phi_G(S) \leq \epsilon$. Hence, the bound on the first term in step $1$ (i.e, \eqref{eqn:st-1}) follows by a union bound over all $j \in [r]$. For the second term in step $2$~\eqref{eqn:st-2}, we first observe that since $\cQ$ is realized conditioned on the event $\cA \in V_{\rm good}$, we have $|\Pi(A,z)| = 1$ and hence $z(i^*_A) = \top$ using the definition of the map $i^*$. Furthermore, 	
	conditioned on ${\theta}(i^*_A) = 1$ we have $z_j(i^*_A) = z(i^*_A)$ with probability $1$ (see Step \ref{step:xz-set} of Figure \ref{fig:dksh-test})
	and hence $z_j(i^*_A) = \top$ for every $j \in [r]$. Finally, fixing $(z_j(i^*_A))_{j \in [r]}$, each $z'_j(i^*_A)$ is a $(1 - \eta)$ correlated copy of $z_j(i^*_A)$. Hence, 
	\[
	\Pr_{\{z'_j\}^r_{j = 1}} \left[\forall \ j \in [r] : z'_{j}(i^*_A) = \top \Big| \theta(i^*_A) = 1\right]
	\geq 1 - \sum_{j \in [r]} \Pr_{{z'_j}(i^*_A) \corr{(1 - \eta)} z_j(i^*_A)} \left[ z'_{j}(i^*_A) \neq z_j(i^*_A)\right] \geq 1 - \eta r.
	\]
	Finally, step $3$ \eqref{eqn:dksh-0} follows using the fact that $\theta(i^*_A) \sim \{0,1\}_{\rho}$ under our test distribution. 
	
	{\bf Bounding $\ell = 0$ from \eqref{eqn:prob-rhs1}}. Recall that from the definition of $Q_0$, for every $i \in Q_0$ we have $(A(i),z(i)) \in S \times \{\bot\}$. Hence we can lower bound:  
	\begin{align}
	&\Ex_{\cA|\cQ}\Pr_{\cB| \cA} \left[ \forall \ j \in [r], \ \forall \ i \in Q_0 : (B_j(i),z'_j(i)) \notin S \times \{\top\} \right]  \non\\
	&\geq \Ex_{\cA|\cQ}\Pr_{\cB| \cA} \left[ \forall \ j \in [r], \forall \ i \in Q_0 : z'_j(i) = \bot \right] 		\non\\
	&\overset{1}{=} \Ex_{\cA|\cQ}\left[\prod_{i \in Q_0} \Pr_{\cB| \cA} \left[ \forall \ j \in [r], z'_j(i) = \bot \Big| \ i \in Q_0 \right] \right]		\non\\
	&{=} \Ex_{\cA|\cQ}\left[\left(1 - \Pr_{\cB| \cA} \left[ \exists \ j \in [r],  z'_j(i) = \top \Big| \ i \in Q_0  \right]\right)^{|Q_0|} \right] 			\non\\
	&{=} \Ex_{\cA|\cQ}\left[\left(1 - \Pr_{\{z'_j\}^r_{j = 1}} \left[ \exists \ j \in [r],  z'_j(i) = \top \Big| \ z(i) = \bot \right]\right)^{|Q_0|} \right] 		\tag{Using Defn. of $Q_0$}\non\\
	&{\geq} \Ex_{\cA|\cQ}\left[\left(1 - \sum_{j = 1}^r\Pr_{\{z'_j\}^r_{j = 1}} \left[ z'_j(i) = \top \Big| \ z(i) = \bot \right]\right)^{|Q_0|} \right] 		\non\\
	&\overset{2}{\geq} \left(1 - r(\eta + \beta(1 - \rho))\right)^{|Q_0|} 		\non \\
	&\overset{3}{\geq} \left(1 - 2r\beta(1 - \rho)\right)^{|Q_0|} 			\label{eqn:q0-val}
	\end{align}
	where step $1$ is again using the fact that the random variables $(B_j(i),z'_j(i))_{j \in [r]}$ are independent for different $i \in [R]$ for a fixing of $\cQ$. Step $2$ can be argued as follows. For any $i \in Q_0$ we can bound 
	\begin{align}
	\Pr_{\{z'_j\}^r_{j = 1}} \left[ z'_j(i) = \top \Big| z(i) = \bot\right]
	& = \Ex_{z_j(i) \underset{\rho}{\sim}{z}(i)} \ \  \Pr_{z'_j(i) \underset{1 - \eta}{\sim }z_j(i)} \left[ z'_j(i) = \top \Big| z(i) = \bot \right] \non \\
	& \leq \eta + (1 - \rho)\beta.				\label{eqn:z'-bound}
	\end{align} 
	Finally step $3$ uses the fact that $\eta = \beta^2/r$ and hence $\eta \leq \beta(1 - \rho)$. 
	
	{\bf Bounding $\ell = 1$ from \eqref{eqn:prob-rhs1}}. Since $i \in Q_1$ implies $(A(i),z(i)) \in S^c \times \{\top\}$, here we can use the bound on the expansion of $S$. Observe that 
	\begin{align}
		&\Ex_{\cA|\cQ}\Pr_{\cB| \cA} \left[ \forall j \in [r], \ \forall \ i \in Q_1 : B_j(i) \notin S \times \{\top\} \right] 		\non\\
		&\geq \Ex_{\cA|\cQ}\left(1 - \sum_{i \in Q_1}\sum_{j \in [r]} \Pr_{\cB| \cA} \left[ B_j(i) \in S \big| \ i \in Q_1\right]\right) 		\non\\
		&= \Ex_{\cA|\cQ}\left(1 - \sum_{i \in Q_1}\sum_{j \in [r]} \Pr_{\cB| \cA} \left[ B_j(i) \in S \big| \ A(i) \notin S\right]\right) 		\non\\
		&\overset{1}{=} 1 - r|Q_1|\Pr_{B_j(i) \sim G_{\eta}(A(i))}\left[B_j(i) \in S | A(i) \notin S\right] 		\label{eqn:comp-step}\\
		&{=} 1 - r|Q_1|\left(\eta \Pr_{B(i) \sim G} \left[B(i) \in S\right] + (1 - \eta) \Pr_{B_j(i) \sim G(A(i))}\left[B_j(i) \in S | A(i) \notin S\right]\right) 		\non\\
		&\overset{2}{\geq} 1 - r|Q_1|\left(\delta \eta + \frac{\Pr_{B_j(i) \sim G(A(i))}\left[B(i) \in S \wedge  A(i) \notin S\right]}{\Pr\left[A(i) \notin S\right]}\right) 		\non\\
		&\overset{3}{\geq} 1 - r|Q_1|( \delta \eta + \delta \epsilon/(1 - \delta)) \non \\
		&\geq 1 - 2r|Q_1| \delta (\epsilon + \eta), 		\label{eqn:q1-val}
	\end{align}
	where step $1$ can be argued as follows. Fixing $Q_1$,  $A(i)$ is identically distributed for every $i \in Q_1$. Furthermore, observe that $B_j(i)$'s are sampled independently fixing $A(i)$, it follows that for every $j \in [r]$, and $i \in Q_1$ the marginal distribution of $B_j(i)$ conditioned on $Q_1$ is identical. 	Therefore, the probability terms inside the summation over $i,j$ are  all identical quantities. For step $2$, we observe that $G_\eta$ preforms a random walk on $G$ with probability $1 - \eta$ and returns a completely random vertex with probability $\eta$. Step $3$ uses the bound on the expansion of $S$ and the fact that ${\sf Vol}(S) = \delta$. 
	
{\bf Bounding $\ell = 2$ from \eqref{eqn:prob-rhs1}}. We bound this term by combining the arguments for the $\ell = 0,1$ terms: 
	\begin{align}
		&\Ex_{\cA|\cQ}\Pr_{\cB| \cA} \Big[ \forall \ j \in [r], \ \forall \ i \in Q_2 : (B_j(i),z'_{j}(i)) \notin S \times \{\top\} \Big]	\non \\
		&\overset{1}{\geq} 1 - r|Q_2|\Pr\Big[(B_j(i),z'_{j}(i)) \in S \times \{\top\} \Big| (A(i),z(i)) \in S^c \times \{\bot\}\Big] 		\non\\
		&= 1 - r|Q_2|\Pr_{z'_j(i) \underset{\rho(1 - \eta)}{\sim}(z(i))}\left[z'_j(i)  = \top \Big| z(i) = \bot\right] 
		\Pr_{B_j(i) \sim G_\eta(A(i))}\Big[B_j(i) \in S | A(i) \notin S\Big] 		\non\\
		&\overset{2}{\geq} 1 - \left(\frac{1}{\beta \delta}\right)\left(\eta + \beta (1 - \rho)\right)			
		\left(\delta \eta + \frac{\Pr\left[B_j(i) \in S \wedge  A(i) \notin S\right]}{\Pr\left[A(i) \notin S\right]}\right) 		\non\\
		&\overset{3}{\geq} 1 - \left(\frac{1}{\beta \delta}\right)\left(\eta + \beta (1 - \rho)\right)			
		\left(\delta \eta + \epsilon \delta/(1 - \delta)\right) 			\non\\
		&\overset{4}{\geq} 1 - \left(2\delta^{-1} \Big(\delta \eta + \epsilon \delta/(1 - \delta)\Big)\right) 	\non\\
		&\geq 1 - 4(\epsilon + \eta). 		\label{eqn:q2-val}
	\end{align}
	Here step $1$ uses an argument similar to step in \eqref{eqn:comp-step}. In step $2$ we trivially upper bound $|Q_2| \leq R = 1/(r\beta \delta)$ and the probability terms are bounded using \eqref{eqn:z'-bound}. For step $3$, we use the bound on the expansion of $S$, and step $4$ again follows from $2\eta \leq \beta(1 - \rho)$ using our choice of parameters. 
	
{\bf Putting Things Together}. Denote  the event $|Q_0| \leq C/(\beta r)$ as $\cE_Q$. Combining the bounds from \eqref{eqn:dksh-0},\eqref{eqn:q0-val}, \eqref{eqn:q1-val} and \eqref{eqn:q2-val} and plugging them in \eqref{eqn:prob-rhs1} we get that 
	\begin{align*}
		&\Ex_{\cQ|\cE_{\rm good}}\left[\Ex_{\cA|\cQ}\Pr_{\cB| \cA} 
		\Big[\cE_0 \wedge \cE_1\Big]\right] \\
		&\geq \Ex_{\cQ | \cE_{\rm good}}\Ex_{\cA|\cQ}\left(\Pr_{\cB| \cA} \left[\cE_2 \right]\right) 
		\cdot\prod^{2}_{\ell = 0}\left(\Ex_{\cA|\cQ}\Pr_{\cB| \cA} \left[ \forall \ j \in [r], \forall \ i \in Q_\ell, \ (B_j(i),z'_j(i)) \notin \cS_\top \right]\right) 		\tag{From \eqref{eqn:prob-rhs1}} \\
		&\geq \Ex_{\cQ|\cE_{\rm good}}\left[0.9\rho \cdot \Big(1 - 2r\beta(1 - \rho)\Big)^{|Q_0|}\Big(1 - 2r|Q_1| (\epsilon + \eta) \delta\Big)(1 - 4(\epsilon + \eta))\right]		\tag{From \eqref{eqn:dksh-0},\eqref{eqn:q0-val}, \eqref{eqn:q1-val},\eqref{eqn:q2-val}} \\
		&\overset{1}{\geq} \Ex_{\cQ|\cE_Q,\cE_{\rm good}}\left[0.8\rho\Big(1 - 2r\beta(1 - \rho)\Big)^{|Q_0|}\Big(1 - 2r|Q_1|\delta(\eta + \epsilon )\Big)\right] - e^{-C/(8\beta r)} \\
		&\overset{2}{\geq} \Ex_{\cQ|\cE_Q,\cE_{\rm good}}\left[0.8\rho\Big(1 - 2r\beta(1 - \rho)\Big)^{C/(\beta r)}\Big(1 - 2r|Q_1|\delta(\eta + \epsilon )\Big)\right] - e^{-C/(8\beta r)} \\
		&{\geq} 0.8\rho e^{-2C}\Ex_{\cQ|\cE_Q,\cE_{\rm good}}\Bigg[~1 - 2r|Q_1|\delta (\eta + \epsilon )~\Bigg] - e^{-C/(8\beta r)} \\
		&= 0.8\rho e^{-2C}\Big(1 - 2r\Ex_{\cQ|\cE_Q,\cE_{\rm good}}\big[|Q_1|\big]\delta (\eta + \epsilon )\Big) - e^{-C/(8\beta r)} \\
		&\overset{3}{\geq} 0.9\rho e^{-2C}\left(1 - \frac{12}{\delta} \cdot \delta(r\eta + r\epsilon )\right) - e^{-C/(8\beta r)} \\
		&= 0.8\rho e^{-2C}\left(1 - 12(r\epsilon + r\eta) \right) - e^{-C/(8\beta r)} \\
		&\overset{4}{\geq} 0.5\rho e^{-2C}.
	\end{align*}
	Here step $1$ uses the first item of Observation \ref{obs:bounds}, step $2$ is due to the fact that the event $\cE_Q$ implies $|Q_0| \leq C/r\beta$. Step $3$ follows from  the second item of Observation \ref{obs:bounds}. Finally, in step $4$, the first term dominates the second due to $\rho =1/{C'r^2\log(1/\mu)}$ and $\beta \ll \mu$. Plugging in the above bound into \eqref{eqn:prob-rhs} finishes the proof.
\end{proof}

{\bf Labeling Strategy}.

Consider the following labeling strategy. For every $(A,z) \in V^R \times \{\top,\bot\}^R$, we define the corresponding long code as $f_{A,z}: \{0,1\}^R \to \{0,1\}$ as $f_{A,z} = \chi_{i^*(A,z)}$ i.e., $f_{A,z}(x) := x(i^*(A,z))$. Note that this assignment satisfies
\begin{align*}
	\Ex_{(A,x,z) \sim V^R \times \{0,1\}^R_\mu \times \{\bot,\top\}^R_\beta}\left[f(A,x,z) \right] 
	& = \Ex_{(A,z) \sim V^R \times \{\bot,\top\}^R_\beta} \Ex_{x \sim \{0,1\}^R_\mu} \left[f_{A,z}(x)\right] \\
	& = \Ex_{(A,z) \sim V^R \times \{\bot,\top\}^R_\beta} \Ex_{x \sim \{0,1\}^R_\mu} \left[x(i^*(A,z))\right] \\
	& = \mu,
\end{align*}
i.e, $f$ indicates a set of relative weight $\mu$ in the hypergraph defined in Figure \ref{fig:dksh-test}. Now we bound the fraction of edges induced by the set indicated $f$, which is the same as the probability of the test accepting this assignment. Let the events $\cE_0,\cE_1,\cE_{\rm good}$ be as in the proof of Lemma \ref{lem:match-ix} and define $\cE := \cE_0 \wedge \cE_1 \wedge \cE_{\rm good}$. Then,
\begin{align}
\Pr\left[\textrm{Test Accepts}\right] 
& = \Pr\left[\forall \ j \in [r] : f\left(\pi_j\circ\left(B'_j,x'_j,z'_j\right)\right) = 1\right] 			\non\\
& \geq \Pr\left[\cE\right]  \Pr\left[\forall \ j \in [r] : f\left(\pi_j\circ\left(B'_j,x'_j,z'_j\right)\right) = 1
\Big| \cE\right] 		\non\\
& = \Pr\left[\cE\right]  \Pr\left[\forall \ j \in [r] : f\left(B'_j,x'_j,z'_j\right) = 1
\Big| \cE \right] 	\tag{Claim \ref{cl:perm}}	\non\\
& \overset{1}{\geq} 0.5\rho r^{-1} e^{-2C} \Pr\left[\forall \ j \in [r] : f\left(B'_j,x'_j,z'_j\right) = 1
\Big| \cE \right] 		\non\\
& \overset{2}{=} 0.5\rho e^{-2C}r^{-1}\Pr\left[\forall \ j \in [r] :  x'_j({i^*(\cB'_j)}) = 1 \Big| \cE \right] 		\non\\
& \geq 0.5\rho e^{-2C}r^{-1}\Pr\left[\forall \ j \in [r] : x({i^*_A}) = x'_j({i^*(\cB'_j)}) = 1 \Big| \cE \right] 		\non\\
& \overset{3}{=} 0.5\rho e^{-2C}r^{-1}\Pr\left[\forall \ j \in [r] : {x}({i^*_A}) = \hat{x}_j(i^*(\cB'_j)) = 1 \Big| \cE \right]		\non\\
& \overset{4}{=} 0.5\rho e^{-2C}r^{-1}\Pr\left[\forall \ j \in [r] : {x}({i^*_A}) = \hat{x}_j({i^*_A}) = 1 \Big| \cE \right]		\non\\
& \overset{5}{\geq} 0.5\rho e^{-2C}r^{-1}\left(\Pr\left[{x}(i^*_A) =  1 \Big| \cE \right] - r\eta \right) 		\non\\
& \overset{6}{=} 0.5\rho e^{-2C}r^{-1}\left(\Pr\Big[{x}(i^*_A) = 1  \Big] - r\eta \right) 		\non\\
& = 0.5e^{-2C} (\mu - r\eta)\rho r^{-1} 			\non\\
& \overset{7}{\geq} 0.5e^{-2C} \rho \mu r^{-1}/2.			\label{eqn:sse-comp}
\end{align}

We explain the various steps above. Step $1$ lower bounds the first probability expression using Lemma \ref{lem:match-ix}. Step $2$ follows from our definition of $f$. Step $3$ follows from the observation that conditioned on the event $\cE$, for every $j \in [R]$, we have $|\Pi(\cB'_j)| = 1$ and hence $z'_j(i^*(\cB'_j)) = \top$ using the definition of $i^*(\cB'_j)$. This in turn implies that $\hat{x}_j(i^*(\cB'_j)) = x'_j(i^*(\cB'_j))$ with probability $1$. Step $4$ is by observing that conditioned on $\cE$ we have $i^*(\cB'_j) = i^*_A$ for every $j \in [r]$. Step $5$ follows by combining the following observations:
\begin{itemize} 
	\item Conditioned on $\cE$, we have $\theta(i^*_A) = 1$ and hence $x_j(i^*_A) = x(i^*_A)$ for every $j \in [r]$ with probability $1$.
	\item For every $j \in [r]$, $\hat{x}_j(i^*_A)$ is an independent $(1 - \eta)$-correlated copy of $x_j(i^*_A)$.
\end{itemize}
Therefore, combining the two above observations we have 
\begin{align*}
&\Pr\left[\forall j \in [r] : \hat{x}_j(i^*_A) = x(i^*_A) = 1 \Big| \cE\right] \\
&\geq \Pr\left[\forall j \in [r] : {x}_j(i^*_A) = x(i^*_A) = 1 \Big| \cE\right] - \Pr\left[\exists j \in [r] : \hat{x}_j(i^*_A) \neq  x_j(i^*_A) \Big| \cE\right] \\
&\geq \Pr\left[ x(i^*_A) = 1 \Big| \cE\right] - \eta r.
\end{align*}
Step $6$ uses the observation that $x$ is independent of the variables $A,\{B_j\}_{j \in [r]},\{\theta(i)\}_{i \in [R]}$ which are variables that determine the events $\cE$. Finally step $7$ follows from our choice of $\eta$. 

Summarizing, we showed that if $G$ is a YES instance, then there exists a set of volume $\mu$ which induces at least $\Omega(r^{-3}/\log(1/\mu))$-weight of hyperedges in the hypergraph output by the test.

\subsection{Miscellaneous Lemmas for Completeness Analysis}

\begin{claim}				\label{cl:match-ix}
	For any triple $(A,x,z)$, and any $(A',x',z)$ generated by sampling $(A',x') \sim M_z(A,x)$, we have $\Pi(A',z) = \Pi(A,z)$ with probability $1$. Consequently, if $|\Pi(A,z)| = 1$, then $i^*(A,z) = i^*(A',z)$ with probability $1$.
\end{claim}

\begin{proof}
	To begin with, fix an $i \in \Pi(A,z)$. Then using the definition of $\Pi(A,z)$, we must have $z(i) = \top$ and hence $(A'(i),z(i)) = (A(i),z(i)) \in S \times \{\top\}$ with probability $1$. Conversely, let us look at an $i \notin \Pi(A,z)$. We now consider two cases:
	
	{\bf Case (i)}: Suppose $z(i) = \bot$. Then with probability $1$, $(A'(i),z(i)) \notin S \times \{\top\}$.
	
	{\bf Case (ii)}: Suppose $z(i) = \top$. Since $i \notin \Pi(A,z)$ we must have $A(i) \notin S$. Again, since $z(i) = \top$, it follows that $A'(i) = A(i) \notin S$ with probability $1$.    
	
	Combining cases (i) and (ii), we get that for every $i \in [R] \setminus \Pi(A,z)$ we must have $(A'(i),z(i)) \notin S \times \{\top\}$. The above observations together imply that $i \in \Pi(A,z)$ if and only if $i \in \Pi(A',z)$ and hence the first claim follows. The second claim follows directly using the definition of the map $i^*$.
	
\end{proof}

 \begin{claim}				\label{cl:perm}
 	Conditioned on events $\cE = \cE_0 \wedge \cE_1 \wedge \cE_{\rm good}$, for any $j \in [R]$ and any permutation $\pi:[R] \to [R]$ we have $f(\pi\circ(B'_j,x'_j,z'_j)) = f(B'_j,x'_j,z'_j)$. 
 \end{claim}
\begin{proof}
	Since $\cE$ holds, using Claim \ref{cl:match-ix} we have 
	\[
	|\Pi(B'_j,z'_j)| = |\Pi(B_j,z'_j)| = |\Pi(A,z)| = 1
	\] 
	for every $j \in [r]$. Now let $\Pi(B'_j,z'_j) = \{i_j\}$. Then note that for any permutation $\pi:[R] \to [R]$ we have $\Pi(\pi \circ (B'_j,z'_j)) = \{\pi(i_j)\}$ and hence $i^*(\pi \circ(B'_j,z'_j)) = \pi(i_j)$. Then,
	\[
	f\left(\pi \circ(B'_j,x'_j,z'_j)\right) = 
	f_{\pi \circ(B'_j,z'_j)}(\pi(x'_j)) = \chi_{\pi(i_j)}(\pi(x'_j)) = \chi_{i_j}(x'_j) = f(B'_j,x'_j,z'_j).
	\]
\end{proof}
	
\begin{observation}				\label{obs:bounds}
	Let $\cE_Q$ and $\cE_{\rm good}$ denote the events $|Q_0| \leq C/(\beta r)$ and $\cA \in V_{\rm good}$ respectively. Then,
	\begin{equation}				\label{eqn:stmt}
	\Pr_{\cQ|\cE_{\rm good}} \left[|Q_0| \geq \frac{Ce}{r\beta}\right] \leq e^{-C/(8\beta r)} 
	\qquad\qquad \textnormal{and} 	\qquad\qquad
	\Ex_{\cQ|\cE_Q,\cE_{\rm good}} \left[|Q_1| \right] \leq  \frac{2e}{\delta}.
	\end{equation}
\end{observation}
\begin{proof}
From the choice of the test distribution, randomizing over the choice of $\cQ$ we have 
\[
\Ex_{\cQ}\left[|Q_0|\right] = R \Pr_{\cQ}\Big[(A(i),z(i)) \in S \times \{\bot\}\Big] = \delta(1 - \beta)R \leq 1/(r\beta).
\]
where in the last step we use $R = 1/(r \beta \delta)$. Furthermore, since each $(A(i),z(i))$ are independent for distinct choices of $i \in [R]$, the random variables $\mathbbm{1}(i \in Q_0)$ are i.i.d Bernoulli random variables. Therefore using Chernoff Bound, randomizing over the choice of $Q_0$, we get that	
\begin{equation}			\label{eqn:q0-conc}
	e^{-C/(4r\beta)} \geq \Pr_{\cQ}\left[|Q_0| \geq \frac{C}{\beta r}\right] \geq \left(\frac{e^{-1}}{r}\right) \Pr_{\cQ}\left[|Q_0| \geq \frac{C}{r\beta} \Big| A \in V_{\rm good}\right] 
\end{equation}	
which on rearranging gives us that the last probability expression is at most $e^{-C/(8\beta r)}$ whenever $\beta$ is small enough as function of $r$. For the second item of \eqref{eqn:stmt} observe that 
\begin{equation}			\label{eqn:q-step1}
	\Ex_{\cQ}\left[|Q_1|\right] = R \Pr_{\cQ}\left[(A(i),z(i)) \in S^c \times \{\top\}\right] = (1 - \delta)\beta R \leq 1/(\delta r).
\end{equation}
Furthermore, note that 
\begin{align}			
\Pr\Big[\cE_Q \wedge \cE_{\rm good}\Big] 
\geq \Pr\big[\cE_{\rm good}\big] - \Pr\big[\cE^c_Q\big]  	
&\geq R\beta\delta(1 - \beta \delta)^{R-1} - e^{-C/8r\beta} 	\non \\
&\geq \frac{e}{r}  - e^{-C/8r\beta} 			\non \\
&\geq \frac{1}{2er} \label{eqn:q-step2}
\end{align}
where the second inequality step follows from the first item of this lemma. Therefore, combining \eqref{eqn:q-step1} and \eqref{eqn:q-step2} we get that 
\[
\frac{1}{\delta r} \geq \Ex_{\cQ}\left[|Q_1|\right] \geq \frac{1}{2er} \cdot \Ex_{\cQ| \cE_Q,\cE_{\rm good}} \left[|Q_1|\right],
\]
which on rearranging gives us the second inequality of the lemma.
\end{proof}

\subsection{Soundness analysis for \dksh~reduction}

Let $G = (V,E)$ be a NO instance as in the setting of Theorem \ref{thm:dksh-redn}. Let $f:V^R \times \{0,1\}^R \times \{\bot,\top\}^R \to \{0,1\}$ be an assignment satisfying the global constraint 
\begin{equation}
\Ex_{A \sim V^R} \Ex_{x \sim \{0,1\}^R_\mu} \Ex_{z \sim \{\bot,\top\}^R_\beta} \left[f(A,x,z)\right] = \mu.				\label{eqn:bias-cond}
\end{equation}
To begin with, we observe that we can arithmetize the probability of the test accepting as: 
\begin{align}
\Pr\Big[\mbox{ Test Accepts }\Big]
 &= \Pr\left[\forall \ j \in [r] : f\left(\pi_j \circ \left(B'_j,x'_j,z'_j\right)\right) = 1\right] \non\\
 & = \Ex_{A \sim V^R}  \Ex_{(B'_j,x'_j,z'_j)^{r}_{j = 1}}  
 \Ex_{\pi_1,\ldots,\pi_r \sim \mathds{S}_R}\left[\prod_{j = 1}^r f\left(\pi_j \circ \left(B'_j,x'_j,z'_j\right)\right) \right] \non \\
 & = \Ex_{A \sim V^R}\Ex_{(\hat{x}_j,z'_j)^r_{j = 1}}\left[\prod_{j=1}^r \Ex_{B_j \sim G^{\otimes R}_\eta(A)} \Ex_{(B'_j,x'_j) \sim M_{z'_j}(B_j,\hat{x}_j)}\Ex_{\pi_j \sim \mathds{S}_R} \left[f\left(\pi_j \circ \left(B'_j,x'_j,z'_j\right)\right) \right] \right]. \label{eqn:dks-rhs}  
\end{align}
Symmetrizing over the (i) the noisy random walk over $G^{\otimes R}_\eta$ (ii) the action of the $M_{z_j}$ operator and (iii) the choice of random permutation, for every $A \in V^R$, we shall define the averaged function $g_{A}:\{0,1\}^R \times \{\bot,\top\}^R \to [0,1]$ as 
\[
g_{A}(x,z) := \Ex_{B \sim G^{\otimes R}_\eta(A)} \Ex_{(B',x') \sim M_{z}(B,x)} \Ex_{\pi \sim \mathds{S}_R} \left[f\left(\pi \circ \left(B',x',z\right)\right)\right].
\]
Using the above definition, we can rewrite \eqref{eqn:dks-rhs} as:
\begin{align}
&  \Ex_{A \sim V^R}\Ex_{(\hat{x}_j,z'_j)^r_{j = 1}}\left[\prod_{j=1}^r \Ex_{B_j \sim G^{\otimes R}_\eta(A)} \Ex_{(B'_j,x'_j) \sim M_{z'_j}(B_j,\hat{x}_j)}\Ex_{\pi_j \sim \mathds{S}_R} \left[f\left(\pi_j \circ \left(B'_j,x'_j,z'_j\right)\right)  \right] \right] \non\\
& =  \Ex_{A \sim V^R}\Ex_{(\hat{x}_j,z'_j)^r_{j = 1}}\left[\prod_{j=1}^r g_{A}(\hat{x}_j,z'_j) \right] 		\non\\
& =  \Ex_{A \sim V^R}\Ex_{(x_j,z_j)^r_{j = 1}}\left[\prod_{j=1}^r \Ex_{(\hat{x}_j,z'_j) \underset{1 - \eta}{\sim} (x_j,z_j)} g_{A}(\hat{x}_j,z'_j) \right]  \non	\\
& =  \Ex_{A \sim V^R}\Ex_{(x_j,z_j)^r_{j = 1}}\left[\prod_{j=1}^r T_{1 - \eta} g_{A}({x}_j,z_j) \right], 		
\label{eqn:dks-rhs1}
\end{align}
where $T_{1 - \eta}$ is the $(1 - \eta)$-correlated noise operator in the probability space $\{0,1\}^R_\mu \otimes \{\bot,\top\}^R_\beta$.

{\bf Invariance Principle Step}. Define the probability space $(\Omega,\gamma) = \{0,1\}_\mu \otimes \{\bot,\top\}_\beta$. Furthermore, define the set $V^1_{\rm nice} \subseteq V^R$ as 
\[
V^1_{\rm nice} := \left\{A \in V^R \Big| \max_{i \in [R]}\Inf{i}{T_{1 - \eta}g_A} \leq \tau\right\}
\]
where the influences are defined with respect to the probability space $(\Omega,\gamma)$. Since $G$ is a NO instance, by a standard influence decoding argument, we can show that most averaged functions $\{g_A\}_{A \in V^R}$ must have small influential coordinates. We show this formally in the following lemma.

\begin{lemma}			\label{lem:sse-dec1}
	Suppose $G$ is a NO instance as in the setting of Theorem \ref{thm:dksh-redn}. Then we have $|V^1_{\rm nice}| \geq (1 - \nu)|V^R|$.
\end{lemma}
We defer the proof of the above lemma to Section \ref{sec:sse-dec}. Furthermore, for $A \in V^R$ let $\mu_A := \Ex_{(x,z) \sim \gamma^R} \left[T_{1 - \eta}g_A(x,z)\right]$ and define the set $V^2_{\rm nice}$ as 
\[
V^2_{\rm nice} := \left\{A \in V^R \Big| \mu_A \in \mu(1 \pm \mu)\right\}
\]
Analogously, we have the following lemma which bounds the size of $V^2_{\rm nice}$.
\begin{lemma}				\label{lem:vp-nice}
	The set $V^2_{\rm nice}$ as defined above satisfies $|V^2_{\rm nice}| \geq (1 - \mu^{2r})|V|^R$.
\end{lemma}
\begin{proof}
	It suffices to show that 
	\[
	\Pr_{A \sim V^R} \left[|\mu_A - \mu| \geq \mu^2\right] \leq \mu^{2r}.
	\]
	Towards that, using the bias constraint from \eqref{eqn:bias-cond} we observe that: 
	\begin{align*}
		\Ex_A\big[\mu_A\big] 
		&= \Ex_A \Ex_{(x,z) \sim \gamma^R} \big[T_{1 - \eta} g_A(x,z)\big] \\
		&= \Ex_A \Ex_{B \sim G^{\otimes R}_\eta(A)} \Ex_{(x,z) \sim \gamma^R} \Ex_{(\hat{x},z') \underset{1 - \eta}{\sim} (x,z)} \Ex_{(B',x') \sim M_{z'}(B,\hat{x})} \Ex_{\pi \sim \mathbbm{S}_R}\big[f\left(\pi \circ\left(B',x',z'\right)\right)\big] \\
		&= \Ex_A \Ex_{(x,z) \sim \gamma^R}\big[f\left(A,x,z\right)\big] \\
		& = \mu,
	\end{align*}
	the expected bias of a random averaged long code is $\mu$. To show concentration around the expectation, we shall use the following key lemma from \cite{RST12}:
	\begin{lemma}[Lemma 6.7 \cite{RST12}]						\label{lem:bias-conc}
		Let $\{f_A\}_{A \in V^R}$ be a set of functions $f_A : \Omega^R \to [0,1]$. Furthermore, define $g_A$ as 
		\[
		g_A(x,z) \defeq \Ex_{B \sim G^{\otimes R}_\eta(A)} \Ex_{(B',x') \sim M_z(B,x)}\Ex_{\pi \sim \mathds{S}_R} \left[f\left(\pi \circ \left(B',x',z\right)\right) \right],
		\]
		and let $\mu_A := \Ex_{(x,z) \sim \Omega}\left[T_{1 - \eta}g_A(x,z)\right]$. Then for every $\gamma \geq 0$ we have 
		\[
		\Pr_{A \sim V^R} \left[\Big|\mu_A - \Ex_A\mu_A \Big| \geq \gamma \sqrt{\Ex_A \mu_A}\right] \leq \frac{\beta}{\gamma^2}.
		\]
	\end{lemma}
	We point out that \cite{RST12} actually states the above for the random variables $\Ex_{(x,z)} g_A(x,z)$, but it is equivalent to the version stated above since $\Ex_{(x,z)} g_A(x,z) = \Ex_{(x,z)}T_{1 - \eta}g_A(x,z)$ for every $A \in V^R$.
	
	Now instantiating Lemma \ref{lem:bias-conc} with $\gamma = \beta^{1/4}$ yields:
	\[
	\Pr_{A \sim V^R}\left[|\mu_A - \mu| \geq \mu^2\right] \leq \beta^{1/2} \leq \mu^{2r}
	\]
	where the last inequality follows from our choice of $\beta$. 
\end{proof}

Let $V_{\rm nice} := V^1_{\rm nice} \cap V^2_{\rm nice}$ i.e., it is the set of vertices in $V^R$ for which (i) the corresponding averaged long code $g_A$ has small influences and (ii) the average value $\mu_A$ is close to $\mu$. The next lemma is the key technical step of the soundness analysis which says that for any fixing of $A \in V_{\rm nice}$, the corresponding expectation term in \eqref{eqn:dks-rhs1} can be bounded by $O(\mu^r) + \nu$.
\begin{lemma} 		\label{lem:dksh-stab}
	For every $A \in V_{\rm nice}$ we have 
	\[
	\Ex_{(x_j,z_j)^r_{j = 1}}\left[\prod_{j=1}^r T_{1 - \eta} g_{A}({x}_j,z_j) \right] \leq 3.5\mu^r + \nu.
	\]
\end{lemma}
\begin{proof}
	Firstly, observe that the variables $(x_j(i),z_j(i))_{j \in [r]}$ are independent for distinct $i \in [R]$ and for each $i \in [R]$, the random variables $(x_1(i),z_1(i)),\ldots,(x_r(i),z_r(i))$ are jointly distributed as $\cA_{r,\rho}(\Omega,\gamma)$ under the test distribution (as in Definition \ref{defn:dist}). Furthermore, using the definition of $V_{\rm nice}$ it follows that  
	\[
	\max_{i \in [R]} \Inf{i}{T_{1 - \eta}g_A} \leq \tau
	\]
	where $\tau = \tau(\nu,r,\alpha)$ is chosen as in Theorem \ref{thm:ks-stab}. Furthermore, $\mu_A = \Ex_{(x,z) \sim \gamma^R} \left[T_{1 - \eta} g_A(x,z)\right]$ satisfies $\mu_A \in [\mu - \mu^2, \mu + \mu^2]$, and hence 
	\[
	\rho = \frac{1}{2C' r^2 \log (1/\mu)} \leq \frac{1}{C'r^2 \log(1/\mu_A)}.
	\]
	Hence, the function $T_{1 - \eta}g_A$ on the distribution $\gamma^R$ along with our choice of $\rho$ satisfies the conditions of Theorem \ref{thm:ks-stab}. Therefore, instantiating Theorem \ref{thm:ks-stab} with $f = T_{1 - \eta}g_A$ on the probability space $(\Omega,\gamma)$ we get that 
	\[
	\Ex_{(x_j,z_j)^r_{j = 1}}\left[\prod_{j=1}^r T_{1 - \eta} g_{A}({x}_j,z_j) \right] \leq 3\mu^r_A + \nu \leq 3.5\mu^r + \nu, 
	\]
	where the last inequality uses $\mu \leq 2^{-r}$.
\end{proof}
Therefore, continuing with bounding \eqref{eqn:dks-rhs1} we have:
\begin{align}
	&\Ex_{A \sim V^R}\Ex_{(x_j,z_j)^r_{j = 1}}\left[\prod_{j=1}^r T_{1 - \eta} g_{A}({x}_j,z_j) \right] \\
	& \leq \Ex_{A \sim V_{\rm nice}}\Ex_{(x_j,z_j)^r_{j = 1}}\left[\prod_{j=1}^r T_{1 - \eta} g_{A}({x}_j,z_j) \right] + \nu + \mu^{2r} \tag{Lemma \ref{lem:sse-dec1}, \ref{lem:vp-nice}}		\non\\
	& \leq 3.5 \mu^r + \nu + \nu + \mu^{2r}		\tag{Lemma \ref{lem:dksh-stab}} 		\non\\
	& \leq 4\mu^r 		\label{eqn:dks-rhs2}
\end{align}
where the last inequality follows using our choice of $\nu$.

{\bf Cleaning Up}. Finally, stitching together the bounds from \eqref{eqn:dks-rhs},\eqref{eqn:dks-rhs1} and \eqref{eqn:dks-rhs2} we get that 
\begin{align}
\Pr\big[\mbox{ Test Accepts }\big]
& \overset{\eqref{eqn:dks-rhs} + \eqref{eqn:dks-rhs1}}{=} \Ex_{A \sim V^R}\Ex_{(x_j,z_j)^r_{j = 1}}\left[\prod_{j=1}^r T_{1 - \eta} g_{A}(x_j,z_j) \right] 	\non\\
& \overset{\eqref{eqn:dks-rhs2}}{\leq} 4\mu^r, 				\label{eqn:sse-sound}
\end{align}
which concludes the soundness analysis. 

\subsection{Proof of Theorem \ref{thm:dksh-redn}}

Let $G = (V,E)$ be a $(\epsilon,\delta,M)$-\smallsetexpansion~instance as in the statement of Theorem \ref{thm:dksh-redn}. Then if $G$ is a YES instance from the completeness analysis (see \eqref{eqn:sse-comp}) we have that there exists an assignment to the long code tables $\{f_A\}_{A \in V^R}$ such that $\Ex_{(A,x,z)} f(A,x,z) = \mu$ which passes the test with probability at least $\Omega(r^{-3}\mu/\log(1/\mu))$. On the other hand, if $G$ is a NO instance, then the soundness analysis (see \eqref{eqn:sse-sound}) shows that for every assignment to long code tables $\{f_A\}_{A \in V^R}$ satisfying $\Ex_{(A,x,z)} f(A,x,z) = \mu$ passes the test with probability at most $O\left(\mu^r\right)$. Combining the two directions completes the proof of Theorem \ref{thm:dksh-redn}.
 
\section{Auxiliary Lemmas for Theorem \ref{thm:dksh-redn}} \label{sec:sse-dec}

The following lemma is used to bound the fraction of vertices in $V^R$ for which the averaged functions $g_A:\Omega^R \to [0,1]$ have influential coordinates.

\begin{lemma}		\label{lem:sse-dec}
	Let $(\Omega^R,\gamma^R)$ be a product probability space. Let $\{f_A\}_{A \in V^R}$ be functions defined on the product probability space $f_A:\Omega^R \to [0,1]$ that are permutation invariant i.e., for every $A \in V^R, \omega \in \Omega^R$ and permutation $\pi:[R] \to [R]$ we have $f_{\pi(A)}(\pi(\omega)) = f_A(\omega)$. Furthermore, define the averaged functions $g_A := \Ex_{B \sim G^{\otimes R}_\eta(A)} f_A$. Then if $G$ is a NO instance (as in the statement of Theorem \ref{thm:dksh-redn}), we have 
	\[
	\Pr_{A \sim V^R} \left[\max_{i \in [R]} \Inf{i}{T_{1 - \eta} g_A} \geq \tau\right] < \frac{\nu^2}{8r}.
	\]
\end{lemma}

The proof of the above uses the following lemma from \cite{RS10} which is based on the reduction from \ug~to \smallsetexpansion. Towards stating the lemma, let us introduce the noise operator $T^V_{1 - \eta}$ in the product space $V^R$ which is defined as follows. For $A \in V^R$, $\tilde{A} \sim T^V_{1 - \eta}(A)$ is sampled as follows. For every $i \in [R]$, do the following independently: with probability $1 - \eta$, set $\tilde{A}(i) = A(i)$, and with probability $1 - \eta$, set $\tilde{A}(i) \sim V$. 	

\begin{lemma}[Lemma 6.11 + Claim A.1~\cite{RST12}]				\label{lem:sse-label}
	Let $G$ be a graph $G = (V,E)$. Let $(\tilde{A},\tilde{B})$ be a distribution over vertex pairs generated as follows. Sample $A' \sim V^R$ and let $\tilde{A} \sim T^V_{1 - \eta'}(A')$ and let $\tilde{B} \sim G^{\otimes R}_{1 - \eta'}(A')$. Let $F:V^R \to [R]$ be an assignment satisfying the following.
	\begin{equation}			\label{eqn:sse-label}
	\Pr_{\tilde{A},\tilde{B}} \Pr_{\pi_{\tilde{A}},\pi_{\tilde{B}} \sim \mathds{S}_R} \left[\pi^{-1}_{\tilde{A}}\left(F\left(\pi_{\tilde{A}} \circ \tilde{A}\right)\right) = \pi^{-1}_{\tilde{B}}\left(F\left(\pi_{\tilde{B}} \circ \tilde{B}\right)\right) \right] \geq \zeta.
	\end{equation}
	Then there exists a set $S \subset V$  with ${\sf vol}(S) \in \left[\frac{\zeta}{16R},\frac{3}{\eta' R}\right]$ such that $\phi_G(S) \leq 1- \zeta/16$.
\end{lemma} 

We shall also need the following easy observation.

\begin{observation}				\label{obs:d-match}
Consider the distribution on $(\tilde{A},\tilde{B})$ from Lemma \ref{lem:sse-label} instantiated with $\eta'$ satisfying $1 - \eta' = \sqrt{1 - \eta}$. Then the distribution over $(A,B)$ sampled as $A \sim V^R$ and $G^{\otimes R}_{\eta}(A)$ is identical to that of $(\tilde{A},\tilde{B})$. 
\end{observation}
\begin{proof}
It suffices to prove the claim for $R = 1$. Firstly, note that in the setting of Lemma \ref{lem:sse-label}, the random vertex $\tilde{A}$ is marginally distributed uniformly over $V$. Since the random walk corresponding to operator $T^V_{1 - \eta'}$ is reversible, fixing $\tilde{A}$, observe that we have that $A'$ is distributed as $A' \sim T^V_{1-\eta'}(\tilde{A})$. Furthermore, fixing $A'$, we have $\tilde{B} \sim T^V_{1 - \eta'} \circ G(A')$. Overall, fixing $\tilde{A}$ we have $\tilde{B} \sim T^V_{1 - \eta'} \circ G \circ T^V_{1 - \eta'}(\tilde{A})$. Furthermore, since operators $T^V_{1 - \eta'}$ and $G$ are reversible, they commute, and hence we have 
\[
T^V_{1 - \eta'} \circ G \circ T^V_{1 - \eta'} = T^V_{1 - \eta'} \circ T^V_{1 - \eta'} \circ G = T^V_{(1 -\eta')^2} \circ G = G_\eta.
\]
Hence, we can equivalently think of the $(\tilde{A},\tilde{B})$ pair as being generated as $\tilde{A} \sim V$ and $\tilde{B} \sim G_\eta(\tilde{A})$ which establishes the claim.
\end{proof}
We now use the above to prove Lemma \ref{lem:sse-dec}.
\begin{proof}[Proof of Lemma \ref{lem:sse-dec}]	
	The proof of the lemma again goes through the standard influence decoding argument. For contradiction, assume that 
	\[
	\Pr_{A \sim V^R} \left[\max_{i \in [R]} \Inf{i}{T_{1 - \eta}g_A} \geq \tau\right] \geq \frac{\nu^2}{8r}
	\]
	For every $A \in V^R$, define the following sets
	\[
	L_{A,1} := \left\{ i \in [R] \Big| \Inf{i}{T_{1 - \eta} f_A}\geq \frac{\tau}{2} \right\}
	\ \ \ \ \ \ \textnormal{ and } \ \ \ \ \ 
	L_{A,2} := \left\{ i \in [R] \Big| \Inf{i}{T_{1 - \eta} g_A}\geq \frac{\tau}{2} \right\}.
	\]
	Now consider the following randomized construction of $F:V^R \to [R]$. For every $A \in V^R$ do the following randomly.
	\begin{itemize}
		\item W.p. $1/2$, if $L_{A,1} \neq \emptyset$, set $F(A) \sim L_{A,1}$, otherwise set $F(A)$ arbitrarily.
		\item W.p. $1/2$, if $L_{A,2} \neq \emptyset$, set $F(A) \sim L_{A,2}$, otherwise set $F(A)$ arbitrarily.
	\end{itemize}
	We now bound the expected value of the LHS of \eqref{eqn:sse-label} under the randomized construction of $F$. Towards that, define $V' \subset V^R$ to be the set of $A$'s for which $L_{A,2} \neq \emptyset$. For any such $A \in V^R$ with $L_{A,2} \neq \emptyset$, there exists coordinate $i_A \in [R]$ for which $\Inf{i_A}{T_{1 - \eta} g_A} \geq \tau$. Then using the convexity of influences we have 
	\[
	\tau \leq \Inf{i_A}{T_{1 - \eta}g_A} = \Inf{i_A}{\Ex_{B \sim G^{\otimes R}_\eta(A)}T_{1 - \eta}f_B} \leq \Ex_{B \sim G^{\otimes R}_\eta(A)} \left[\Inf{i_A}{T_{1 - \eta}f_B}\right], 
	\]
	which in turn by averaging implies that 
	\[
	\Pr_{B \sim G^{\otimes R}_\eta(A)} \left[\Inf{i_A}{T_{1 - \eta}f_B} \geq \frac{\tau}{2}\right] \geq \frac{\tau}{2}.
	\]
	Then for any such $A \in V'$, we identify $\cN(A) \subset V^R$ as the subset of vertices $B$ for which $i_A$ is at least $\tau$-influential in $T_{1 - \eta}f_B$. 
	
	Furthermore, it is folklore that for any function $g_A$ with variance at most $1$ (Lemma \ref{lem:inf-size}), we have 
	\[
	\left|\left\{i \in [R] : \Inf{i}{T_{1 - \eta}g_A} > \tau \right\}\right| \leq \frac{1}{\eta\tau}
	\]
	This implies that $|L_{A,i}| \leq 2/\tau\eta$ for any $A \in V^R$ and $i = 1,2$.  Now note that since the global assignment $f:V^R \times \Omega^R \to [0,1]$ is permutation invariant, for any $\pi:[R] \to [R]$ we have 
	\[
	f_{\pi(A)}(\omega) = f_{\pi(A)}\left(\pi \circ \pi^{-1}(\omega)\right) = f_A(\pi^{-1}(\omega)) \qquad\qquad \forall A \in V^R, \omega \in \Omega^R,
	\]
	which implies that 
	\[
	L_{\pi(A),i} = \left\{\pi(j) | j \in L_{A,i}\right\}
	\]
	for $i = 1,2$, i.e., the label lists are permutation invariant. Hence, for any fixed $A \in V^R$, and any fixed permutation $\pi:[R] \to [R]$, randomizing over the choice of $F$, we have that $\pi^{-1}(F(\pi \circ A))$ is identically distributed to $F(A)$. Equipped with this observation, we now proceed to bound the LHS of \eqref{eqn:sse-label} under the randomized assignment $F$. Let $(\tilde{A},\tilde{B})$ be distributed over $V^R \times V^R$ as in Lemma \ref{lem:sse-label} with $\eta'$ satisfying $1 - \eta' = \sqrt{1 - \eta}$. Then,
	\begin{align*}
	&\Ex_F\Ex_{(\tilde{A},\tilde{B})} \Ex_{\pi_{\tilde{A}},\pi_{\tilde{B}}} \left[\mathbbm{1} \left(\pi^{-1}_{\tilde{A}}\left(F\left(\pi_{\tilde{A}} \circ \tilde{A}\right)\right) = \pi^{-1}_{\tilde{B}}\left(F\left(\pi_{\tilde{B}} \circ \tilde{B}\right)\right) \right)\right]	\\
	&= \Ex_F\Ex_{A \sim V^R}\Ex_{B \sim G^{\otimes R}_\eta(A)} \Ex_{\pi_A,\pi_B} \left[\mathbbm{1} \left(\pi^{-1}_A\left(F\left(\pi_A \circ A\right)\right) = \pi^{-1}_B\left(F\left(\pi_B \circ B\right)\right) \right)\right]		\tag{Observation \ref{obs:d-match}} \\
	&\geq \Pr\left[A \in V'\right]\Ex_F\Ex_{A \sim V' }\Ex_{B \sim G^{\otimes R}_\eta(A)} \Ex_{\pi_A,\pi_B} \left[\mathbbm{1} \left(\pi^{-1}_A\left(F\left(\pi_A \circ A\right)\right) = \pi^{-1}_B\left(F\left(\pi_B \circ B\right)\right) \right)\right] \\ 
	&\geq \frac{\nu^2}{8r}\Ex_F\Ex_{A \sim V' } \Pr_{B \sim G^{\otimes R}_\eta(A)}\left[B \in \cN(A)\right]\Ex_{B \sim \cN(A)} \Ex_{\pi_A,\pi_B}\left[\mathbbm{1} \left(\pi^{-1}_A\left(F\left(\pi_A \circ A\right)\right) = \pi^{-1}_B\left(F\left(\pi_B \circ B\right)\right) \right)\right] \\ 
	&\geq \frac{\nu^2\tau}{16r}\Ex_{A \sim V' }\Ex_{B \sim \cN(A)} \Ex_F\left[\mathbbm{1} \left(F(A) = F(B) = i_A \right)\right] 		\tag{Permutation Invariance}\\
	&= \frac{\nu^2\tau}{16r}\Ex_{A \sim V' }\Ex_{B \sim \cN(A)} \Ex_F\left[\frac{1}{|L_{A,2}|}\cdot \frac{1}{|L_{B,1}|} \right] 	\\
	&\geq \frac{\nu^2\eta^2\tau^3}{c r}.  
	\end{align*}
	where $c = 64$. The above computation implies that there exists a labeling $F:V^R \to [R]$ for which the RHS of \eqref{eqn:sse-label} is at least $\zeta:= \nu^2 \eta^2\tau^3/cr$. Therefore, using Lemma \ref{lem:sse-label}, we get that there exists a set $S \subset V$ such that 
	\begin{align*}
		{\sf Vol}(S) &\in \left[\frac{\nu^2\eta^2\tau^3}{16cr\cdot R},\frac{3}{\nu^2 \eta' R}\right] \\
					 & \subset  \left[\frac{\beta\nu^2\eta^2\tau^3}{16cr} \cdot \delta,\frac{6\beta r}{\nu^2 \eta}\cdot \delta\right] 	\tag{Since $R = 1/(r\beta\delta)$,$\eta' \geq \eta/2$}  \non\\
					 &\subseteq  \left[\frac{\delta}{M} , M \delta\right], 		\tag{From our choice of $M$}		\non\\
	\end{align*}
	and  $\phi_G(S) \leq 1- \nu^2\eta^2\tau^3/(16 cr) < 1 - \epsilon$ from our choice of $\epsilon$, which contradicts the NO case guarantee of $G$.
\end{proof}

{\bf Proof of Lemma \ref{lem:sse-dec1}}. Using the above, we can complete the proof of Lemma \ref{lem:sse-dec1}.

\begin{proof}[Proof of Lemma \ref{lem:sse-dec1}]
	Recall that $V_{\rm nice}$ is defined as 
	\[
	V_{\rm nice} := \left\{A \in V^R \Big| \max_{i \in [R]} \Inf{i}{T_{1 - \eta}g_A} \geq \tau \right\}
	\]
	where $g_A = \Ex_{B \sim G^{\otimes R}(A)} \bar{f}_B$ are defined over the product probability $L_2(\Omega^R,\gamma^R)$ and $\bar{f}_A$ is defined as 
	\[
	\bar{f}_A(x,z) = \Ex_{(B',x')\sim M_z(B,x)} \Ex_{\pi \sim \mathbbm{S}_R}\left[f\left(\pi \circ (B',x',z)\right)\right].
	\]
	Note that by definition of $\bar{f}_A$, for any permutation $\pi':[R] \to [R]$ we have 
	\[
	\bar{f}_A(\pi'\circ(A,x,z)) = \bar{f}_A(A,x,z)
	\]
	i.e, $\{\bar{f}_A\}_{A \in V^R}$ is a family of permutation invariant functions satisfying the premise of Lemma \ref{lem:sse-dec}. Hence, using Lemma \ref{lem:sse-dec} we get that 
	\[
	\frac{|V^c_{\rm nice}|}{|V|^R} = \Pr_{A \sim V^R} \left[\max_{i \in [R]} \Inf{i}{T_{1 - \eta}g_A} \geq \tau \right] \leq \nu,
	\]
	which finishes the proof.
\end{proof}

\section{Hardness of \mkcsp}

We establish our hardness result using a factor preserving reduction from the following variant of \uniquegames.

\begin{definition}[$(\epsilon_c,\epsilon_s)$-\uniquegames]
	A \ug~instance $\cG(V_\cG,E_\cG,[t],\{\pi_e\}_{e \in E})$ is a $2$-CSP on the vertex set $V_{\cG}$, edge set $E_{\cG}$ and label set $[t]$. Each edge $e = (u,v)$ is identified with a bijection constraint $\pi_{u \to v}:[t] \to [t]$. A labeling of the vertices $\sigma: V_{\cG} \to [t]$ satisfies the edge constraint $e = (u,v)$ if and only if $\pi_{u \to v}(\sigma(u)) = \sigma(v)$. The objective of a {\sc Unique Games} instance is to find a labeling $\sigma:V_\cG \to [t]$ which satisfies the maximum fraction of constraints -- we will call this maximum fraction as as the value of the {\sc Unique Game} and denote it by ${\sf Opt}(\cG)$.   
	
	In particular, a $(\epsilon_c,\epsilon_s)$-\uniquegames~instance $\cG = (V,E,[t],\{\pi_{u \to v}\}_{(u,v) \in E})$ is the decision problem where the objective is to distinguish between the two cases:
	\[
	\textnormal{YES Case:} \ \ \ {\sf Opt}(\cG) \geq 1 - \epsilon_c
	\qquad\qquad\textnormal{and} \qquad\qquad 
	\textnormal{NO Case:} \ \ \ {\sf Opt}(\cG) \leq \epsilon_s.
	\]
\end{definition}

Our hardness reduction starts from the hard instance of \ug~using Khot's Unique Games Conjecture.

\begin{conjecture}[Unique Games Conjecture~\cite{Khot02a}]
	There exists constant $\epsilon_0 \in (0,1)$ such that the following holds. For every fixed choice of $\epsilon_c,\epsilon_s \in (0,\epsilon_0)$, there exists $t = t(\epsilon_c,\epsilon_s)$ such that the $(\epsilon_c,\epsilon_s)$-\uniquegames~problem is hard on label sets of size $t$.
\end{conjecture} 

Our main result follows directly from the following theorem which gives a factor preserving reduction from \ug~to {\sc Max}-$k$-{\sc CSP}.

\begin{theorem}					\label{thm:ug-redn}
	The following holds for every $R,k$ such that $R > 2^k$. Let $\epsilon_c =2/10k$ and $\epsilon_s = (Rk)^{-4R^2k^2}$. Then there exists a polynomial time reduction from an $(\epsilon_c,\epsilon_s)$-\uniquegames~instance $\cG = (V,E,[t],\{\pi_e\}_{e \in E})$ to a {\sc Max}-$k$-CSP instance $\Psi(V,E,[R],\{\Pi_e\}_{e \in E})$ such that the following properties hold.
	\begin{itemize}
		\item {\bf Completeness}. If $\cG$ is a YES instance, then ${\sf Opt}(\Psi) \geq \frac{C_1}{k^2\log(R)} $.
		\item {\bf Soundness}. If $\cG$ is a NO instance, then ${\sf Opt}(\Psi) \leq C_2  R^{-(k-1)}$.  
	\end{itemize}
	Here $C_1,C_2 > 0$ are absolute constants independent of $R$ and $k$. 
\end{theorem}

\subsection{The PCP verifier and its Analysis}

The hard instance of the CSP is going to be a hypergraph variant of {\sc UniqueGames}, and the test itself is similar to the test from Theorem \ref{thm:dksh-redn}, the key difference is that here the underlying gadget is the noisy $R$-ary hypercube instead of the $\mu$-biased hypercube. Furthermore, both the test and the analysis here are relatively simpler, since we work with {\sc Unique Games} as the outer verifier instead of \smallsetexpansion. Here the distribution over hyperedge constraints is given by the following dictatorship test. 

\begin{figure}[ht!]
	\begin{mdframed}
		\begin{center}
			\underline{\bf PCP Test for $k$-CSP} \\[5pt]
		\end{center}	
		{\bf Test}: 
		\begin{enumerate}
			\item Sample $v \sim V_\cG$ and $w_1,\ldots,w_k \sim N_\cG(v)$.
			\item Let $\tilde{f}_{w_1},\ldots,\tilde{f}_{w_k}:[R]^t \to [R]$ be the corresponding folded long codes as defined in \eqref{eqn:folded}.
			\item Sample $z \sim [R]^t$ uniformly.
			\item Sample $[R]^t$-valued random variables $z_1,\ldots,z_t$ as follows. For every $i \in [t]$ do the following independently:
			\begin{itemize}
				\item Sample $\theta(i) \sim \{0,1\}_\rho$.
				\item If $\theta(i) = 1$, then for every $j \in [k]$, set $z_j(i) = z(i)$.
				\item If $\theta(i) = 0$, then for every $j \in [k]$, sample $z_j(i) \sim [R]$ independently.
			\end{itemize}
			\item For every $i \in [k]$, let $z'_i \underset{1 - \eta}{\sim} z_i$ in the $[R]^t$ space.
			\item Accept {\em iff} for every $i,j \in [k]$ we have  
			\[
			\tilde{f}_{w_i}\left(\pi_{w_i \to v} \circ z'_i\right) = \tilde{f}_{w_j} \left( \pi_{w_j \to v} \circ z'_j\right)
			\]
		\end{enumerate}
	\end{mdframed}
	\caption{PCP Verifier for Hypergraph Unique Games}
	\label{fig:csp-test}
\end{figure}

The parameters used in the above reduction and its analysis are set to the following values.

\begin{itemize}
	\item $\rho = 1/(C'k^2\log(R))$ where $C'$ is as in Theorem \ref{thm:ks-stab}.
	\item $\nu = R^{-2k}$.
	\item $\eta = \frac{1}{k^2R}$.
	\item $\alpha = \rho R^{-k}$.
	\item $\tau = \tau(\nu,k,\alpha)$ as in Theorem \ref{thm:ks-stab}. 
\end{itemize}	

Note that our setting of $\tau$ depends only on $R$ and $k$. Now we analyze the above reduction.

\subsection{Completeness}				\label{sec:ug-comp}

Suppose ${\sf Opt}(\cG) = 1 - \epsilon_c$, and let $\sigma:V_\cG \to [t]$ be the labeling which achieves the optimal value. For every vertex $v \in V_\cG$, we define $f_{v} := \Lambda_{\sigma(v)}$ to be the $\sigma(v)^{th}$-dictator function. It is easy to see that with probability at least $1 - \epsilon_c k$ over the choices of vertices $v,w_1,\ldots,w_k$ in the test, we have $\sigma(v)  = \pi_{w_i \to v}(\sigma(w_i))$ for every $i \in [k]$. Furthermore, under the test distribution we have  
\begin{align*}
&\Pr_{\theta(\sigma(v))}\Big[\theta(\sigma(v)) = 1\Big] = \rho.
\end{align*}
Additionally, conditioned on $\theta(\sigma(v)) = 1$,  we have  $z'_{1}(\sigma(v)) = z'_2(\sigma(v)) = \cdots = z'_{k}(\sigma(v))$ with probability at least $1 - \eta k$, losing an additional factor of $\eta k$ due to the $(1 - \eta)$-correlated sampling of $z'_j{(\sigma(v))}$ from $z_j{(\sigma(v))}$. Therefore conditioned on the event ``all the permuted labels match'' i.e., $\pi_{w_i \to v} (\sigma(w_i)) = \sigma(v)$, we can bound the probability of the test accepting as:
\begin{align*}			
&\Pr\left[\forall \ i,j, \ \tilde{f}_{w_i}\left(\pi_{w_i \to v} \circ z'_i\right) = \tilde{f}_{w_j} \left( \pi_{w_j \to v} \circ z'_j\right)\right]  		\non\\
&=\Pr\left[\forall \ i,j, \ f_{w_i}\left(\pi_{w_i \to v} \circ z'_i\right) = f_{w_j} \left( \pi_{w_j \to v} \circ z'_j\right)\right]		\tag{Proposition \ref{prop:fold}}  		\non\\
&= \Pr\left[\forall \ i,j, \ z'_i(\sigma(v)) = z'_j(\sigma(v))\right] 		\non\\
&\geq \Pr_{{\theta}}\Big[\theta(\sigma(v)) = 1\Big]\Pr\left[\forall \ i,j, \ z'_i(\sigma(v)) = z'_j(\sigma(v)) \Big| \theta(\sigma(v))= 1\right] 		\non\\
& \geq \rho(1 - \eta k). 
\end{align*}
Since the above bound holds for any choice of $v,w_1,\ldots,w_j$ for which the permuted labels match, overall we have:
\begin{equation}		\label{eqn:ug-comp}
\Pr\left[\mbox{ Test Accepts }\right] \geq (1 - \epsilon_c k) \cdot{\rho (1 - \eta k)} \geq \frac{\rho}{2}
\end{equation}
where the last inequality follows from our choice of parameters $\epsilon_c,\eta$ and $k$. 

\subsection{Soundness}					\label{sec:ug-sound}

Let $\cG$ be a NO instance. Let $\{f_v\}_{v \in \cG}$ be a set of long codes. Since the test in Figure \ref{fig:csp-test} only queries positions with respect to the folded code, without loss of generality, we may assume that the long codes are folded and hence $\tilde{f}_v = f_v$ for every $v \in V_{\cG}$. Given a long code $f_v:[R]^t \to [R]$, as described in Section \ref{sec:r-ary}, we can write $f_v = \left(f^{(1)}_v,\ldots,f^{(R)}_v\right)$, where $f^{(j)}_v:[R]^t \to [0,1]$ is the $j^{th}$ coordinate function. Using the above interpretation, we arithmetize the probability of the test accepting as
\begin{align*}
	\Pr\left[\mbox{ Test Accepts}\right]
	&= \sum_{i = 1}^R \Ex_{v \sim V_{\cG}} \Ex_{w_1,\ldots,w_k \sim N_{\cG}(v)} \Ex_{z'_1,\ldots,z'_k} \left[\prod_{j \in [k]} f^{(i)}_{w_j}\left(\pi_{w_j \to v} \circ z_j\right)\right] \\
	&\overset{1}{=} \sum_{i = 1}^R \Ex_{v \sim V_{\cG}} \Ex_{w_1,\ldots,w_k \sim N_{\cG}(v)} \Ex_{z_1,\ldots,z_k}  \left[\prod_{j \in [k]}\Ex_{z'_i \underset{1 - \eta}{\sim} z_i }f^{(i)}_{w_j}\left(\pi_{w_j \to v} \circ z_j\right)\right] \\
	&= \sum_{i = 1}^R \Ex_{v \sim V_{\cG}} \Ex_{w_1,\ldots,w_k \sim N_{\cG}(v)} \Ex_{z_1,\ldots,z_k} \left[\prod_{j \in [k]} T_{1 - \eta} f^{(i)}_{w_j}\left(\pi_{w_j \to v} \circ z_j\right)\right]
\end{align*}
where step $1$ uses the fact fixing $z_i$, for every $i \in [k]$ every $z'_i$ is an independent $(1 - \eta)$-correlated copy of $z_i$ and in the last step, $T_{1 - \eta}$ denotes the $(1 - \eta)$ correlated noise operator for the inner product space corresponding to the uniform distribution over $[R]^t$. Now, for every $v \in V_{\cG}$, and coordinate $i \in [R]$, define the averaged function
\[
g^{(i)}_v(z) \defeq \Ex_{w \sim N_\cG(v)} \left[f^{(i)}_v\left(\pi_{v \to w} \circ z\right)\right].
\]
Then using the fact that for every $j \in [k]$, $w_j$ is an independently chosen random neighbor of $v$ we can write 
\begin{align}
	&\sum_{i = 1}^R \Ex_{v \sim V_{\cG}} \Ex_{w_1,\ldots,w_k \sim N_{\cG}(v)} \Ex_{z_1,\ldots,z_k} \left[\prod_{j \in [k]} T_{1 - \eta} f^{(i)}_{w_j}\left(\pi_{w_j \to v} \circ z_j\right)\right] \\
	&\sum_{i = 1}^R \Ex_{v \sim V_{\cG}}  \Ex_{z_1,\ldots,z_k} \left[\prod_{j \in [k]}\Ex_{w_j \sim N_{\cG}(v)} T_{1 - \eta} f^{(i)}_{w_j}\left(\pi_{w_j \to v} \circ z_j\right)\right] \\
	&= \sum_{i = 1}^R \Ex_{v \sim V_{\cG}} \Ex_{z_1,\ldots,z_k }\left[\prod_{j \in [k]} T_{1 - \eta} g^{(i)}_{v}\left(z_j\right)\right].			\label{eqn:csp-eq1} 
\end{align}

{\bf Invariance Principle Step}. As in the soundness analysis of Theorem \ref{thm:dksh-redn}, we shall use the fact that for most choices of $v \in V_\cG$, the corresponding averaged coordinate functions $g^{(1)}_{v},\ldots,g^{(R)}_v$ have low influences, and for every such collection of functions, we can bound the summation over expectation terms by $O(R^{-k+1})$. Towards that, define $V_{\rm nice} \subset V_{\cG}$ as 
\[
V_{\rm nice} := \left\{ v \in V_{\cG} \Big|  \max_{\ell \in [R]}\max_{i \in [t]} \Inf{i}{{T}_{1 - \eta} g^{(\ell)}_{v}} > \tau \right\},
\]
where parameters $\nu,\tau,\eta$ are as defined below Figure \ref{fig:csp-test} and $T_{1 - \eta}$ is the $(1 - \eta)$-correlated noise operator in the space $[R]^t$. Since $\cG$ is a NO instance, then for most choices of $v \in V_{\cG}$, the corresponding averaged long codes will have small influences (otherwise we can decode a good labeling for $\cG$). This is stated formally as the following lemma:
\begin{lemma}			\label{lem:ug-dec1}
	Suppose $\cG$ is a NO instance as in the setting of Theorem \ref{thm:ug-redn}. Then,
	\[
	\left| \left\{ v \in V_{\cG} \Big| \max_{\ell \in [R]} \max_{i \in [t]} \Inf{i}{{T}_{1 - \eta} g^{(\ell)}_{v}} > \tau \right\} \right| \leq \nu \cdot|V_\cG|.
	\]
\end{lemma}
We defer the proof of the above lemma to Section \ref{sec:ug-dec}. As a next step, as in the soundness analysis of Theorem \ref{thm:dksh-redn}, we  will now show again that for every $v \in V_{\rm nice}$, we can bound the corresponding expectation term with $O(R^{-k + 1})$:
\begin{lemma}				\label{lem:csp-clt}
	For every $v \in V_{\rm nice}$ we have 
	\begin{equation}				\label{eqn:csp-clt}	
		\sum_{i = 1}^R \Ex_{(z_1,\ldots,z_k) }\left[\prod_{j \in [k]} T_{1 - \eta} g^{(i)}_{v}\left(z_j\right)\right]
		\leq 3.5 R^{-k + 1}.
	\end{equation}  
\end{lemma}

\begin{proof}
	Fix a vertex $v \in V_{\rm nice}$. We shall instantiate Theorem \ref{thm:ks-stab} as follows. Let $(\Omega,\gamma)$ denote the uniform distribution on $[R]$. Then $g^{(1)}_v,\ldots,g^{(R)}_v : \Omega^t \to [0,1]$ are all functions on the $t$-wise product probability space $L_2(\Omega^t,\gamma^t)$ satisfying
	\[
	\max_{j \in [t]} \Inf{j}{T_{1 - \eta}g^{(i)}_v} \leq \tau
	\]
	using the definition of $V_{\rm nice}$. Furthermore, by folding (see Proposition \ref{prop:fold}) we have $\mu:= \Ex_{\omega \sim \gamma^t}\left[T_{1  - \eta}g^{(i)}(\omega)\right] = 1/R$ for every $i \in [R]$ and hence our choice of $\rho$ satisfies 
	\[
	\rho = 1/(C'k^2\log(R)) = 1/(C'k^2\log(1/\mu)) .
	\]
	Therefore, for any $i \in [R]$, the function $f = T_{1 - \eta}g^{(i)}_v$ along with the choice of $\rho$ satisfies the conditions of Theorem \ref{thm:ks-stab} and hence we have
	\[
	\Ex_{(z_1,\ldots,z_k) }\left[\prod_{j \in [k]} T_{1 - \eta} g^{(i)}_{v}\left(z_j\right)\right]
	\leq 3R^{-k} + \nu. 
	\]
	Hence applying the above argument point-wise for every $i \in [R]$ we get that 
	\[
	\sum_{i \in [R]}\Ex_{(z_1,\ldots,z_k) }\left[\prod_{j \in [k]} T_{1 - \eta} g^{(i)}_{v}\left(z_j\right)\right]
	\leq 3R \left(R^{-k} + \nu\right) \leq 3.5R^{-k + 1},
	\]
	where the last step follows from $\nu \leq R^{-2k}$ from our choice of parameters.
\end{proof}	

Therefore, using the above we now proceed to upper bound \eqref{eqn:csp-eq1}.
\begin{align}
\sum_{i = 1}^R \Ex_{v \sim V_{\cG}} \Ex_{z_1,\ldots,z_k }\left[\prod_{j \in [k]} T_{1 - \eta} g^{(i)}_{v}\left(z_j\right)\right]
&\leq \Ex_{v \sim V_{\rm nice}}\Big[ 3R^{-k + 1}\Big] + 4R\nu \tag{Lemma \ref{lem:ug-dec}}  \non\\
&\leq 3.5R^{-k + 1} + 4R\nu. 		\tag{Lemma \ref{lem:csp-clt}}		\\
&\leq 4R^{-k + 1} \label{eqn:csp-eq2}
\end{align} 
where the last inequality again follows using $\nu \leq R^{-2k}$ from our choice of parameters.

{\bf Putting Things Together}. Therefore stitching together the bounds from \eqref{eqn:csp-eq1} and \eqref{eqn:csp-eq2}, we get that 
\begin{align}
\Pr\left[\mbox{ Test Accepts }\right]
&\overset{\eqref{eqn:csp-eq1}}{=} \sum_{i = 1}^R \Ex_{v} \Ex_{z_1,\ldots,z_k }\left[\prod_{j \in [k]} T_{1 - \eta} g^{(i)}_{v}\left(z_j\right)\right] \\
&\overset{\eqref{eqn:csp-eq2}}{\leq} 4R^{-k + 1},		\label{eqn:ug-sound}
\end{align}
where the last inequality follows from our choice of $\nu$.

\subsection{Proof of Theorem \ref{thm:ug-redn}}

Let $\cG = (V_{\cG},E_{\cG},[t],\{\pi_e\}_{e \in E})$ be a $(\epsilon_c,\epsilon_s)$-\uniquegames~instance with $\epsilon_c,\epsilon_s$ chosen as in Theorem \ref{thm:ug-redn}. Then, if $\cG$ is a YES instance, the completeness analysis (see Eq. \eqref{eqn:ug-comp}) shows that there exists a choice of assignment to the long code tables $\{f_v\}_{v \in V_{\cG}}$ which passes the test with probability at least $\rho/2 = \frac{1}{4 k^2 \log R}$. On the other hand, if $\cG$ is a NO instance, then the soundness analysis (see Eq. \eqref{eqn:ug-sound}) shows that for any assignment to the long code tables $\{f_v\}_{v \in \cG}$, the test accepts with probability at most $4R^{- k + 1}$. Combining the two guarantees establishes Theorem \ref{thm:ug-redn}.

\subsection{Proof of Lemma \ref{lem:ug-dec1}}			\label{sec:ug-dec}

\begin{lemma}[Folklore]					\label{lem:ug-dec}
	The following holds for any probability space $(\Omega,\mu)$ and $t \in \mathbbm{N}$ large enough. Let $\cG = (V,E,[t],\{\pi_{e}\}_{e \in E})$ be a NO instance with ${\sf Opt}(\cG) < \nu\eta^2\tau^3/(16Rk)$. Let $\{f_v\}_{v \in  V_{\cG}}$ be a family of long codes $f_v:\Omega^t \to [0,1]$ defined over product probability space $(\Omega^t,\mu^t)$. Furthermore, we define the averaged function as $g_v = \Ex_{w \sim N_\cG(v)}\left[\pi_{w \to v} \circ f_v \right]$. Then,
	\[
	\Pr_{v \sim V_{\cG}} \left[\max_{i \in [R]}\max_{j \in [t]} \Inf{j}{T_{1 - \eta}g^{(i)}_v} > \tau \right] \leq \nu.
	\]
\end{lemma}
\begin{proof}
	We again use the influence decoding argument. For contradiction, let us assume that 
	\[
	\Pr_{v \sim V_{\cG}} \left[\max_{\ell \in [R]}\max_{j \in [t]} \Inf{j}{T_{1 - \eta}g^{(\ell)}_v} > \tau \right] \geq \nu.
	\]
	Then by averaging, there exists a choice of $\ell \in [R]$ for which 
	\[
	\Pr_{v \sim V_{\cG}} \left[\max_{j \in [t]} \Inf{j}{T_{1 - \eta}g^{(\ell)}_v} > \tau \right] \geq \frac{\nu}{R}.
	\]
	Denote $V_{\rm bad} \subset V_{\cG}$ as the subset of vertices $v \in V_{\cG}$ for which 
	\[
	\max_{j \in [t]} \Inf{j}{T_{1 - \eta}g^{(\ell)}_v} > \tau.
	\]
	Now, for any fixed $v \in V_{\rm bad}$. Then $i_v \in [t]$ such that $\Inf{i_v}{T_{1 - \eta} g^{(\ell)}_v} \geq \tau$ and hence using the convexity of influences we have 
	\[
	\tau \leq \Inf{i_v}{\Ex_{w \sim N_{\cG}}\left[\pi_{v \to w} \circ T_{1 - \eta}f_w\right]} 
		\leq \Ex_{w \sim N_{\cG}(v)}\left[\Inf{i_v}{\pi_{v \to w} \circ T_{1 - \eta} f_w} \right]
	\]
	which again by averaging implies that 
	\begin{equation}			\label{eqn:large-inf-prob}
	\Pr_{w \sim N_{\cG}(v)} \left[\Inf{i_v}{ \pi_{v \to w} \circ T_{1 - \eta}f_w} \geq \frac{\tau}{2}\right] \geq \frac{\tau}{2}
	\end{equation}
	for any fixed choice of $v \in V_{\rm bad}$. For every $v \in V_{\rm bad}$, let $S(v) \subset N_{\cG}(v)$ be the subset of vertices for whic hwe have $\Inf{i_v}{\pi_{v \to w} \circ T_{1 - \eta}f_w} \geq \tau/2$. Now define the lists $L_{v,1}$ and $L_{v,2}$ as follows.
	\[
	L_{v,1} := \left\{i \in [t] \Big| \Inf{i}{T_{1 - \eta}g_v} \geq \frac{\tau}{2}\right\}
	\qquad\qquad\mbox{and}\qquad\qquad
	L_{v,2} := \left\{i \in [t] \Big| \Inf{i}{T_{1 - \eta}f_v} \geq \frac{\tau}{2}\right\}
	\]
	Note that again we must have $|L_{v,1}|,|L_{v,2}| \leq 2/\epsilon\eta$ for every $v \in V_{\cG}$. Now consider the following randomized decoding scheme. For every $v \in V_{\cG}$ do the following independently. Sample $a \sim \{1,2\}$. Then if $L_{v,a} \neq \emptyset$ then assign $\sigma(v) \sim L_{v,a}$, otherwise assign $\sigma(v)$ arbitrarily. Now we proceed to bound the expected fraction of constraints satisfied by this labeling:
	\begin{align*}
		&\Ex_{v \sim V_{\cG}} \Ex_{w \sim N_{\cG}(v)} \Ex_\sigma \left[ \sigma(v) = \pi_{w \to v}(\sigma(w))\right] \\
		&\geq \frac{\nu}{R}\Ex_{v \sim V_{\rm bad}} \Ex_{w \sim N_{\cG}(v)} \Ex_\sigma \left[ \sigma(v) = \pi_{w \to v}(\sigma(w))\right] \\
		&\geq \frac{\nu}{R}\cdot\frac{\tau}{2}\Ex_{v \sim V_{\rm bad}} \Ex_{w \sim S(v)} \Ex_\sigma \left[ \sigma(v) = \pi_{w \to v}(\sigma(w))\right] \\
		&\geq \frac{\nu}{R}\cdot\frac{\tau}{2}\Ex_{v \sim V_{\rm bad}} \Ex_{w \sim S(v)} \Pr_\sigma \left[ \sigma(v) = i_v \wedge \sigma(w) = \pi_{v \to w}(i_v)\right] \\
		&\geq \frac{\nu}{R}\cdot\frac{\tau}{2}\Ex_{v \sim V_{\rm bad}} \Ex_{w \sim S(v)}  \left[ \frac{1}{|L_{v,1}|}\cdot\frac{1}{|L_{v,2}|}\right] \\
		&\geq \frac{\nu}{R}\cdot\frac{\tau^3\eta^2}{8} \\ 
		&> {\rm Opt}(\cG)
	\end{align*}
	thus giving us a contradiction.
\end{proof}

\part{Approximation Algorithm for \dksh}

In this section, we give our approximation guarantee for {\sc Densest}-$k$-{\sc SubHypergraph} problems. Let $\delta_{k}(\cdot)$ (and $\delta^{(r)}_k(\cdot)$) denote the optimal value of the D$k$S problem (and the D$k$SH problem) on graphs and hypergraph of arity-$k$ respectively, subject to the constraint that the set is of size at most $k$. The following theorem states our guarantee for \dksh.

\dkshalg*

The above approximation guarantee matches the SSEH based lower bound from Theorem \ref{thm:dksh-hard} up to constant factors for every constant arity $r$. 
Our approximation algorithm for the above theorem will use the following known $\Omega(\mu \log(1/\mu))$-approximation guarantee for \dks~as a black-box.
\begin{theorem}[\cite{CHK11} + \cite{KKT15}]		\label{thm:dks-alg}
	The following holds for every $\mu \in (0,1)$. Given a graph $G = (V,E)$, there exists a randomized algorithm which r	uns in time $|V|^{{\rm poly}(1/\mu)}$ and outputs a set $S$ such that $|S| = \mu|V|$ and $|E_G[S]| \geq C\mu\log(1/\mu)\delta_{\mu|V|}(G)$, where $C$ is an absolute constant. 
\end{theorem}	
The above guarantee follows by immediately combining the results from \cite{CHK11} and \cite{KKT15}\footnote{This observation is due to anonymous reviewers}, although to the best of our knowledge, the above bound is not stated as is in the literature. For completeness, we derive the above bound in Appendix \ref{sec:dks-algo}.  

\section{$\Omega(\mu^{r-1}\log(1/\mu))$-Approximation Algorithm for \dksh}

In this section, we proves Theorem \ref{thm:dksh-alg}. The algorithm for the above theorem is almost identical to the algorithm for {\sc Max}-$k$-{\sc CSP}. We describe it in Figure \ref{fig:dksh-algo} for completeness.

\begin{figure}[ht!]
	\begin{mdframed}
		{\bf Input}: Weighted Hypergraph $H = (V,E,w)$. \\
		{\bf Algorithm}:
		\begin{enumerate}
			\item Construct a graph $G' = (V,E',w')$ as follows. For every hyperedge $e \in E$ and $S \in {e \choose 2}$ introduce an edge $e_S$ with weight
			\[
			w'(e) := \mu^{r - 2}\frac{w(e)}{{r \choose 2}}.
			\]
			\item Run the algorithm from Theorem \ref{thm:dks-alg} on $G$ with bias parameter $\mu$. Let $S \subset V$ denote the solution returned by the algorithm.
			\item {\bf Rounding}. Set $\alpha := 2/r$. For every $i \in V$, do the following independently:
			\begin{itemize}
				\item W.p. $\alpha$, set $x_i \gets \mathbbm{1}_S(i)$.
				\item W.p. $1 - \alpha$, let $x_i \sim \{0,1\}_\mu$.
			\end{itemize}
			\item Output the set $S'$ indicated by $x$.
		\end{enumerate}
	\end{mdframed}
	\caption{Algorithm for \dksh}
	\label{fig:dksh-algo}
\end{figure} 

\begin{proof}[Proof of Theorem \ref{thm:dksh-alg}]

The proof of the Theorem \ref{thm:dksh-alg} will require us to (i) bound the size of the set $S'$ returned by the algorithm and (ii) the bounding the expected weight of hyperedges induced by the set $S'$.

{\bf Bounding $|S'|$} Let $y := \mathbbm{1}_S$, and define $\tilde{\mu}_i = \alpha(y_i) + (1 - \alpha)\mu$. Furthermore, for every $i \in V$, let $Y_i := \mathbbm{1}(i \in S')$. Then, $\Ex\left[\sum_{i \in V} Y_i\right] = \mu|V|$. Since given a fixing of $y$,  $\{Y_i\}_{i \in V}$ are independent $0/1$ random variables, using Hoeffding's inequality we get that $\Pr\left[|S| \geq (1 + \epsilon)\mu|V|\right] \leq \exp(-\epsilon^2\mu|V|)$.

{\bf Bounding $\Ex\left[w\left(E_H[S']\right)\right]$}. We proceed to bound the expected number of induced hyperedges in the set $S'$. For every hyperedge $e$ and $T \in {e \choose 2}$, define the event $\cE_{e,S}$ as the event where (i) the vertices in $T$ were assigned values from $S$ and (ii) the vertices in $e \setminus T$ were assigned values by independently sampling from $\{0,1\}_\mu$. We shall analyze the expected weight contributed by a hyperedge $e$:
\begin{align}
w(e)\Pr_{S'}\left[e \in H[S']\right] 
&\geq w(e)\sum_{T \in {e \choose w}} \Pr\left[\cE_{e,T}\right]\Pr_S\left[e \in H[S]\Big| \cE_{e,T}\right] \\
&= \alpha^{2}(1 - \alpha)^{r-2}w(e)\sum_{T \in {e \choose 2}} \Pr_S\left[e \in H[S]\Big| \cE_{e,T}\right] \\
&= \alpha^{2}(1 - \alpha)^{r-2}w(e)\sum_{T \in {e \choose 2}} \mu^{r-2}\mathbbm{1}_{T \in E_G[S]} \\
&= {r \choose 2}\alpha^{2}(1 - \alpha)^{r-2}\sum_{T \in {e \choose 2}} \frac{\mu^{r-2}w(e)}{{r \choose 2}}\mathbbm{1}_{T \in E_G[S]} \\
&= {r \choose 2}\alpha^{2}(1 - \alpha)^{r-2}\sum_{T \in {e \choose 2}} w'(e) \mathbbm{1}_{T \in E_G[S]} 
\end{align}
Therefore, summing over all hyperedges $e$, the expected weight of induced hyperedges can be bounded as 
\begin{align*}
& {r \choose 2}\alpha^{2}(1 - \alpha)^{r-2}\sum_{e \in E}\sum_{T \in {e \choose 2}} w'(e) \mathbbm{1}_{T \in E_G[S]} \\
& =  {r \choose 2}\alpha^{2}(1 - \alpha)^{r-2}w'\left(G[S]\right)\\
& \overset{1}{\geq} \frac14 \cdot C\mu\log(1/\mu)\delta_{\mu|V|}(G) \\
& \overset{2}{\geq} \frac{C\mu\log(1/\mu)}{4}\cdot \mu^{r-2}\delta^{(r)}_{\mu|V|}(H) \\
& \gtrsim \mu^{r-1}\log(1/\mu){\sf Opt}_\mu(H) 
\end{align*}
Here step $1$ is achieved by setting $\alpha = 2/r$. For step $2$, fix the set $S^* \in V$ which achieves $\delta^{(r)}_{\mu|V|}(H)$. Then, we can bound the weight of edges induced by $S^*$ in the graph $G$ as 
\begin{align*}
w'\left(E_G[S^*]\right) 
&= \sum_{e \in E_H} \sum_{T \in \choose 2} \frac{w(e) \mu^{r-2}}{{r \choose 2}} \mathbbm{1}(e \subset S^*) \\
&\geq \mu^{r-2}\sum_{e \in E_H[S^*]} \sum_{T \in \choose 2} \frac{w(e)}{{r \choose 2}} \mathbbm{1}(e \subset S^*) \\
& = \mu^{r-2} \sum_{e \in E_H[S^*]} w(e) \\
& = \mu^{r-2}w\left(E_H[S^*]\right) = \delta^{(r)}_{\mu|V|}(H).
\end{align*} 
Since there exists a set $S$ of size $\mu|V|$ for which the weight of induced hyperedges in $G$ is at least $\mu^{r - 2}{\sf Opt}_\mu(H)$, the claim follows.

\end{proof}

\paragraph{Acknowledgements.}
We thank the anonymous reviewers for pointing us to the \dks~approximation algorithm using \cite{CHK11} and \cite{KKT15}, as well as the dictatorship test and noise stability bounds from \cite{KS15}.

\bibliographystyle{alpha}
\bibliography{main}

\part{Appendix}

\appendix
\section{Example for Remark \ref{rem:smooth}}				\label{sec:example}

Let $H = (V,E)$ be a uniformly weighted $\dksh_r$ instance with $k = \mu^{5}|V|$. Now we construct a $\ldksh_{r + 1}$ instance $H' = (V',E',w',\mu)$ as follows. The vertex set is $V' = V \cup \{x,y\}$ where $x,y$ are a pair of new vertices. The edge set is defined as $E' = \{e \cup\{y\} | e \in E\}$ i.e, we include every edge $e \in E$ in $E'$ and add the vertex $y$ to every vertex. Finally, we assign the following weights to the vertices:
\[
w(v) = 
\begin{cases}
	\mu^{5}/|V| & \mbox{ if } v \in V, \\
	\mu - \mu^{10} & \mbox{ if } v = y, \\
	1 - \mu - \mu^5 + \mu^{10} & \mbox{ if } v = x.
\end{cases}
\]
Now it is easy to verify that non-trivial solutions of weight at most $\mu$ in $H'$ admit a one-to-one correspondence with solutions of relative weight at most $\mu^5$ in $H$, and in particular $\alpha^*$ approximation algorithm for instances $H'$ yield a $\alpha^*$ approximation algorithm for $H$. Therefore, 
\[
\alpha^* \lesssim_r (\mu^5)^{r - 1} \log(1/\mu)
\]
On the other hand, if Theorem \ref{thm:pred-approx} holds for algorithms guaranteed to return solutions of weight at most $\mu$ (i.e., no multiplicative slack), then,
\[
\alpha^* \gtrsim_r \mu^{r} \log(1/\mu)
\]
which contradicts the above upper bound.

\section{Proof of Theorem \ref{thm:dks-alg}}			\label{sec:dks-algo}

Theorem \ref{thm:dks-alg} follows directly using the following results from \cite{CHK11} and \cite{KKT15}.

\begin{theorem}[Theorem 6~\cite{CHK11} restated]
	There exists a randomized polynomial time reduction from $\dks$~instances on $n$ vertices with $k = \mu n$ to {\sc Max}-$2$-{\sc CSP}-instances $\Psi(V',E',[R],\{\Pi_e\}_{e \in E'})$ on $|V'| = \mu n$ vertices with label set size $R = (1/\mu)$ satisfying 
	\[
	\frac{\delta_{\mu|V|}(G)}{2} \leq {\sf Opt}(\Psi) \leq \delta_{\mu|V|}(G) 
	\] 
	with high probability.
\end{theorem}
\begin{proof}
	For simplicity, we will assume that $1/\mu$ is an integer. The proof is via an elementary reduction from $\dks$ to {\sc Max}-$2$-{\sc CSP} instances. Given a \dks~instance $G = (V,E)$, construct a {\sc Max}-$2$-{\sc CSP} instance $\Psi(V',E',[R])$ with $R = 1/\mu$ as follows. 
	
	{\bf Vertex Set}. Consider a random partition of $V$ into $V_1 \uplus \cdots \uplus  V_\ell$ where $\ell = \mu n$ and each $V_i$ is a $1/\mu$ sized subset. For every $i \in [\ell]$, introduce a vertex $v_i$ in $V'$.
	
	{\bf Constraint Set}. For every edge $i \in [\ell]$, fix a bijection $\pi_i : V_i \to [R]$. Now for every $(i,j) \in {V' \choose 2}$, define the set of accepting labelings $\Pi_{(i,j)}$ as 
	\[
	\Pi_{(i,j)} := \left\{(\pi_i(a),\pi_j(b)) | (a,b) \in E\right\}
	\]
	
	The completeness and soundness of the above reduction are easily established.
	
	{\bf Completeness}. Suppose $S \subset V$ is the optimal $\mu|V|$ sized subset which achieves $\delta_{\mu|V|}(G)$. Let $\hat{S} \subset S$ be the subset of vertices where the partition $\uplus_{i \in [\ell]} V_i$ uniquely intersects with $S$ i.e.,
	\[
	\hat{S} := \Big\{i \in S \Big| i \in V_j \implies V_j \cap S = \{i\} \Big\}
	\]
	Then for any fixed $(i,j) \in E_G[S]$ we have 
	\[
	\Pr\left[i,j \in \hat{S}\right] = \prod_{j = 1}^{2r - 2}\left(\frac{(1 - \mu)n - j - 1}{n - j - 1}\right) \geq \frac12,
	\]
	and therefore $\Ex\left[|E_H[\hat{S}]|\right] \geq 0.5 |E_H[S]|$. Now, one can construct a labeling using $\hat{S}$ which satisfies at least $|E_H[\hat{S}]|$ constraints in $\Psi$ as follows. Let $\kappa : \hat{S} \to [\ell]$ be the mapping which identifies the vertices in $\hat{S}$ with the corresponding partition. Note that $\kappa$ is well defined due to the definition of $\hat{S}$. Now define the labeling $\sigma : V' \to [R]$ as follows:
	\[
	\sigma(i):= 
	\begin{cases}
		\pi_i(a) &\mbox{ if } S \cap V_i = a, \\
		1  & \mbox{ otherwise.}
	\end{cases}
	\]
	Observe that for any edge $(i,j) \in E_H[\hat{S}]$, then the corresponding constraint $(\kappa(i),\kappa(j))$ is satisfied by the labeling $\sigma$. Hence, the number of edges satisfied by $\sigma$ in expectation is at least $|E_H[S]|/2$.
	
	{\bf Soundness}. Fix a labeling $\sigma:V' \to [\ell]$ which achieves ${\sf Opt}(\Psi)$. Then construct $S$ from $\sigma$ as 
	\[
	S:= \left\{\pi^{-1}_i(\sigma(i)) \Big|i \in [\ell]\right\}
	\]
	Clearly, $S$ is a $\mu |V|$ sized subset. Now observe that whenever $\sigma$ satisfies an edge $(i,j) \in E'$, the corresponding pair $(\pi{-1}_i(\sigma(i)),\pi^{-1}_j(\sigma(j)))$ identifies a unique edge in $H$ whose both endpoints are in $S$. Therefore, 
	\[
	|E_H[S]| \geq \left|\{(i,j) \in E' | (\sigma(i),\sigma(j)) \in \Pi_{ij}\}\right| = {\sf Opt}(\Psi),
	\]
	from which the soundness follows.
	
\end{proof}

Combining the above with the following theorem immediately establishes Theorem \ref{thm:dks-alg}.

\begin{theorem}[Theorem 2.1~\cite{KKT15}]
	For every integer $R \geq 2$, there exists an efficient $\Omega(\log R/R)$-approximation algorithm for {\sc Max}-$2$-{\sc CSP} instances over label sets of size $R$. 
\end{theorem}
\section{Proof of Theorem \ref{thm:ks-stab}}				\label{sec:stab}

Here we shall prove Theorem \ref{thm:ks-stab}, which we restate here for convenience.

\stab*

Towards proving the above, we shall need some additional notation from \cite{Mos10} and \cite{KS15}. For any $\rho,\mu_1,\mu_2$, let $\Gamma_\rho(\mu_1,\mu_2)$ be the $(\mu_1,\mu_2)$ biased bilinear Gaussian stability defined as 
\[
\Gamma_\rho(\mu_1,\mu_2) = \Pr_{g_1 \underset{\rho}{\sim} g_2 } \Big[g_1 \leq \Phi^{-1}(\mu_1) \wedge g_2 \leq \Phi^{-1}(\mu_2)\Big].
\]
Using the above, we can define the iterative Gaussian stability $\Gamma_\rho(\mu_1,\ldots,\mu_r)$ as 
\begin{equation}		\label{eqn:gamma-def}
\Gamma_\rho(\mu_1,\ldots,\mu_r) = \Gamma_\rho\left(\mu_1,\Gamma_\rho(\mu_2,\ldots,\mu_r)\right)
\end{equation}
For brevity, we shall use $\Gamma^{(r)}_\rho(\mu)$ to denote $\Gamma_\rho(\mu_1,\ldots,\mu_r)$ where $\mu_i = \mu$ for every $i \in [r]$. 
We shall also need the following analytical tools from \cite{KS15} for bounding the expected value of products of low influence functions.
\begin{theorem}[Theorem 2.10~\cite{KS15} restated]			\label{thm:ks-stab1}
	Let $(\prod_{i \in [r]}\Omega_i,\bgamma) := (\Omega^r,\bgamma)$ be a $r$-ary correlated probability space such that for any $\omega \in \Omega^r$, we have $\alpha \leq \gamma(\omega)$ and $\alpha \leq 1/2$. Furthermore, suppose $\rho(\Omega_1,\ldots,\Omega_r;\bgamma) \leq \rho$. Then for every $\nu \in (0,1)$, there exists $\tau := \tau(\nu,r,\alpha)$ such that the following holds. Suppose $f:\Omega^R \to [0,1]$ is a function satisfying 
	\[
	\max_{i \in [R]} \Inf{i}{f} \leq \tau.
	\]
	Then,
	\[
	\Ex_{(\bomega_1,\ldots,\bomega_r) \sim \bgamma^R}\left[\prod_{i \in [r]} f(\bomega_i)\right]
	\leq \Gamma_{\rho}\left(\Ex\left[f(\bomega_1)\right],\ldots,\Ex\left[f(\bomega_r)\right]\right) + \nu,
	\]
	where $\bgamma^R$ is the $R$-wise product measure corresponding to $\bgamma$. 
\end{theorem}

\begin{lemma}[Lemma 2.4~\cite{KS15} restated]			\label{lem:ks-stab1}
	There exists a constant $C > 0$ for which the following holds. 	Let $r \geq 2$ be a integer and $\mu_1,\ldots,\mu_r \in (0,1)$. Define $\mu^* = \min_{i \in [r]} \mu_i$. Then for any $\rho \leq 1/(2Cr^2\log(r/\mu^*))$ we have 
	\[
	\Gamma_\rho(\mu_1,\ldots,\mu_r) \leq 3\prod_{i \in [r]} \mu_i
	\]
\end{lemma}

Now we are ready to prove Theorem \ref{thm:ks-stab}.

\begin{proof}
	To begin with we claim that $\rho(\Omega_1,\ldots,\Omega_r;\bgamma) = \rho$ in the setting of the theorem. To see this, fix a $j \in [r]$. Then we claim that the leave-one-out correlation with respect to $j$ satisfies 
	\[
	\rho(\prod_{i \neq j} \Omega_i,\Omega_j;\bgamma) \leq \rho.
	\]
	This is because with probability at least $1 - \rho$, the variables $(x_1,\ldots,x_r) \sim (\Omega^r,\bgamma)$ are all independent (from Definition \ref{defn:dist}). Hence, we have 
	\begin{equation}		\label{eqn:corr-bound}
	\rho(\Omega_1,\ldots,\Omega_r;\bgamma) = \max_{j \in [r]} \rho\left(\prod_{i \neq j}\Omega_i,\Omega_j;\bgamma\right) \leq \rho.
	\end{equation}
	Furthermore, let $\kappa = \min_{\omega \in \Omega} \gamma(\omega)$. Then any event in $(\prod_{i \in[r]} \Omega_i,\bgamma)$ happens with probability at least $\alpha = \rho \kappa^r \leq 1/2$. Finally, in the setting of the theorem we have 
	\[
	\max_{i \in [r]} \Inf{i}{f} \leq \tau.
	\]
	where $\tau = \tau(\nu,r,\alpha)$ as in Theorem \ref{thm:ks-stab1}. Therefore, instantiating Theorem \ref{thm:ks-stab1} with $f$ over the probability space $(\Omega^R,\bgamma)$ we get that 
	\begin{equation}			\label{eqn:stab-1}
	\Ex_{(\bomega_1,\ldots,\bomega_R) \sim \bgamma^R} \left[\prod_{i \in [r]} f(\bomega_i)\right]
	\leq \Gamma_\rho\left(\Ex\left[f(\bomega_1) \right],\ldots,\Ex\left[f(\bomega_r)\right]\right) + \nu 			
	= \Gamma^{(r)}_\rho\left(\mu\right) + \nu.
	\end{equation}
	Finally, note that our choice of $\rho$ satisfies $\rho \leq 1/(2C' r^2 \log(1/\mu))$. Hence using Lemma \ref{lem:ks-stab1} we can further upper bound
	\[
	\Gamma^{(r)}_\rho(\mu) \leq 3\mu^r.
	\]
	Plugging in the above bound into \eqref{eqn:stab-1} completes the proof.
\end{proof}

\end{document}